\documentclass[12pt]{article}

\usepackage[titletoc,title]{appendix}
\usepackage[latin1]{inputenc}
\usepackage{amsmath}
\usepackage{amssymb}
\usepackage{booktabs}
\usepackage{graphicx}
\usepackage[margin=1.25in]{geometry}
\usepackage[bottom]{footmisc}
\usepackage{indentfirst}
\usepackage{endnotes}
\usepackage{mathabx}
\usepackage{enumerate}
\usepackage{rotating}
\usepackage{multirow}
\usepackage{booktabs,calc}
\usepackage{threeparttable}
\usepackage{float}
\usepackage{url}
\usepackage{epsfig}
\usepackage{amsmath,amsfonts}
\usepackage{amsthm}
\usepackage{subcaption}
\usepackage[onehalfspacing]{setspace}
\usepackage{pgfplots}
\usetikzlibrary{calc}
\usepackage{amsmath}

\usepackage{epigraph}
\usepackage{etoolbox}
\makeatletter
\patchcmd{\epigraph}{\@epitext{#1}}{\itshape\@epitext{#1}}{}{}
\makeatother

\pgfplotsset{width=7cm,compat=newest}

\newcommand{\suchthat}{\;\ifnum\currentgrouptype=16 \middle\fi|\;}

\newtheorem{assm}{Assumption}

\newtheorem{prop}{Proposition}
\newtheorem{fact}{Fact}
\newtheorem{lemma}{Lemma}
\newtheorem{cor}{Corollary}

\newtheorem{defn}{Definition}

\newenvironment{aprime}[1]
  {\renewcommand{\theassm}{\ref{#1}$'$}%
   \addtocounter{assm}{-1}%
   \begin{assm}}
  {\end{assm}}

\theoremstyle{remark}
\newtheorem{remark}{Remark}

\usepackage{chngcntr}
\usepackage{graphicx}
\usepackage{amsmath}
\usepackage{thmtools}
\usepackage{amssymb}
\usepackage{parskip}
\usepackage{sgame}
\usepackage{color}
\usepackage{tikz}
\usetikzlibrary{trees,calc}

\DeclareMathOperator*{\argmax}{\arg\!\max}
\DeclareMathOperator*{\interior}{int}
\DeclareMathOperator*{\ri}{ri}

\newcommand{\E}{\mathbb{E}}

\DeclareMathOperator{\Co}{Co}
\DeclareMathOperator{\CCo}{CCo}
\newcommand{\eps}{\varepsilon}

\usepackage{natbib}

\newcommand{\pa}{\tilde{p}^1}
\newcommand{\pb}{\tilde{p}^0}
\newcommand{\xa}{\tilde{x}^1}
\newcommand{\xb}{\tilde{x}^0}

\usepackage{xcolor}

\begin{document}

\title{\large{Counterfactual and Welfare Analysis with an Approximate Model}\thanks{We thank Victor Aguiar, Victor Aguirregabiria, Lars Hansen, Nail Kashaev, Lance Lochner, Nirav Mehta, Magne Mogstad, Ismael Mourifi\'{e},  Salvador Navarro, Andres Santos, and participants at the Banff Empirical Microeconomics Workshop, the Western Conference on Counterfactuals with Economic Restrictions, the University of Chicago, and the University of Toronto for helpful comments.} }

\author{ \small{Roy Allen}  \\
    \small{Department of Economics} \\
    \small{University of Western Ontario} \\
    \small{rallen46@uwo.ca}
    \and 
    \small{John Rehbeck} \\
    \small{Department of Economics} \\
    \small{The Ohio State University} \\
    \small{rehbeck.7@osu.edu}
}
\date{\small{ \today }} 

\maketitle

\begin{abstract}
We propose a conceptual framework for counterfactual and welfare analysis for approximate models. Our key assumption is that model approximation error is the same magnitude at new choices as the observed data. Applying the framework to quasilinear utility, we obtain bounds on quantities at new prices using an approximate law of demand. We then bound utility differences between bundles and welfare differences between prices. All bounds are computable as linear programs. We provide detailed analytical results describing how the data map to the bounds including shape restrictions that provide a foundation for plug-in estimation. An application to gasoline demand illustrates the methodology.
\end{abstract}

\newpage

\section{Introduction}
\iffalse
NOTE: Tentative thought is start with something like this: "Models are often viewed as approximations. How do we use them when they are approximations?" Then formalize our assumption 1. Then state this generalizes standard framework. Ats ome point, discuss issues with alternative approaches (e.g. empty sets)... At some point emphasize that we use quasilinear utility as an approximation...Thus, the perspective that ``all models are wrong, but some are useful'' \citep{box1987empirical} is not helpful in this context because using existing tools, \textit{all} models that are not perfectly consistent with data are unable to provide quantitative guidance.
\fi

Models are generally viewed as approximations.  A common intuition in empirical work is that conclusions of a model are robust to ``small'' amounts of approximation error. Unfortunately, this intuition does not apply to many standard frameworks. For example, when performing a revealed preference analysis \citep{varian1982nonparametric}, if a dataset is inconsistent with a model, then counterfactual predictions described by certain inequalities cross. Thus, if there is \emph{any} violation of the model (no matter how ``small''), then the model fails to generate coherent counterfactual or welfare statements.\footnote{Related concerns have been raised in the econometric literature on partial identification \citep{ponomareva2011misspecification,muller2016credibility}.}

Alternatively, one can formally acknowledge approximation error throughout the analysis. Rather than treat approximation error as nonexistent, one can place restrictions on the magnitude of the approximation error. This paper does so for counterfactual and welfare analysis, with the following assumption on this magnitude.

\begin{assm}\label{assm:maint}
When making counterfactual predictions or measuring welfare changes, we assume the approximation error of the model on the counterfactual predictions is the same as the approximation error of the model on the observed dataset.
\end{assm}

This assumption is a natural extension of the standard approach to generate counterfactual predictions that assumes that both the observed data and counterfactual predictions are consistent with a model. We present a framework in which a model can be used even though it is not exactly consistent with observed data. In particular, this paper assumes that the approximation error of the model on  observed and unobserved situations has the same magnitude. We call the counterfactuals that are consistent with Assumption~\ref{assm:maint} \emph{adaptive counterfactuals} because they adapt to approximation error present in the observed dataset. This assumption can be questioned, especially when the counterfactual setting is significantly different than observed data, yet provides a way to conduct counterfactual analysis taking approximation error seriously.

The conceptual framework of this paper is general and can be applied to different settings. In this paper, we formalize how to generate counterfactual predictions and measure welfare changes for the quasilinear utility model using the notion of approximation error from \cite{allen2020satisficing}. In particular, we present bounds on counterfactual quantities at new prices, differences of utility over consumption bundles, and welfare differences involving a price change.  The computation of all bounds is facilitated by linear programming and we demonstrate the methods in an  illustrative empirical example using gasoline demand data from \cite{blundell2012measuring}. 

The quasilinear model is a suitable setting in which to study approximation error because it is implicitly viewed as an approximation.\footnote{A notable exception that studies the approximation error explicitly is \cite{willig1976consumer}.} The most common criticism is that the model does not allow income effects. In addition, the model neglects dynamics, limited consideration, and peer effects, among many other omitted features. We present a framework in which one does not need to pick a single story why the baseline model is imperfect when generating counterfactual predictions on quantities. However, to make welfare comparisons we take a stand on the interpretation of the approximation error: the individual ranks bundles according to a quasilinear utility function but for reasons we do not model explicitly, the choices do not exactly maximize the function. Overall, the framework we propose permits many reasons why the baseline model is wrong, provided the approximation error is the same magnitude in the counterfactual setting.

Despite being viewed as an approximation, the quasilinear model is widely used. Examples include work in insurance choice \citep{einav2010estimating,bundorf2012pricing,tebaldi2018nonparametric} and public health \citep{cohen2010free}. In addition, the quasilinear structure is closely related to a large class of latent utility models (e.g. \cite{mcfadden1981econometric}, \cite{ARidentification}), and so the insights of this paper are directly relevant beyond a setting with just prices and quantities. In particular, many latent utility models used in applied work involve characteristics other than prices that shift the desirability of goods but not the budget constraint.\footnote{Our analysis is also relevant for specifications in which latent utilities depend on a nonlinear function of prices. For example, \cite{berry1995automobile} specifies that the utility of alternative $j$ depends on several observables including a term $\ln(m - p_j)$ where $p_j$ is the price of good $j$ and $m$ is income. This is a quasilinear model in the variable $\tilde{p}_j = \ln(m - p_j)$.}

We now describe the framework in more detail. We begin by studing counterfactual bounds for the demand of goods at new prices. We construct the counterfactual bounds by looking for the maximal and minimal demand for each good in the presence of approximation error. We assume approximation error does not increase for the counterfactual predictions relative to the dataset we have seen. This leads to nontrivial restrictions on demand at new prices as long as the prices are not too low. For each counterfactual price, this procedure gives an interval for lower and upper demands for a given good.

For welfare analysis we view approximation error with a specific interpretation. In particular, we assume an individual cannot perfectly maximize their utility function because they are satisficers in the spirit of \cite{simon1947administrative}. This means an individual chooses quantities that are ``close'' to the maximum utility possible, but not necessarily optimal. The interpretation of satisficing allows us to assign special welfare significance to the latent utility over bundles. This utility over consumption bundles is key for policy decisions involving the allocation of goods. In addition, we study welfare over price changes, which is key for tax policy and other policy thought to affect prices.

Two features limit the ability to measure welfare changes using data. First, the choices we see do not exactly maximize utility, which leads to a ``measurement wedge'' on the underlying utility function over bundles. Second, even if we knew the utility function over bundles exactly, we do not know which approximately-optimizing choices will be made at new prices. This leads to an additional ``prediction wedge'' when bounding utility differences obtained at different \textit{prices}.\footnote{We thus complement the core analysis of \cite{bernheim2009beyond}, which focuses on recovering ordinal information on preferences over \textit{consumption bundles} and does not distinguish between these wedges.} 

Taking into account these wedges, we present bounds on differences in utility over consumption bundles and robust consumer surplus bounds involving price changes. We present computational results for both bounds, as well as analytical results designed to interpret specifically how the data are used to measure welfare changes. These bounds generalize existing work in several directions: first, and most importantly, they are valid with approximation error; second, they apply to finite datasets rather than requiring demand functions; third, the bounds apply to the (approximate) indirect utility at a new price without needing to first bound the quantity at that price. In particular, we show that the bounds on (approximate) indirect utility is a generalization of the standard integral definition of consumer surplus and we establish a close connection between counterfactual quantities and welfare bounds.

In order to understand how the counterfactual/welfare bounds depend on data and prices, we present several shape restrictions. We are not aware of any work in the tradition of \cite{varian1982nonparametric} that discusses shape restrictions of the counterfactual/welfare bounds (viewed as functions of data or counterfactual prices). Studying these shape restrictions is important to understand how data are used in an empirical analysis.

For counterfactual quantities, we find that when there is a single good and the counterfactual price changes, the upper and lower bounds on counterfactual demands are each weakly decreasing in price. Thus, the upper and lower bounds on demand are functions that satisfy the law of demand. For utility differences, we establish monotonicity and continuity properties of the bounds as the quantities being compared change. For a price change, the approximate indirect utility bounds satisfy convexity and monotonicity conditions as prices vary; these are also key shape restrictions for the indirect utility function for quasilinear utility.\footnote{More specifically, convexity holds for the money metric version of the utility function. In general one can only obtain quasiconvexity.} Finally, we establish several convexity properties describing how quantities data map to the bounds. To our knowledge these results are all new even under correct specification for quasilinear utility.\footnote{The closest work appears to be a computational approach to describing bounds for models related to quasilinear utility, without describing detailed shape restrictions \citep{chiong2017counterfactual,tebaldi2018nonparametric,ARidentification}.}

Our analysis also establishes that the bounds satisfy a key continuity property: as the degree of approximation error limits to $0$, our analysis limits to the analysis under correct specification. In fact, we show a stronger property that the counterfactual and approximate indirect utility bounds are \textit{jointly} continuous when viewed as a function of the quantities in the data and the degree of approximation error. This facilities plug-in estimation of the bounds in which we replace true quantities with estimated quantities. This is needed to cover our empirical application in which we apply the framework with data on gasoline purchases used previously in \cite{blundell2012measuring}. The data is a single cross section, and we pre-process the data as in \cite{blundell2012measuring} by kernel smoothing. We conduct a representative agent analysis with quantities (conditional means) estimated from the kernel smoothed data. Like \cite{blundell2012measuring}, we find that for several natural choices of the bandwidth, demand is not downward sloping. Thus, it is inconsistent with the exact quasilinear model. Nonetheless, the minimal degree of approximation error need to describe data is small and welfare bounds are surprisingly narrow for all bandwidths we consider. In contrast, the informativeness of the counterfactual bounds depends on the bandwidth.

The rest of this paper is organized as follows. After reviewing the  related literature, Section~\ref{sec:setting} describes the setup and conceptual framework of approximate counterfactuals for the quasilinear framework. Section~\ref{sec:counterfactuals} studies counterfactuals. Section~\ref{sec:welfare} studies welfare. Section~\ref{sec:shape} presents additional shape restrictions and discusses plug-in estimation. Section~\ref{sec:application} contains the application to gasoline demand. Section~\ref{sec:conclusion} concludes.

\subsection{Literature Review}

This paper is part of a long literature that uses the revealed preference approach to do counterfactual and welfare analysis. The primary model used is the general model of utility maximization subject to a budget constraint, whose empirical content has been characterized in \cite{afriat1967construction}, \cite{diewert1973afriat}, and \cite{varian1982nonparametric}. Recent econometric work considering counterfactual or welfare bounds includes \cite{blundell2003nonparametric}, \cite{blundell2008best}, \cite{blundell2012measuring}, \cite{blundell2014bounding}, \cite{hoderlein2015testing}, \cite{kline2016bounding}, \cite{blundell2017individual}, \cite{cosaert2018nonparametric}, \cite{aguiar2018meas}, \cite{adams2019mutually}, \cite{cherchye2019bounding}, and \cite{kitamura2019nonparametric}. Several proposals have been made to assess the fit of a model using revealed preference tools, including \cite*{afriat1973system}, \cite*{houtman1985determining}, \cite*{varian1990goodness}, and \cite*{echenique2011money}.\footnote{See \cite{ARmeasuring} for additional references and discussion of units.} Other papers outside of the revealed preference literature that discuss fit of an approximate model include \cite{kydland1982time}, \cite{vuong1989likelihood}, and \cite{hansen1997assessing}. The primary way in which we differ from existing work is that we \textit{use} a measure of fit to adjust bounds on counterfactuals and welfare. In addition, relative to the general model with income effects, which has been the focus of the revealed preference literature, we conduct counterfactual analysis fixing prices at a new value without also fixing expenditure.\footnote{Work that uses revealed preference techniques to go beyond measuring the fit of the model includes \cite{varian1990goodness}, \cite{halevy2018parametric}, and \cite{gauthier2019}, which study parameter recoverability. See also \cite{chetty2012bounds} for a related approach.}
 
A growing econometric literature has studied sensitivity analysis and other ways in which a model can be used formally viewing it as an approximation. Examples include \cite{imbens2003sensitivity},
\cite{conley2012plausibly}, \cite{kline2013sensitivity},
\cite{andrews2017measuring}, \cite{manski2018right}, \cite{masten2018identification,masten2019inference}, \cite{armstrong2018sensitivity}, \cite{bonhomme2018minimizing},
\cite{christen2019}, \cite{d2018rationalizing}, \cite{fessler2019use}, \cite{salanie2019fast}, and \cite{andrews2019inference}. See \cite{masten2018salvaging} for additional references and discussion. We differ from this work by focusing on a notion of approximation derived from approximate optimization. We complement the robustness approach \citep{hansen2008robustness} by focusing on the theoretical bounds for welfare/counterfactuals rather than focusing on an optimal decision.

This paper is naturally related to other frameworks for welfare analysis that go beyond the classic revealed preference tradition. Our approach differs substantially from an approach that posits individuals have consistent choices (that can be modeled as solutions to a ``decision utility''), but whose consistent choices do not reveal the ``true utility.'' Subtleties with this approach have been discussed by \cite{bernheim2009beyond} among others. A framework of ``behavioral welfare analysis'' is presented in \cite{bernheim2016good} and \cite{bernheim2018behavioral}, which provide further summaries of what has become a large literature. Broadly, we differ from welfare proposals in behavioral economics by addressing welfare questions given data (only) on prices and quantities, without observing decision frames or other variables thought to alter the choice process. In addition, when conducting welfare analysis concerning price changes, we introduce the ``prediction wedge'' because given prices, we do not know precisely what an individual would choose, even if utility were known \textit{a priori}.

\iffalse
...
Alternatively, individuals may not be rational in the sense that they do not maximize a standard utility function (quasilinear or not). These concerns and many others are all relevant.  In this paper we consider \textit{general} (or ``unstructured'') deviations from the model as in  \cite{hansen2008robustness,hansen2018structured} and \cite*{chetty2012bounds}, as well as \textit{specific} deviations from the model.
\fi

\section{Framework and Setting}\label{sec:setting}
The goal of this paper is to provide a framework where a researcher begins with a baseline model that is taken seriously as an approximation. Since the model is an approximation, the researcher does not expect \emph{all} data to be consistent with the model. Nonetheless, the researcher may want to use the model for counterfactual and welfare analysis. We present an adaptive framework framework for this by operationalizing Assumption~\ref{assm:maint} for the quasilinear utility model. Recall that Assumption~\ref{assm:maint} maintains that a researcher considers  counterfactuals that are ``no worse'' than the observed data. To do this, we enlarge the baseline model to fit the data, which introduces a measurement wedge and prediction wedge. The \emph{measurement wedge} concerns limits on what an analyst can learn due to approximation error for objects defined in the existing dataset (e.g. utility functions). The \emph{prediction wedge} describes limits on what can be said in new settings where the model may not be perfect (e.g. counterfactual quantities). We formalize Assumption~\ref{assm:maint} by making these wedges as small as possible while still fitting the observed data, using the notion of approximation error from \cite{allen2020satisficing}. We describe the framework more below.

We formalize the \textit{baseline model} of quasilinear utility. A consumption bundle $(x,y) \in \mathbb{R}_+^K \times \mathbb{R}$ is evaluated according to $u(x) + y$, where $u:\mathbb{R}_+^K\rightarrow\mathbb{R}$ is a utility function over bundles $x$. The numeraire good is given by $y$ and has a price of one. Given prices $p\in \mathbb{R}_{++}^K$ and income $I \in \mathbb{R}$ decisions in a quasilinear utility model follow  
\begin{align*}
    \max_{x \in \mathbb{R}_{+}^K, y \in \mathbb{R}} &\quad  u(x) + y \quad \iff \quad \max_{x \in \mathbb{R}_{+}^K} \quad  u(x)+I-p\cdot x \\
    \text{s.t.}&\; p\cdot x + y \le I
\end{align*}
where consumption of the numeraire good is allowed to be negative for unobserved borrowing.\footnote{Allowing negative expenditure also avoids boundary issues for chosen consumption bundles.} We study the quasilinear utility model since it is regularly used in applications described in the Introduction, including adaptations to handle non-price characteristics .\footnote{\cite{ARidentification} show many applications including the additive random utility model \citep{mcfadden1981econometric} are quasilinear models with utility indices playing the role of prices.} In addition, it has a tractable notion of welfare in terms of units of the numeraire.
 
Now we present an \textit{enlargement} of the baseline model that relaxes the assumption of exact maximization to a notion of approximate optimization.  Because we focus on the empirical analysis, we define the enlargement in terms of a finite datasets of the form $\{ (x^t,p^t) \}_{t=1}^T$. There are $T$ observations, quantities are weakly positive $x^t \in \mathbb{R}_{+}^K$, and prices are strictly positive $p^t \in \mathbb{R}_{++}^K$. Importantly, quantities can be discrete or continuous, and $0$ quantities are permitted in this framework.

\begin{defn} A dataset $\{ (x^t,p^t) \}_{t=1}^T$ is $\varepsilon$-rationalized by quasilinear utility for $\varepsilon \ge 0$ if there exists a utility function $u: \mathbb{R}_+^K \rightarrow \mathbb{R}$ such that for all $t \in \{1,\ldots,T\}$ and for all $x \in \mathbb{R}_+^K$, the following inequality holds:
\[ u(x^t)-p^t\cdot x^t  \ge  u(x) - p^t \cdot x - \eps .\]
When $\varepsilon$ equals zero, we say the dataset is quasilinear rationalized.
\end{defn}
The value $\eps$ is in the same units as the price of the numeraire good, e.g. dollars per time period. When $\eps > 0$, the observed bundles are within $\eps$ dollars of the maximum utility possible at a given price. One interpretation is that $\eps$ captures ``unstructured'' deviations from the quasilinear utility  model (cf.  \cite{chetty2012bounds}, \cite{hansen2018structured}), without a single interpretation of the nature of the deviations. Instead, the \textit{magnitude} of the deviations is controlled. This interpretation can be pursued when making counterfactual predictions. In contrast, to make welfare predictions an interpretation of the model is crucial. When discussing welfare, we follow \cite{allen2020satisficing} and interpret the value $\eps$ as a level of satisficing in the spirit of \cite{simon1947administrative}. In this case, a higher value of $\eps$ means there is a larger set of consumption bundles that are ``good enough" to be chosen. 

When using the model for counterfactual or welfare analysis, a \emph{measurement wedge} and a \emph{prediction wedge} arise. These concepts will become more clear when we turn to specific analysis below, but we first provide an overview. When observed data is not exactly consistent with a quasilinear utility model, the econometrician knows there is no utility function that rationalizes the entire dataset, so at some observation the quantity is not optimal. Thus, if a researcher still wants to use the quasilinear model even when data is inconsistent with the baseline quasilinear model, then there is a  \emph{measurement wedge} when trying to recover information about candidate utility functions and indirect utility. Second, even after the econometrician has a set of candidate utility functions that match the original dataset, the econometrician cannot know that counterfactual choices will exactly maximize a candidate utility function. Thus, there is a \emph{prediction wedge} when forecasting even after recovering information on the utility function. 

We now discuss how to formalize Assumption~\ref{assm:maint} for the approximate quasilinear utility model in relation to the measurement wedge and prediction wedge. Let $\eps_M$ denote the measurement wedge and $\eps_P$ denote the prediction wedge. In principle, the two wedges may not be the same, but Assumption~\ref{assm:maint} allows us to treat these as equal.\footnote{See Appendix~\ref{sec:alternative} for additional discussion on this case and a more formal treatment of the measurement and prediction wedges.} To further formalize Assumption~\ref{assm:maint}, we introduce  $\eps^*$ as the smallest value of $\eps$ such that the dataset is $\eps^*$-rationalized by quasilinear utility. We also refer to $\eps^*$ as the level of approximation error or approximation error of the quasilinear utility model for the observed dataset.
\begin{prop}[\cite{allen2020satisficing}]\label{prop:minqeps}
Let $\varepsilon^* \ge 0$ be the smallest value such that for all $\varepsilon \ge \varepsilon^*$ the dataset $\{(x^t,p^t)\}_{t=1}^T$ is $\varepsilon$-rationalized by quasilinear utility. The value $\eps^*$ exists and is obtained by a linear program. 
\end{prop}
We note that $\eps^*$ is a function of the dataset to a number, so for a dataset $D=\{(x^t,p^t)\}_{t=1}^T$ we can write $\eps^*(D)$. When we discuss only a single dataset, we typically drop dependence on $D$. The value $\eps^*$ will be used in our framework to place restrictions on the magnitude of the measurement and prediction wedges. Setting $\eps_M \geq \eps^*$ formalizes that the measurement wedge is large enough to explain the data we have seen. Similarly, setting $\eps_P \geq \eps^*$ formalizes that the model is no better at predicting in new settings than the data we have seen. We make these bounds as tight as possible, and formalize Assumption~\ref{assm:maint} for this setting as follows.

\begin{aprime}{assm:maint}\label{assm:prime}
When performing counterfactual analysis, the measurement wedge, prediction wedge, and approximation error of the model are equal,
\[
\eps_M = \eps_P = \eps^*.
\]
\end{aprime}

This is a direct generalization of the standard approach to counterfactual and welfare analysis, which sets $\eps_M = \eps_P = 0$. The conceptual framework of the standard approach only applies to models that perfectly fit the data, which translates to $\eps^* = 0$ here. We later develop a framework for counterfactual and welfare analysis when the measurement and prediction wedges are equal, $\eps_M = \eps_P$. For notational convenience, we will let $\eps$ denote the common value. Assumption~\ref{assm:prime} is the special case where $\eps = \eps^*$.

\section{Counterfactuals} \label{sec:counterfactuals}

For the quasilinear utility model with approximation error that does not exceed $\eps$, we consider sharp counterfactual bounds. More formally, we consider when the measurement and prediction wedges are both equal to a single value $\eps$. We impose $\eps = \eps^*$ to implement Assumption~\ref{assm:prime} and describe some properties on the counterfactual bounds. In Section~\ref{sec:graph}, we provide graphical intuition for the bounds. In Section~\ref{sec:boundquant}, we describe how to compute  bounds on quantities fixing a price. In Section~\ref{sec:exp} we describe additional restrictions that can be imposed to tighten the bounds, such as \textit{a priori} bounds on expenditure at a new price.

To that end, for notational convenience, let $D =\{(x^t, p^t ) \}_{t = 1}^T$ denote the observed dataset. Suppose we have a candidate quantity-price tuple $(\tilde{x},\tilde{p})$. We can add this to the original dataset to form an augmented dataset $D \cup  (\tilde{x}, \tilde{p})$. We consider candidates such that the approximation error of the augmented dataset is bounded by $\eps$. In particular, the set of \emph{consistent demands and prices} for the level of approximation error $\eps$ is given by
\[
C(D, \eps ) = \left\{ (\tilde{x},\tilde{p}) \in \mathbb{R}_{+}^K \times \mathbb{R}_{++}^K \mid \eps^* (D \cup  (\tilde{x},\tilde{p}) ) \leq  \eps \right\}.
\]
To check whether a candidate tuple $(\tilde{x},\tilde{p})$ is in $C(D, \eps )$, one can calculate $\eps^*$ for the augmented dataset using Proposition~\ref{prop:minqeps}. If this measure of approximation error for the augmented dataset is below $\eps$, then the candidate tuple is in the set $C(D, \eps )$.

Our framework imposes Assumption~\ref{assm:prime} to generate counterfactual predictions assuming the level of approximation error does not get worse. This amounts to setting $\eps$ equal to the approximation error of the observed dataset. In particular, we focus on the \emph{adaptive counterfactual set}
\begin{align*}
AC(D) & = \left\{ (\tilde{x},\tilde{p}) \in \mathbb{R}_{+}^K \times \mathbb{R}_{++}^K \mid \eps^* (D \cup  (\tilde{x},\tilde{p}) ) \leq \eps^* (D)  \right\} \\
& = C(D, \eps^*( D) ).
\end{align*}

We collect some facts about $AC(\cdot)$ and $C(\cdot)$.
\begin{fact}[Constant Approximation Error] \label{f:const}
For any $(\tilde{x},\tilde{p}) \in AC(D)$, we have
\[
\eps^*(D \cup (\tilde{x},\tilde{p})) = \eps^*(D).
\]
\end{fact}
Thus, when a candidate observation in $AC(D)$ is added to $D$, the measure of approximation stays the same. This follows from the construction of $AC(D)$. This equality does not hold for all measures of model approximation error. For example, if we had chosen to take $\eps^*$ divided by the number of observations $T$ as the measure of approximation error, then Fact~\ref{f:const} would not hold in general since the measure of approximation error for the augmented dataset would divide by $T + 1$. 

\iffalse
\begin{fact}
If $\eps > \eps^*$, there is some $(\tilde{x},\tilde{p}) \in C(D,\eps)$ such that
\[
\eps^*(D \cup (\tilde{x}, \tilde{p})) = \eps > \eps^*(D).
\]
\end{fact}

Thus, constant approximation error fails when we conduct counterfactuals with $\eps > \eps^*$.
\fi

\begin{fact}[Monotonicity]
If $\eps < \eps'$, then $C(D,\eps) \subseteq C(D,\eps')$.

\end{fact}
Higher values of $\eps$ correspond to less informative counterfactual predictions. This follows from the fact that if a dataset is $\eps$-rationalized by quasilinear utility, then it is also $\eps'$-rationalized for $\eps < \eps'$.

\begin{fact}[Nonemptiness]\label{fact:nonempty}
$C(D,\eps)$ is nonempty if and only if $\eps \geq \eps^*(D)$. 
\end{fact}
This states that the observed data places a lower bound on the minimal amount of approximation error needed to conduct counterfactual anlaysis. If $\eps \geq \eps^*$, nonemptiness of $C(D,\eps)$ is guaranteed by considering $\tilde{p}$ sufficiently high along each dimension and $\tilde{x} = 0$. Alternatively, for the dataset $D$, when $\eps < \eps^*$ even observations within the dataset cannot be $\eps$-quasilinear rationalized.

\begin{fact}[Minimality] \label{f:min}
$AC(D)$ is obtained from the smallest $\eps$ such that $C(D,\eps)$ is nonempty.
\end{fact}
This formalizes that setting $\eps = \eps^*$ for counterfactual values obtains the sharpest restrictions under Assumption~\ref{assm:prime} subject to the constraint that counterfactuals are nonempty. This follows from the previous facts.  To perform a sensitivity analysis, one could examine any $\eps > \eps^*$ and use $C(D,\eps)$ as the counterfactual set. Our framework allows this yet focuses on $\eps = \eps^*$.

\subsection{Approximate Law of Demand} \label{sec:graph}
To gain intuition on the ``shape" of the counterfactual sets $C(D,\eps)$ and $AC(D)$, we present a graphical description of the restrictions on counterfactuals. For exposition we focus on \textit{some} of the restrictions rather than all of them. First we describe a restriction that must hold for a dataset to be $\eps$-rationalized. At price $p^r$ we must have
\[
u(x^r) - p^r \cdot x^r \geq u(x^s) - p^r \cdot x^s - \eps
\]
for some unknown function $u$. This states that $x^s$ cannot be much better than $x^r$ at price $p^r$. Flipping the role of observations $r$ and $s$ and basic algebra yields
\begin{equation}
\frac{1}{2} (p^s-p^r)\cdot(x^s-x^r) \le \eps.
\end{equation}
This is a multivariate approximate law of demand. The usual multivariate law of demand obtains when $\eps = 0$. For a given value $\eps\ge 0$, this also places restrictions on counterfactual demand $\tilde{x}$ at prices $\tilde{p}$ since for any $r \in \{1, \ldots, T \}$, a potential counterfactual tuple must satisfy
\begin{equation} \label{eq:clawofdemand}
\frac{1}{2} (\tilde{p}-p^r)\cdot(\tilde{x}-x^r) \le \varepsilon.
\end{equation}
When we apply Assumption~\ref{assm:prime}, we evaluate counterfacturals at  $\eps = \eps^*$. This inequality places a restriction on candidate quantity-price tuples when compared with any observation in the dataset. In the one dimensional case ($K = 1$), this states that if the price increases from $p^r$ to $\tilde{p}$, then demand cannot increase by too much. The bound on the increase in quantities is inversely related to the magnitude of the price increase. That is, when $\tilde{p}-p^r>0$ we have $\tilde{x} \le x^r + \frac{2 \eps}{\tilde{p}-p^r} $.

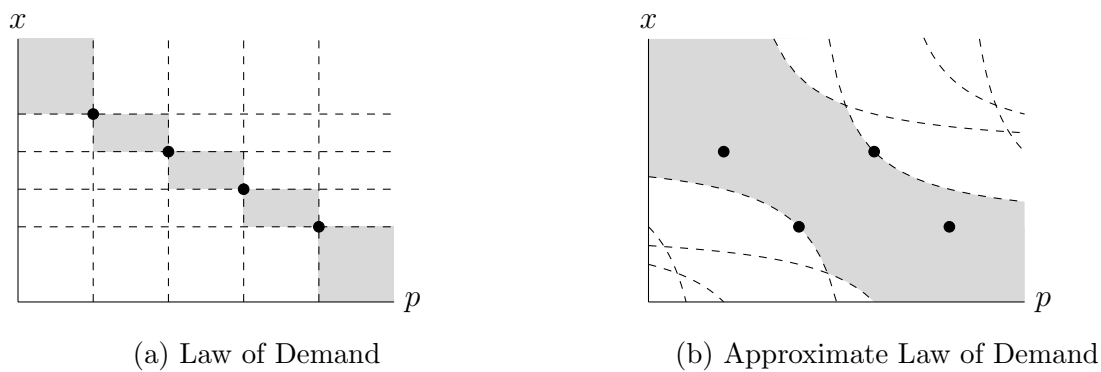
\begin{figure}[H]
  \centering
    \begin{subfigure}[h]{0.45\textwidth}
        \begin{tikzpicture}

    \fill[fill=gray!30] (0,3.5) -- (1,3.5) -- (1,2.5) -- (0,2.5);
    \fill[fill=gray!30] (1,2.5) -- (2,2.5) -- (2,2) -- (1,2);
    \fill[fill=gray!30] (2,2) -- (3,2) -- (3,1.5) -- (2,1.5);
    \fill[fill=gray!30] (3,1.5) -- (4,1.5) -- (4,1) -- (3,1);
    \fill[fill=gray!30] (4,1) -- (5,1) -- (5,0) -- (4,0);    
 
    \draw (0,0) -- (5,0) node[right] {$p$};
    \draw (0,0) -- (0,3.5) node[above] {$x$};
    
    \draw[dashed] (1,0) -- (1,3.5);
    \draw[dashed] (0,2.5) -- (5,2.5);
    \draw[dashed] (2,0) -- (2,3.5);
    \draw[dashed] (0,2) -- (5,2);
    \draw[dashed] (3,0) -- (3,3.5);
    \draw[dashed] (0,1.5) -- (5,1.5);
    \draw[dashed] (4,0) -- (4,3.5);
    \draw[dashed] (0,1) -- (5,1);

    \filldraw 
    (1,2.5) circle (2pt)  ;
    \filldraw 
    (2,2) circle (2pt)  ;
    \filldraw  
    (3,1.5) circle (2pt) ;
    \filldraw 
    (4,1) circle (2pt)  ;
    \end{tikzpicture}
    
    \caption{Law of Demand}
    \end{subfigure}
    \hfill
    \begin{subfigure}[h]{0.45\textwidth}

    \begin{tikzpicture}

    \fill[fill=gray!30, domain=0:2.4] plot (\x,{2 + 1/(\x - 3) })
    -- (2.5,1.5)
    -- (0,3.5) ;
    \fill[fill=gray!30, domain=2.4:3] plot (\x,{1 + 1/(\x - 4) })
    -- (5,0)
    -- (5,1)
    -- (2,2);    
    \fill[fill=gray!30, domain=1.666:2.615] plot (\x,{2 + 1/(\x - 1) })
    -- (2.65,1)
    -- (0,3.5) 
    -- (2.4,3.5);    
    \fill[fill=gray!30, domain=2.615:5] plot (\x,{1 + 1/(\x - 2) })
    -- (5,0)
    -- (2.5,1.5);  
  
    \draw (0,0) -- (5,0) node[right] {$p$};
    \draw (0,0) -- (0,3.5) node[above] {$x$};
    \draw[dashed, domain=0:.5] plot (\x,{2 + 1/(\x - 1) });
    \draw[dashed, domain=1.666:5] plot (\x,{2 + 1/(\x - 1) });
    \draw[dashed, domain=0:1] plot (\x,{1 + 1/(\x - 2) });
    \draw[dashed, domain=2.4:5] plot (\x,{1 + 1/(\x - 2) });
    \draw[dashed, domain=0:2.5] plot (\x,{2 + 1/(\x - 3) });
    \draw[dashed, domain=3.66:5] plot (\x,{2 + 1/(\x - 3) });
    \draw[dashed, domain=0:3] plot (\x,{1 + 1/(\x - 4) });
    \draw[dashed, domain=4.4:5] plot (\x,{1 + 1/(\x - 4) });

    \filldraw 
    (1,2) circle (2pt)  ;
    \filldraw 
    (2,1) circle (2pt)  ;
    \filldraw  
    (3,2) circle (2pt) ;
    \filldraw 
    (4,1) circle (2pt)  ;
    \end{tikzpicture}
    
    \caption{Approximate Law of Demand}
    \end{subfigure}
    \caption{Restrictions of Law of Demand}\label{fig:countlaw}
\end{figure}

We illustrate these bounds in two example datasets displayed in Figure~\ref{fig:countlaw}. Each dataset has four observations respresented as black dots. The gray area denotes the set of quantity-price tuples $(\tilde{x},\tilde{p})$ that are consistent with the existing dataset with minimal level of approximation error  $\eps = \eps^*$. In panel (a), the observed dataset is exactly consistent with quasilinear utility and the counterfactual set has $\eps = 0$. In this case, the constructed bounds have the property that when price increases quantity cannot increase. This leads to the ``rectangular'' bounds in panel (a). Note that for low values of prices, quantity has a lower bound but not an upper bound. Similarly, when prices are higher than any observed data the lower bound on counterfactual demand is zero. 

In panel (b), the dataset is not consistent with quasilinear utility, because there is an instance in which price goes up and quantity goes up. Here we graphically obtain the counterfactual restrictions using the approximate law of demand constructed from Equation~\ref{eq:clawofdemand}, setting $\eps = \eps^*$.\footnote{There are additional restrictions beyond Equation~\ref{eq:clawofdemand}; here we provide a graphical illustration but the general framework uses additional inequalities.} This approach leads to ``hyperbolic'' bounds, in contrast with the rectangular bounds in panel (a). The fact that $\eps = \eps^*$ is the minimal approximation error needed to rationalize the data is demonstrated on the graph by two points touching dashed hyperbolas.

\subsection{Bounding Quantities} \label{sec:boundquant}
The sets $C(\cdot)$ and $AC(\cdot)$ completely describe  counterfactuals. An analyst may not be interested in the entire set of counterfactual quantity-price tuples, but rather certain features of it. For example, an analyst may only be interested in quantities at a fixed counterfactual price $\tilde{p} \in \mathbb{R}^K_{++}$ allowing approximation error $\eps$. This set may be written
\[
X(\tilde{p},D,\eps) = \left\{\tilde{x} \in \mathbb{R}_{+}^K \mid (\tilde{x},\tilde{p}) \in C(D,\eps) \right\}.
\]

Our first question is when this set is nonempty, i.e. when can we conduct counterfactual analysis.
\begin{prop} \label{prop:nonempty}
For a dataset $D$ and counterfactual price $\tilde{p}$, the set $X(\tilde{p},D,\eps)$ is nonempty if and only if $\eps \geq \eps^*$.  Moreover, when $\eps \geq \eps^*$ there is a concave, strictly increasing, continuous utility function $u : \mathbb{R}^K_+ \rightarrow \mathbb{R}$ that $\eps$-rationalizes the dataset and has an exact maximizer for each $p \in \mathbb{R}^K_{++}$.\footnote{By strictly increasing we mean the usual definition, i.e. if each component of $x$ is weakly greater than each component of $z$, then $u(x) \geq u(z)$, and if in addition some component of $x$ is strictly greater than the corresponding component of $z$, then $u(x) > u(z)$.}
\end{prop}
This shows that by allowing enough approximation error, we can find counterfactual quantities for any price. This is stronger than Fact~\ref{fact:nonempty} because it gives nonemptiness of the counterfactual quantity set for \textit{any} price. Existing work has studied when observed datasets can be rationalized by quasilinear utility (\cite{brown} for $\eps = 0$ or with certain random shocks) or an enlargement (\cite{allen2020satisficing} for $\eps \geq 0$). This strengthens those results by showing when we can conduct counterfactual analsyis of quantities at a new price. The question is nontrivial because the utility functions constructed in \cite{brown} and \cite{allen2020satisficing} to explain data have no approximate maximizer for low prices because the indirect utility is infinite.\footnote{See the proof of Proposition~\ref{prop:shapev}. See also \cite{aguiar2020rationalization} for recent work concerning emptiness of counterfactual sets when using the weak axiom of revealed preference.}

We now discuss additional properties of $X(\tilde{p},D,\eps)$.

\begin{prop}\label{prop:Xdetail}
For a dataset $D=\{ (x^t, p^t) \}_{t = 1}^T$, let $\eps \geq \eps^*$. The set $X(\tilde{p},D,\eps)$ is a closed, convex polyhedron. In particular,  $\tilde{x} \in X(\tilde{p},D,\eps)$ if and only if the inequalities  
\begin{equation} \label{eq:xhalfspace}
(\tilde{p} - p^{t_M})\cdot \tilde{x} \leq (M+1) \eps + \tilde{p} \cdot x^{t_1}  - p^{t_M} \cdot x^{t_M} - \sum_{m = 1}^{M-1} p^{t_m} \cdot (x^{t_m} - x^{t_{m+1}})
\end{equation}
hold for all finite sequences $\{ t_m \}_{m = 1}^M$ without cycles where $t_m \in \{1, \ldots, T\}$ and $M \geq 1$.
\end{prop}

When $M = 1$, the inequalities in Equation~\ref{eq:xhalfspace} yield the approximate law of demand described in Equation~\ref{eq:clawofdemand}. In this case we compare an observation in the dataset with a conjectured counterfactual tuple $(\tilde{x},\tilde{p})$, which leads to two instances of $\eps$ in Equation~\ref{eq:xhalfspace}, just like the approximate law of demand. Proposition~\ref{prop:Xdetail} shows there are other restrictions imposed on counterfactuals beyond the law of demand by considering more than one observation at a time ($M \geq 2$). These restrictions arise by adding up additional sequences of inequalities similar to constructing the approximate law of demand. By summing up appropriate sequences, the unknown utility function is removed so restrictions on counterfactual quantities is given only using observable data. We later use similar information to generate bounds on certain welfare objects.

\iffalse
To better understand the counterfactual quantities set, note that from Proposition~\ref{prop:epsqrat}(iii), $X(\tilde{p},D,\eps)$ is characterized by the inequality
\begin{equation} \label{eq:xhalfspace}
(\tilde{p} - p^{t_M})\cdot \tilde{x} \leq M \eps + \tilde{p} \cdot x^{t_2}  - p^{t_M} \cdot x^{t_M} - \sum_{m = 2}^{M-1} p^{t_m} \cdot (x^{t_m} - x^{t_{m+1}})
\end{equation}
holding for all sequences that satisfy the restrictions, together with non-negativity constraints. More formally, we apply Proposition~\ref{prop:epsqrat}(iii) and replace the in-sample observation $t_1$ in a cycle with the counterfactual tuple $(\tilde{x},\tilde{p})$. For fixed $\tilde{p}$, each inequality restricts $\tilde{x}$ to lie in a half-space, and there are finitely many inequalities. From these facts we obtain the following result.
\fi 

Suppose now that we are only interested in bounding the quantity of the $k$-th good at a price $\tilde{p}$, allowing up to $\eps$ approximation error. These bounds are extrema of $X(\tilde{p},D,\eps)$ along the $k$-the dimension. That is, they they are the extreme points of the set
\begin{equation*}
X_k(\tilde{p},D,\eps) = \{ x_k \in \mathbb{R}_+ \mid  \text{There is some } \tilde{x} \in X(\tilde{p},D,\eps) \text{ with } x_k = \tilde{x}_k \}.
\end{equation*}

The following proposition discusses the bounds for the $k$-th good. In particular, when the bounds exist they can be computed by a linear program and the bounds satisfy monotonicity properties with respect to the approximation error $\eps$. 
\begin{prop}\label{prop:qcounterfactuals}
For a dataset $D$, let $\eps \geq \eps^*$. The bounds
\begin{align*}
\overline{x}_k(\tilde{p},\eps) & = \sup_{x_k \in X_k(\tilde{p},D,\eps)} x_k \\
\underline{x}_k(\tilde{p},\eps) & = \inf_{x_k \in X_k(\tilde{p},D,\eps) } x_k
\end{align*}
can each be computed as a linear program whenever they are finite.

Under Assumption~\ref{assm:prime} ($\eps = \eps^*$), these bounds cannot be improved.
\end{prop}
The details on the linear program to compute bounds are found in Proposition~\ref{a:qcounterfactuals} of Appendix~\ref{a:mainproofs}. Recall $X_k(\tilde{p},D,\eps)$ is convex from Proposition~\ref{prop:Xdetail}. Thus, any quantity between $\underline{x}_k(\tilde{p},\eps)$ and $\overline{x}_k(\tilde{p},\eps)$ is a candidate counterfactual quantity for good $k$.

Next we elaborate on when these bounds are finite. We show the lower bound is always finite but the upper bound is finite only when prices are sufficiently high. To formalize this define the upper comprehensive convex hull of a finite set $\{ z^{\ell} \}_{\ell = 1}^L$ as
\begin{align*}
\CCo(\{ z^{\ell} \}_{\ell = 1}^L) = \Bigg\{ z \in \mathbb{R}^K & \mid z \geq \sum_{\ell = 1}^L \alpha_{\ell} z^{\ell} \text{ for some nonnegative } \alpha_1, \ldots, \alpha_L \\
& \text{ such that } \sum_{\ell = 1}^{L}\alpha_{\ell} = 1 \Bigg\}.
\end{align*}
The inequality in the definition here is componentwise. In addition, let $\interior A$ denote the interior of a set $A$.

\begin{prop}\label{prop:Xfinite}
For a dataset $\{ (x^t, p^t) \}_{t = 1}^T$, let $\eps \geq \eps^*$. The upper bound $\overline{x}_k(\tilde{p},\eps)$ is finite if and only if $\tilde{p} \in \interior \CCo ( \{p^t\}_{t=1}^T )$. The lower bound of $\underline{x}_k(\tilde{p},\eps)$ is always finite. The upper bound  $\overline{x}_k(\tilde{p},\eps)$ is weakly increasing in $\eps$ and the lower bound $\underline{x}_k(\tilde{p},\eps)$ is weakly decreasing in $\eps.$
\end{prop}

Finally, we show that for one good $(K = 1)$, the bounds on demand are downward sloping in own-price.
\begin{prop}[Univariate Monotonicity]\label{prop:singleprice}
For a dataset $D$, let $\eps \geq \eps^*$ and suppose there is a single good ($K = 1$). For any pair of prices $\pa, \pb \in \mathbb{R}_{++}$, it follows that
\[
(\overline{x}(\pa,\eps) - \overline{x}(\pb,\eps))(\pa - \pb) \leq 0
\]
and
\[
(\underline{x}(\pa,\eps) - \underline{x}(\pb,\eps))(\pa - \pb) \leq 0.
\]
\end{prop}

When $\eps = 0$, the dataset satisfies the exact law of demand. When $\eps^*$ (and $K = 1$), there is some pair of observations $r,s \in \{1,\ldots,T\}$ that violates the law of demand so
\[
(x^r - x^s) (p^r - p^s) > 0.
\]
Proposition~\ref{prop:singleprice} shows that while such violations can occur in the data, the bounds themselves satisfy the exact law of demand.

\begin{remark}[Sensitivity Analysis]
If an analyst is unsure what is a sensible choice of $\eps$ (other than the requirement $\eps \geq \eps^*$), then it is possible to perform sensitivity analysis of $\overline{x}_k(\tilde{p},\eps)$ and $\underline{x}_k(\tilde{p},\eps)$ as $\eps$ varies. A specific question is the largest amount of approximation error in which one can still bound the quantity of the $k$th good by a pre-specified value, e.g.
\[
\sup \{ \eps \geq \eps^* \mid \overline{x}_k(\tilde{p}, \eps) \leq \overline{q}_k \}.
\]
This bound is related the analysis of breakdown frontiers of \cite{masten2019inference}, which involve the weakest assumptions under which one can reach a conclusion. Here, weakest assumption translates to most approximation error.
\end{remark}

\begin{remark}[Other Bounds]
It is straightforward to generalize Proposition~\ref{prop:qcounterfactuals} to bound certain linear combinations of the candidate demand vector $\tilde{x}$. Bounds on such linear combinations may be computed as the value of a linear programming problem. One interesting linear combination is $\tilde{p} \cdot \tilde{x}$, which is the expenditure on the $K$ goods. Sharp bounds on general functionals $f(\tilde{x})$ can also be described as the value of a constrained optimization problem. For example, an upper bound is given by
\[
\sup_{\tilde{x} \in X(\tilde{p},D,\eps)} f(\tilde{x}).
\]
Recall that Proposition~\ref{prop:Xdetail} states this constraint set is a closed convex polyhedron. This can facilitate computation though we do not formally study computation for general $f$.
\end{remark}

\subsection{Expenditure Bounds} \label{sec:exp}
Additional assumptions can tighten the bounds on quantities in Proposition~\ref{prop:qcounterfactuals}. For example, one can assume  that expenditure is the same at the counterfactual value as the last period of data, so $\tilde{p} \cdot \tilde{x} = p^T \cdot x^T$. Alternatively, one could place bounds on the expenditures so that $\underline{m} \le \tilde{p} \cdot \tilde{x} \le \overline{m}$. One may also impose \textit{a priori} bounds on the quantities of other goods. These bounds can considerably shrink the set of counterfactual bounds, especially when there are multiple goods. In addition, computation with these additional restrictions is not challenging because these are inequality constraints that can be appended to the original linear program. When adding these additional constraints, however, it is possible that the counterfactual set can be empty.

We emphasize that in general, such expenditure bounds are not needed to deliver nontrivial counterfactual bounds. It is helpful to contrast our approach with the general model of utility maximization subject to a budget constraint, with preferences that need not be quasilinear. In the general model, even under correct specification the sharp bounds on quantities of each good at a given price are the trivial bounds $[0,\infty)$ unless the analyst places \textit{a priori} bounds on expenditure at the new price.\footnote{The bounds $[0,\infty)$ are for bounding one good at a time (similar to $\underline{x}_k$ and $\overline{x}_k$ above). There are nontrivial restrictions on the entire demand tuple.} This is because the general model does not rule out expenditure of $0$ or arbitrarily high values at counterfactuals when we only fix prices.\footnote{The results in \cite{deb2018revealed} can be used to show nontrivial bounds are possible in the general model when income is always the same value (inside and outside the dataset) and there is an unobserved good whose price is fixed.}

\section{Welfare} \label{sec:welfare}
To study welfare, we must take a stand on the interpretation of approximation error. For this section, we follow \cite{allen2020satisficing} and treat the approximation error as arising from satisficing in the spirit of \cite{simon1947administrative}. In particular, an individual has a utility function that describes the ranking over goods, but satisfices by choosing bundles that are ``good enough.'' 

We now discuss how satisficing relates to the measurement and prediction wedge. When trying to learn about utility from data, a measurement wedge arises since observed choices may not be optimal. When trying to predict welfare for a price change, the prediction wedge occurs since we only know the region of bundles that are ``good enough." Assumption~\ref{assm:prime} means that the measurement and prediction wedge are the same size as the smallest amount of satisficing needed to describe the data. We note that one can also apply the satisficing interpretation to counterfactual quantities, but it is not necessary. For this reason we did not distinguish between these wedges in Section~\ref{sec:counterfactuals}. Appendix~\ref{sec:alternative} provides additional discussion.

Since we are studying quasilinear utility there are two natural welfare objects. We look at differences in utility over consumption bundles and differences in (approximate) indirect utility over prices. An important asymmetry arises because learning about differences in utility only involves the measurement wedge because it does not involve choices in new situations. In contrast, differences in (approximate) indirect utility over prices involves both the measurement wedge and prediction wedge because one must consider choices in new settings. We elaborate more below.

\subsection{Recoverability of Utility} \label{sec:utilitybounds}
Our first goal is to learn about the unknown utility function over consumption bundles using data. This is helpful when considering policies involving the direct distribution of goods.

In general, there is a collection utility functions that can $\eps$-rationalize a dataset $\{ (x^t, p^t) \}_{t = 1}^T$. We study bounds on utility differences between consumption bundles. Specifically, given two consumption bundles $\xa, \xb \in \mathbb{R}^K_+$ we consider the upper and lower bounds
\begin{align*}
    \overline{u}(\xa, \xb, \eps) & = \sup_{ \left\{u  \mid u \text{ } \eps-\text{rationalizes } \{ (x^t, p^t ) \}_{t = 1}^T \right\}} \left\{ u(\xa) - u(\xb) \right\} \\
    \underline{u}(\xa, \xb, \eps) & = \inf_{ \left\{u \mid u \text{ } \eps-\text{rationalizes } \{ (x^t, p^t ) \}_{t = 1}^T \right\}} \left\{ u(\xa) - u(\xb) \right\}.
\end{align*}
Here, we consider all possible utility functions $u : \mathbb{R}^K_+ \rightarrow \mathbb{R}$ without additional restrictions such as monotonicity or concavity. A utility function $u$ is said to $\eps$-rationalize the dataset $\{ (x^t, p^t ) \}_{t = 1}^T$ when for every $t\in \{1,\ldots,T\}$ the inequality
\[
u(x^t) - p^t \cdot x^t \geq u(x) - p \cdot x - \eps
\]
holds for every $x \in \mathbb{R}^K_+$.

To interpret these bounds, suppose for example that $\overline{u}(\xa, \xb, \eps) < 0$. We conclude that the individual ranks $\xb$ above $\xa$, even when the individual's choices do not exactly maximize utility. Thus, there is no ambiguity in the ranking of these bundles according to the unknown utility function $u$. If $\overline{u}(\xa, \xb, \eps) > 0$, then it is possible that the individual ranks $\xa$ above $\xb$. Lastly, if $\underline{u}(\xa, \xb, \eps) > 0$, then we conclude the individual ranks $\xa$ above $\xb$. More broadly, these bounds provide \textit{cardinal} information on utility differences, in units of the price of the numeraire. 

To gain some intuition how bounds on differences of utility are informed by data, consider two bundles $x^r$ and $x^s$ in the dataset. Since $x^r$ is approximately optimal given prices $p^r$, we have the restriction
\[
u(x^r) - p^r \cdot x^r \geq u(x^s) - p^r \cdot x^s - \eps,
\]
which rearranges to
\begin{align} \label{eq:utilityineq}
u(x^s) - u(x^r) \leq p^r \cdot (x^s - x^r) + \eps.
\end{align}
Differences in utility are thus bounded by changes in expenditure. Here, price is fixed and a change in quantity determines the magnitude of the expenditure change. The inequality in (\ref{eq:utilityineq}) arises because the point in the data $x^r$ was approximately optimal at prices $p^r$. Thus, $\eps$ here directly involves the observed data and is part of the measurement wedge. There is no prediction wedge because an analyst does not contemplate choices in new situations.

We first formalize computation of the bounds before providing additional interpretation. We show the bounds can be calculated as a linear program. An explicit description is relegated to Proposition~\ref{prop:lputility} in Appendix~\ref{a:mainproofs}.
\begin{prop} \label{prop:utilitylp}
For a dataset $\{ (x^t, p^t) \}_{t = 1}^T$, let $\eps \geq \eps^*$. If $\xb$ is in the dataset, i.e. $\xb = x^S$ for some $S \in \{1, \ldots, T \}$, then $\overline{u}(\xa, \xb, \eps)$ is finite and can be calculated as a linear program. If $\xa$ is in the dataset, i.e. $\xa = x^F$ for some $F \in \{1, \ldots, T \}$, then $\underline{u}(\xa, \xb, \eps)$ is finite and can be calculated as a linear program.

Under Assumption~\ref{assm:prime} ($\eps = \eps^*$), these bounds cannot be improved.
\end{prop}

Note that the set
\[
\left\{u  \mid u \text{ } \eps-\text{rationalizes } \{ (x^t, p^t ) \}_{t = 1}^T \right\}
\]
is convex in the sense that if each $u^a, u^b$ $\eps$-rationalize the dataset, then $\alpha u^a + (1 - \alpha)u^b$ does for $\alpha \in [0,1]$. This follows from inspecting inequalities such as
\[
u(x^t) - p^t \cdot x^t \geq u(x) - p \cdot x - \eps
\]
that define $\eps$-rationalizability by a utility function $u$. This means that any value between $\underline{u}(\xa, \xb, \eps)$ and $\overline{u}(\xa, \xb, \eps)$ can be attained.

To gain further intuition how data bound utility differences, we provide an analytical characterization. This characterization builds on inequalities such as (\ref{eq:utilityineq}) above, yet uses longer sequences (rather than just pairs) of observations to describe the tightest possible bounds. This parallels analysis of counterfactuals, where restrictions other than the law of demand arise by considering sequences of observations.
\begin{prop} \label{prop:ubounds}
For a dataset $\{ (x^t, p^t) \}_{t = 1}^T$, let $\eps \geq \eps^*$. If $\xb$ is in the dataset, i.e. $\xb = x^S$ for some $S \in \{1, \ldots, T \}$, then for any $\xa \in \mathbb{R}^K_+$ with $\xa \neq x^S$, the upper bound on utility differences is given by
\[
\overline{u}(\xa, x^S, \eps) = \min_{\sigma \in \Sigma_{S}} \left\{ p^{\sigma(M)}\cdot(\xa-x^{\sigma(M)}) + \sum_{m=1}^{M-1}p^{\sigma(m)}\cdot(x^{\sigma(m+1)}-x^{\sigma(m)}) + M\eps \right\},
\] 
where $\Sigma_S$ is the set of sequences that start with $\sigma(1)=S$, have no cycles, and have length at least $M\ge 1$. Moreover, the function $\overline{u}$ is strictly increasing and continuous in $(\xa,\eps)$ over the region that satisfies $\eps \geq \eps^*$ and excludes $\xa = x^S$.
\end{prop}
The sums inside the minimum are closely related to sums discussed in Proposition~\ref{prop:Xdetail} for counterfactuals. The sums differ because Proposition~\ref{prop:Xdetail} constructs sequences making a cycle (to remove the unknown utility numbers). In contrast, Proposition~\ref{prop:ubounds} considers sequences that do not make a cycle because the goal is to examine differences of utility numbers.

Continuity and concavity fail at $\xa = x^S$ (when $\eps > 0$) because the difference in utility is zero when the quantity is the same. Since $\overline{u}(\xb,\xa,\eps) = -\underline{u}(\xa,\xb,\eps)$, analogous results hold for the lower bound $\underline{u}(x^F, \xa,\eps)$ if the first argument $x^F$ is in the dataset. See Proposition~\ref{prop:lputility} for formal results.

An important feature for practical application is that the bounds on utilities are trivial unless an appropriate quantity is in the dataset. We formalize this as follows.
\begin{prop} \label{prop:unboundedu} 
For a dataset $\{ (x^t, p^t) \}_{t = 1}^T$, let $\eps \geq \eps^*$. If $\xb$ is not in the dataset, i.e. $\xb \neq x^r$ for every $r \in \{1, \ldots, T \}$, then
\[
\overline{u}(\xa, \xb, \eps)  = \infty.
\]
If $\xa$ is not in the dataset, then
\[
\underline{u}(\xa, \xb, \eps) = -\infty.
\]
\end{prop}
Recall that Proposition~\ref{prop:ubounds} shows that $\overline{u}$ is strictly increasing and continuous over a region. Thus, the upper bound on utility differences has some shape restrictions like a ``nice'' utility function. Despite this, the bound is not concave/continuous in the first argument at $\xa = x^S$ (when $\eps > 0$). This means that imposing concavity/continuity can potentially tighten the bounds.\footnote{Continuity and concavity do not tighten the bounds when $\eps = 0$ because the upper bound $\overline{u}$ is then continuous and concave for all values of $\xa$. See the proof of Proposition~\ref{prop:lputility} for more details.} Imposing these (or other) shape restrictions is important if one wishes to bound utility when neither quantity is in the dataset, since from Proposition~\ref{prop:unboundedu} we know the bounds are trivial without more structure.

\subsection{Recoverability of Approximate Indirect Utility} \label{sec:indirectbounds}

We now turn to welfare analysis concerning price changes. Here both the measurement and prediction wedge play a role. Recall that a ``measurement wedge'' shows up for bounds on the utility over bundles as in Section~\ref{sec:utilitybounds} since observations may not \textit{exactly} maximize a quasilinear utility function. Here the prediction wedge also arises when $\eps^* > 0$ because even when we know the utility function, we do not know which approximately-optimal choice would be made at a new price. 

We first discuss the indirect utility, which is the standard welfare object for the exact quasilinear model. Later we introduce the approximate indirect utility to account for the fact that an individual does not exactly optimize. The indirect utility function associated with the utility function $u : \mathbb{R}^K_+ \rightarrow \mathbb{R}$ is given by
\[
V_u(p) = \sup_{x \in \mathbb{R}^K_+} u(x) - x \cdot p.
\]
Since the researcher does not know the individual's utility \textit{a priori}, we consider indirect utility associated with candidate utility functions.

We show how indirect utility interacts with the measurement wedge. If $x^t$ is within $\eps$ of the maximum utility possible at price $p^t$, then we can write
\[
V_u(p^t) \leq u(x^t) - p^t \cdot x^t + \eps.
\]
The definition of the indirect utility yields for arbitrary $p \in \mathbb{R}^K_{++}$, the inequality
\[
V_u(p) \geq u(x^t) - p \cdot x^t.
\]
Differencing these, we obtain
\begin{equation} \label{eq:indineq}
V_u(p^t) - V_u(p) \leq x^t \cdot (p - p^t) + \eps.
\end{equation}
Here, $\eps$ arises because the observed choices need not be exact maximizers and thus is part of the measurement wedge. With a restriction on the magnitude of $\eps$, we can use observations of $x^t$ and $p^t$ to bound differences in indirect utility.

We now introduce the prediction wedge. This wedge arises because \textit{raw} differences in indirect utility are not the natural welfare object in our setting for a price change because we focus on ex ante policy evaluation. Instead, we take into account that when $\eps > 0$, an individual may choose bundles with different utility when facing the same prices. This is because we assume an individual satisfices. 

The utility the individual attains for a given price and choice of consumption bundle is the \textit{approximate indirect utility}. For observation $t \in \{1,\ldots,T\}$, the approximate indirect utility is
\[
u(x^t) - p^t \cdot x^t.
\]
At price $p$, the approximate indirect utility is restricted to be somewhere in the interval
\[
[V_u(p) - \eps, V_u(p)].
\]
In fact, (weakly) further restrictions take into account that the approximate indirect utility attained is bounded below by
\[
\underline{V}_{u,A}(p,\eps) = \inf_{x \in \mathbb{R}_+^K} u(x) - p \cdot x \qquad \text{s.t. } \quad u(x) - p \cdot x \geq V_u(p) - \eps,
\]
while the upper bound is the indirect utility. The lower bound on approximate indirect utility is the lower bound $\underline{V}_{u,A}(p,\eps)$, while the upper bound on approximate indirect utility is the upper bound $\overline{V}_{u,A}(p,\eps) = V_u(p)$.

Now suppose we wish to bound the change in approximate indirect utility between prices $\pb$ and $\pa$. If the utility $u$ and level of satisficing $\eps$ were known, then the welfare bounds would be
\[
[\underline{V}_{u,A}(\pa,\eps) - \overline{V}_{u,A}(\pb,\eps), \overline{V}_{u,A}(\pa,\eps) - \underline{V}_{u,A}(\pb,\eps)].
\]
Fixing $u$, this interval becomes wider when $\eps$ increases. In general, $\eps$ controls the prediction wedge, which arises even if we knew $u$ because we would not know what is chosen by the satisficer.

Since we do not know the utility function \textit{a priori}, we consider bounds involving the smallest and largest changes in approximate indirect utility among all utility functions that $\eps$-rationalize the dataset:
\begin{align*}
\overline{V}(\pa,\pb,\eps) & = \sup_{ \left\{u \mid u \text{ } \eps-\text{rationalizes } \{ (x^t, p^t ) \}_{t = 1}^T \right\}}  \left\{ \overline{V}_{u,A}(\pa,\eps) - \underline{V}_{u,A}(\pb,\eps) \right\} \\
\underline{V}(\pa,\pb,\eps) & = \inf_{ \left\{u \mid u \text{ } \eps-\text{rationalizes } \{ (x^t, p^t ) \}_{t = 1}^T \right\}}  \left\{ \underline{V}_{u,A}(\pa,\eps) - \overline{V}_{u,A}(\pb,\eps) \right\}. \footnotemark
\end{align*}

\footnotetext{Formally, we take the supremum over $u$ such that $\overline{V}_{u,A}(\pb,\eps)$ is not $\infty$, and the infimum over $u$ such that $\overline{V}_{u,A}(\pa,\eps)$ is not $\infty$.}

These bounds incorporate both the measurement and prediction wedges. The measurement wedge shows up when considering $u$ that $\eps$-rationalize the data, while the prediction wedge arises when defining the approximate indirect utility. We use the same value $\eps$ for both since we maintain Assumption~\ref{assm:prime}.

These bounds can inform a researcher about changes in welfare even in the presence of satisficing. If $\underline{V}(\pa,\pb,\eps) > 0$, then we can conclude that given a price change from $\pb$ to $\pa$ the individual is better off at $\pa$. If $\overline{V}(\pa,\pb,\eps) < 0$, then the price change from $\pb$ to $\pa$ makes the individual worse off. In contrast, ambiguity arises when $\underline{V}(\pa,\pb,\eps)<0$ and $\overline{V}(\pa,\pb,\eps) > 0$. In this case an individual may be better or worse given the price change, but the data alone are inconclusive.  

We now state a computational result for the bounds. A specific description of the linear program is given in Proposition~\ref{prop:lpwelfare} in Appendix~\ref{a:mainproofs}.
\begin{prop}\label{prop:welfarecomp}
For a dataset $\{(x^t,p^t)\}_{t=1}^T$, let $\eps \geq \eps^*$. The bounds on approximate indirect utility $
\overline{V}(\pa,\pb,\eps)$ and $\underline{V}(\pa,\pb,\eps)$ can each be computed as a linear program whenever they are finite.

Under Assumption~\ref{assm:prime} $(\eps = \eps^*)$, these bounds cannot be improved.
\end{prop}

When $\eps = \eps^* = 0$, the approximate indirect utility equals the indirect utility, and these are the sharp bounds on consumer surplus with limited price variation. When we set $\eps = \eps^*$, these are the \textit{adaptive consumer surplus bounds}. These bounds may be used for arbitrary prices $\pa, \pb$, not only at prices in $\{ ( p^t ) \}_{t = 1}^T$. In particular, these bounds provide welfare bounds at new prices \textit{without} needing to first provide bounds on the quantities at the prices.

Recall (\ref{eq:indineq}) established for $p^t$ in the dataset,
\[
V_u(p^t) - V_u(p) \leq x^t \cdot (p - p^t) + \eps.
\]
This states that differences in indirect utility are bounded by changes in expenditure. Here, the change in expenditure involves keeping the quantity fixed and changing prices. We present lower and upper bounds on $\overline{V}$ that build on this inequality. To state the result, first suppose $\pa = p^S$ is in the dataset. Define
\[
h(\pb) = \min_{\sigma \in \Sigma_S} \left\{ x^{\sigma(M)} \cdot (\pb - p^{\sigma(M)}) + \sum_{m = 1}^{M-1} x^{\sigma(m)} \cdot (p^{\sigma(m+1)} - p^{\sigma(m)}) + M \eps \right\},
\]
where $\Sigma_S$ is the set of sequences that start with $\sigma(1)=S$, have no cycles, and have length at least $M\ge 1$.
\begin{prop} \label{prop:welfaresandwich}
For a dataset $\{ (x^t, p^t) \}_{t = 1}^T$, let $\eps \geq \eps^*$. If $\pa$ is in the dataset, i.e. $\pa = p^S$ for some $S \in \{1, \ldots, T \}$, then for $\pa \neq \pb$,
\[
h(\pb) - \eps \leq \overline{V}(\pa,\pb,\eps) \leq h(\pb) + \eps,
\]
and for $\pa = \pb$, $\overline{V}(\pa,\pb,\eps) = \eps$.
\end{prop}
This result is established by leveraging duality results we present in Appendix~\ref{supp:duality}. Analogous results exist for $\underline{V}$ because $\overline{V}(\pa,\pb,\eps) = -\underline{V}(\pb,\pa,\eps)$, and are omitted for brevity.
 
We reiterate that Proposition~\ref{prop:welfarecomp} describes that $\overline{V}$ can be computed exactly as a linear program. The goal of Proposition~\ref{prop:welfaresandwich} is to make this process less of a ``black box.'' Note that when $\eps = 0$, the lower and upper bounds coincide and we characterize $\overline{V}(\pa,\pb,0)$. We recognize $h$ as a function closely related to the  construction of the Riemann integral, since it computes the sum of the area of certain rectangles. We may view $h$ as a ``discrete'' analogue of the consumer surplus formula, which states that differences in indirect utility are the area of a demand function. In fact, this integration intuition can be formalized in the special case of a single good $(K = 1)$, when $\eps = 0$.
\begin{prop} \label{prop:intformula}
Suppose there is a single good ($K = 1$), the dataset $\{ (x^t, p^t \}_{t = 1}^T$ is exactly consistent with quasilinear utility ($\eps^* = 0$), and we set $\eps = 0$. If $\pa > \min \{p_1, \ldots, p_T, \pb\}$, then
\[
\overline{V}(\pa,\pb,0) = \int^1_0 \overline{x}_1(t \pa + (1 - t) \pb, 0) (\pb_1 - \pa_1) dt.
\]
\end{prop}
Proposition~\ref{prop:intformula} shows that in a certain case, there is a tight connection between bounds on quantities and welfare bounds. Further relationships between welfare and counterfactual quantities are left for future work.

To close this section, we present shape restrictions on $\overline{V}$ that hold for all $\eps \geq 0$. 

\begin{prop} \label{prop:shapev}
For a dataset $\{ (x^t, p^t) \}_{t = 1}^T$, let $\eps \geq \eps^*$. $\overline{V}(\pa,\pb,\eps)$ is convex, weakly decreasing, and lower semicontinuous in $\pa$,\footnote{A function $f : H \rightarrow \mathbb{R} \cup \{ \infty, -\infty \}$ is lower semicontinuous if for any $a \in \mathbb{R}$ the set $\{ x \in H \mid f(x) \le a \}$ is closed in the topology on $H$.} and weakly increasing in $\eps$. If $\pa \in \CCo(\{p^t\}_{t = 1}^T)$, then $\overline{V}(\pa,\pb,\eps)$ is finite. If $\pa \not\in \CCo(\{p^t\}_{t = 1}^T \cup \pb)$, then $\overline{V}(\pa,\pb,\eps) = \infty$.
\end{prop}
The shape restrictions in Proposition~\ref{prop:shapev} are those of an indirect utility function. We do not obtain global continuity here because the welfare bounds can be infinite. However, $\overline{V}(\pa,\pb,\eps)$ is continuous in $\pa$ over the relative interior of $\CCo(\{p^t\}_{t = 1}^T)$ because it is convex and finite over this set (\cite{rockafellar2015convex}, Theorem 10.1).

Recall that Proposition~\ref{prop:unboundedu} shows that bounds on utility differences are trivial when a quantity is not in the dataset. In contrast, Proposition~\ref{prop:shapev} shows that the bound on approximate indirect utility $\overline{V}(\pa,\pb,\eps)$ is typically finite provided $\pa$ is not too low. In particular, neither $\pa$ nor $\pb$ need be in the dataset. The reason we obtain these contrasting results is that indirect utility functions must satisfy certain shape restrictions while we consider utility functions that need not satisfy shape restrictions such as concavity or monotonicity.

\section{Continuity and Convexity in Quantities and Approximation Error} \label{sec:shape}

In classic revealed preference, a small amount of measurement error can lead to refutation of the model. In this case, there is no way to use the model for counterfactual or welfare analysis. Below we show continuity of the welfare/counterfactual bounds in both quantities and degree of approximation error. Thus, we provide a way to still conduct analysis when the model is not perfect, and do so in a way that is a continuous enlargement of the standard conceptual framework.\footnote{To be clear, results in this paper are also new under $\eps = 0$ with a few exceptions that are noted.}

In more detail, here we study the joint mapping from quantities and approximation error to the bounds analyzed previously. One motivation for this is that in applications, an analyst may not observe a dataset of interest $\{ (x^t, p^t)\}_{t = 1}^T$ exactly, and may instead only have an estimate of the quantities. We show below that if we can consistently estimate quantities, then we can consistently estimate the bounds. 

For concreteness, suppose an analyst is conducting a representative agent analysis, and quantities are mean quantities from a population at a each time period. We examine the mean demand vector at period $t$ so that $x^t = \E[ X^{i,t}]$, where $X^{i,t}$ is demand for individual $i$ at time $t$. Here $X^{i,t}$ is treated as a random variable that is identically distributed across individuals. An analyst estimates $\widehat{\E[X^{i,t}]}$ from a cross-sectional dataset in which individuals at each time period face the same prices. For example, the estimator could be the sample average of demands at time $t$ across many individuals.

We now turn to the formal results. While we allow estimation error associated with quantities, here we take each price $p^t$ as nonrandom and measured exactly. Recall that $\overline{x}_k(\tilde{p},\eps)$ is the maximal quantity of good $k$ at counterfactual price $\tilde{p}$ assuming approximation error is no greater than $\eps$. The definition of $\overline{x}_k$ is presented in  Proposition~\ref{prop:qcounterfactuals}. We now treat $\overline{x}_k$ as a function of the dataset of quantities, and with minor abuse of notation we write $\overline{x}_k \left(d, \tilde{p}, \eps \right)$, where $d = \left(d^1, \ldots, d^T \right) \in \mathbb{R}^{K \times T}_+$ denotes quantities across all goods at time periods. This allows us to study how the bound depends on quantities in the dataset while keeping prices, $\{p^t\}_{t=1}^T$, fixed. Similarly, $\underline{x}_k \left(d,\tilde{p},\eps \right)$ denotes the lower bound. Finally, let $A \subseteq \mathbb{R}^{K \times T}_{+} \times \mathbb{R}_+$ denote combinations of quantities and approximation error such that the counterfactual/welfare objects are defined, i.e.
\[
A = \{ (d,\eps) \in \mathbb{R}_{+}^{K \times T} \times \mathbb{R}_{+} \mid \eps \ge \eps^* \left(\{ (d^t, p^t) \}_{t = 1}^{T} \right) \}.
\]
\begin{prop} \label{prop:convexcount}
Fix a price $\tilde{p}$ where we wish to bound counterfactual quantities and assume the dataset of prices $\{ p^t \}_{t = 1}^T$ is fixed. The set $A$ is convex. The mapping $\overline{x}_k(\cdot, \tilde{p}, \cdot) : A \rightarrow \mathbb{R}_+ \cup \{ \infty \}$ is  concave in $(d,\eps)$, and is continuous in $(d,\eps)$ at any point where it is finite. The mapping $\underline{x}_k(\cdot, \tilde{p}, \cdot) : A \rightarrow \mathbb{R}_{+}$ is convex and continuous in $(d,\eps)$.
\end{prop}
We obtain a similar result for the bounds on approximate indirect utility $\overline{V}$ and $\underline{V}$ when we view them as a function of the dataset of quantities. To formalize this, with minor abuse of notation let $\overline{V} \left(d,\pa,\pb,\eps \right)$ describe the upper bound as a mapping of  the quantities $d \in \mathbb{R}^{K \times T}_+$ in a dataset. Similarly, $\underline{V} \left(d, \pa,\pb,\eps \right)$ denotes the lower bound.

\begin{prop} \label{prop:convexwelfare}
Fix a price pair $\pa$ and $\pb$ where we wish to bound the difference in approximate indirect utility, and assume the dataset of prices $\{ p^t \}_{t = 1}^T$ is fixed. The mapping  $\overline{V}(\cdot, \pa,\pb,\cdot) : A \rightarrow \mathbb{R} \cup \{ \infty \}$ is concave in $(d,\eps)$, and is continuous in $(d,\eps)$ at any point where it is finite. The mapping $\underline{V}(\cdot, \pa,\pb,\cdot) : A \rightarrow \mathbb{R} \cup \{ -\infty \}$ is convex in $(d,\eps)$, and is continuous in $(d,\eps)$ at any point where it is finite.
\end{prop}

Recall Proposition~\ref{prop:shapev} shows that when $\pa \in \CCo(\{p^t\}_{t=1}^T)$, $\overline{V}(d, \pa,\pb,\eps)$ is finite for any $(d,\eps) \in A$.

We need one more result. Here, we interpret the minimal approximation error $\eps^*$ as a function of quantities for fixed  prices $\{ p^t \}_{t = 1}^T$.
\begin{prop}[\cite{allen2020satisficing}] \label{prop:meascont}
The mapping $\eps^* :  \mathbb{R}^{K \times T}_+ \rightarrow \mathbb{R}_+$ is convex and continuous.
\end{prop}
\iffalse
By directional differentiability we mean that for quantities $\tilde{d} \in \mathbb{R}^{K \times T}_+$ and direction $h \in \mathbb{R}^{K \times T}$,
\[
\lim_{\lambda \downarrow 0} \frac{\eps^* ( \tilde{d} + \lambda h) - \eps^* ( \tilde{d} )}{\lambda}
\]
exists. Formally, we mean this only for directions such that $\tilde{d} + \lambda h \in \mathbb{R}^{K \times T}_+$ for small positive $\lambda$.
\fi

The previous continuity results imply the following consistency results.
\begin{cor} \label{cor:counterfactualconsistency}
Suppose we have some estimator of the quantities that satisfies $\hat{d}^n \xrightarrow{p} d$. Then
\begin{align*}
\overline{x}_k \left(\hat{d}^n, \tilde{p}, \eps^*\left(\hat{d}^n \right)\right) & \xrightarrow{p} \overline{x}_k(d, \tilde{p}, \eps^*(d)) \\
\underline{x}_k \left(\hat{d}^n, \tilde{p}, \eps^*\left(\hat{d}^n \right) \right) & \xrightarrow{p} \underline{x}_k(d, \tilde{p}, \eps^*(d) ). \\
\overline{V} \left(\hat{d}^n, \pa, \pb, \eps^*\left(\hat{d}^n \right) \right) & \xrightarrow{p} \overline{V} \left(d, \pa, \pb, \eps^*\left(d \right) \right) \\
\underline{V} \left(\hat{d}^n, \pa, \pb, \eps^*\left(\hat{d}^n \right) \right) & \xrightarrow{p} \underline{V} \left(d, \pa, \pb, \eps^*\left(d \right)  \right),
\end{align*}
where each result holds whenever the right hand side result is finite.
\end{cor}
This provides a theoretical foundation for plug-in estimation. We omit a formal description of the sampling scheme since the result applies to any collection of random variables $\left\{ \hat{d}^n \right\}$ that converges in probability to $d$. For example, if we have panel data and the quantities $(X^{i,t})_{t = 1}^T$ are independent and identically distributed across individuals, one can use sample averages so that $\hat{d}^n= \left( \frac{1}{n}\sum_{i=1}^n X^{i,t} \right)_{t=1}^T$ when estimating $d= \left( \mathbb{E}[X^{i,t}] \right)_{t=1}^T$.

Finally, we consider shape restrictions for the bounds on utility differences $\overline{u}$ and $\underline{u}$, viewed as functions of quantities. As before, fixing prices $\{ p^t \}_{t = 1}^T$, we study dependence on the quantities $d \in \mathbb{R}^{K \times T}_{+}$. With minor abuse of notation write $\overline{u}(d, \xb, \xa, \eps)$ as a function of quantities and approximation error, and similarly for $\underline{u}$.

Recall that Proposition~\ref{prop:ubounds}  shows that when $\xb$ is in the dataset of quantities $d$, $\overline{u}(d,\xa,\xb,\eps)$ is finite provided $\eps \geq \eps^*$. In contrast, Proposition~\ref{prop:unboundedu} shows that whenever $\xb$ is outside the dataset, $\overline{u}(d,\xa,\xb,\eps) = \infty$. We conclude that when viewed as a mapping of quantities, $\overline{u}$ is no longer continuous. It is, however, continuous over a certain subset of $A$. To describe this, for a vector $\tilde{x} \in \mathbb{R}^K_{+}$ let $A(\tilde{x}) = \{ (d,\eps) \in A \mid d^1 = \tilde{x} \}$. This restricts attention to quantities datasets that all contain a certain vector $\tilde{x}$ as the first component. We formalize continuity and concavity results as follows.

\begin{prop} \label{prop:ushape}
Fix a quantity pair $\xa \neq \xb$ where we wish to bound the difference in utility, and assume the dataset of prices $\{ p^t \}_{t = 1}^T$ is fixed. The set $A(\tilde{x})$ is convex for any $\tilde{x} \in \mathbb{R}_+^K$. The mapping $\overline{u}(\cdot,\xa,\xb,\cdot) : A \rightarrow \mathbb{R} \cup \{ \infty \}$ is concave and continuous in $(d,\eps)$ over the region $A(\xb)$. The mapping $\underline{u}(\cdot,\xa,\xb,\cdot) : A \rightarrow \mathbb{R} \cup \{ -\infty \}$ is convex and continuous in $(d,\eps)$ over the region $A(\xa)$.
\end{prop}
Continuity over all of $A$ does not hold and so one cannot directly apply the continuous mapping theorem to establish a consistency result like Corollary~\ref{cor:counterfactualconsistency}. If some quantity vector $d^1$ is measured without error, however, then it is possible to consistently estimate the bounds on utility differences between $d^1$ and other bundles though we omit details for brevity.

\section{Application} \label{sec:application}

We now illustrate the results in the paper with data on the demand for gasoline. Data are from the 2001 United States National Household Travel Survey, and have previously been used in \cite{blundell2012measuring}. The data are from a single cross-section. For brevity we refer to \cite{blundell2012measuring} for additional details, including construction of the particular sample.

The primary observables of interest are quantities and prices. Quantities are annual gasoline consumption, which is constructed from odometer readings and an estimate of fuel efficiency. Prices are the average tax-inclusive price per gallon, in the county where the individual lives.

First note it is possible to map the raw quantities and prices to a dataset $\{ (X^i, P^i) \}_{i = 1}^n$, and then apply our previous analysis. Here $i$ denotes the individual and $n$ denotes the sample size. We use this notation rather than $t$ and $T$ to emphasize we have a cross-section. We use upper case $X^i$ and $P^i$ to denote that these are random variables.

We do not use the raw dataset, and instead ``pre-process it'' to map to our framework. We do so because we have a cross-section of individuals. We wish to both to diminish the impact of sampling variability as well as incorporate heterogeneity along observable variables.\footnote{\cite{allen2020satisficing} study how stochastic shocks \textit{and} approximation error can be studied in a common framework. That paper provides several aggregation theorems, and discusses a representative agent in this setting.} As in, \cite{blundell2012measuring} we pre-process by first considering a partially linear model given by
\[
X = g(P,Y) + \beta'W + U,
\]
where $P$ is price, $Y$ is income, $W$ are observed covariates, and $U$ is unobservable heterogeneity. While \cite{blundell2012measuring} interpret $g$ as a demand curve for a representative agent for the general model of utility maximization subject to a budget constraint, here we have a different interpretation. We interpret $g(\cdot, Y)$ as the demand curve for the representative agent with income level $Y$; $Y$ thus serves as a demographic characteristic that alters the shape of the demand curve. We close the model with the restriction
\[
\E[ U \mid P = p, Y = y, W = w] = 0.
\]
This specification allows price sensitivity to depend on the level of income of an individual. For each level of income $\tilde{y}$, we consider a dataset (in the sense used previously in the paper) of the form
\[
D(\tilde{y}) = \left\{ \hat{g}(P^i,\tilde{y}), P^i \right\}_{i = 1}^{\tilde{n}},
\]
where $\hat{g}$ is an estimator of $g$ described below. Thus, $\hat{g}(P^i,\tilde{y})$ is akin to the structural quantity $x^t$ in the previous notation, and $P^i$ is akin to $p^t$. Like \cite{blundell2012measuring}, we consider prices between the $5$-th and $95$-th quantile to mitigate endpoint issues, so $\tilde{n}$ enumerates these observations.

The estimator $\hat{g}$ is constructed similar to \cite{blundell2012measuring}. We first  estimate $\hat{\beta}$ by a double residual regression as in \cite{robinson1988root}.\footnote{We use the biweight kernel with ad hoc bandwidth $.75$ after standardizing the data.} Then we set
\[
\hat{g}(P^i, \tilde{y}) = \frac{\sum_{j = 1}^N \left(X^j - \hat{\beta}'W_j\right) K_{h_p}(P^j - P^i)K_{h_y}(Y^j -\tilde{y})}{\sum_{j = 1}^N K_{h_p}(P^j - P^i)K_{h_y}(Y^j -\tilde{y})},
\]
where $K_h$ is a kernel with bandwith $h$. Following \cite{blundell2012measuring} we use the biweight kernel. Throughout, the bandwidths $h_p$ and $h_y$ are chosen so that
\[
\frac{h_p}{h_y} = \frac{\hat{\sigma}_P}{\hat{\sigma}_Y},
\]
where $\hat{\sigma}^2_P = \frac{1}{n}\sum_{i = 1}^n (P^i - \overline{P})^2$, $\overline{P} = \frac{1}{n} \sum_{i = 1}^n P_i$, $\hat{\sigma}^2_Y = \frac{1}{n}\sum_{i = 1}^n (Y^i - \overline{Y})^2$, and $\overline{Y} = \frac{1}{n} \sum_{i = 1}^n Y_i$. Note that these are all constructed with all observations.

Figure~\ref{fig:empirical} presents analysis with two choices of bandwidths. These correspond to the ad hoc choices $.75$ and $1$ after standardizing. The top two panels display the kernel-smoothed ``dataset''
\[
D \left(\overline{Y} \right) = \left\{ \hat{g} \left(P^i,\overline{Y} \right), P^i \right\}_{i = 1}^{\tilde{n}}
\]
as well as counterfactual bounds, where $\hat{g} \left(P^i,\overline{Y} \right)$ is interpreted as a quantity for observation $i$ facing prices $P^i$. Recall $\tilde{n}$ denotes the middle $90\%$ of observations in terms of price, where we drop the lower and upper $5\%$ to mitigate endpoint issues. Income is evaluated at the sample mean $\overline{Y}$. The welfare bounds for approximate indirect utility are displayed in the middle panels. The bounds are evaluated relative to the mean price $\overline{P} = 1.334$. The bounds for differences in utility for certain quantities are displayed in the lower panels. The utility bounds are relative to the median quantity in the dataset. There are 101 comparisons, which for computational reasons are made between 101 of the points in the dataset of quantities $\left\{ \hat{g} \left(P^i,\overline{Y} \right) \right\}_{i = 1}^{\tilde{n}}$. Recall from Proposition~\ref{prop:unboundedu}, comparisons in utility when one quantity is not in the dataset will have at least one trivial bound. This is why we restrict attention to comparisons in which both quantities are in the dataset. It is important to note that in practice, simply bounding quantities over a grid will lead to trivial bounds for many points in the grid.

\begin{figure}[H]

    \begin{subfigure}[h]{0.48\textwidth}
    \includegraphics[scale = .25]{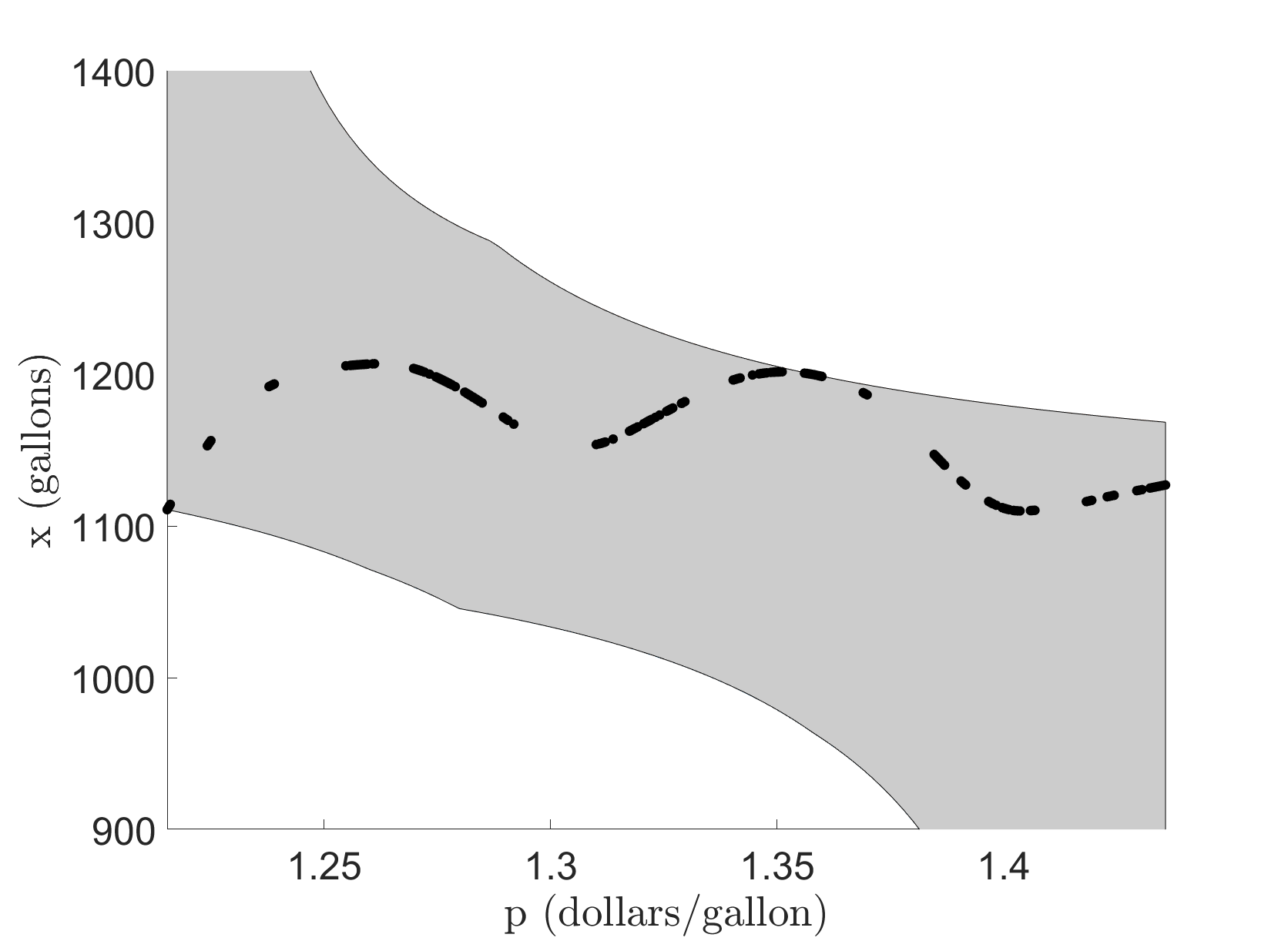}
    \caption{}\label{fig:low_a}
    \end{subfigure}
    \begin{subfigure}[h]{0.48\textwidth}
    \includegraphics[scale = .25]{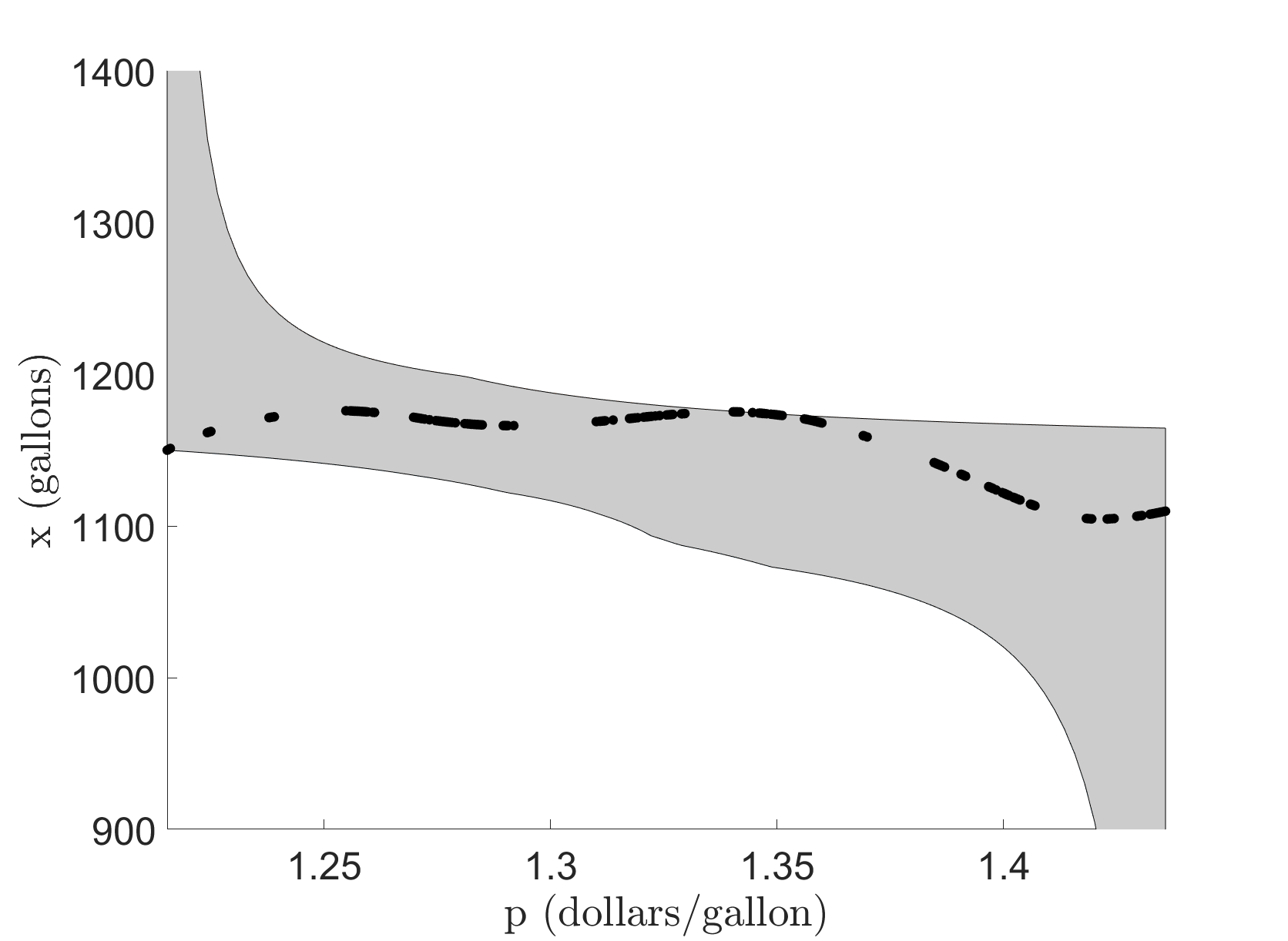}
    \caption{}\label{fig:low_b}
    \end{subfigure}

    \begin{subfigure}[h]{0.48\textwidth}
    \includegraphics[scale = .25]{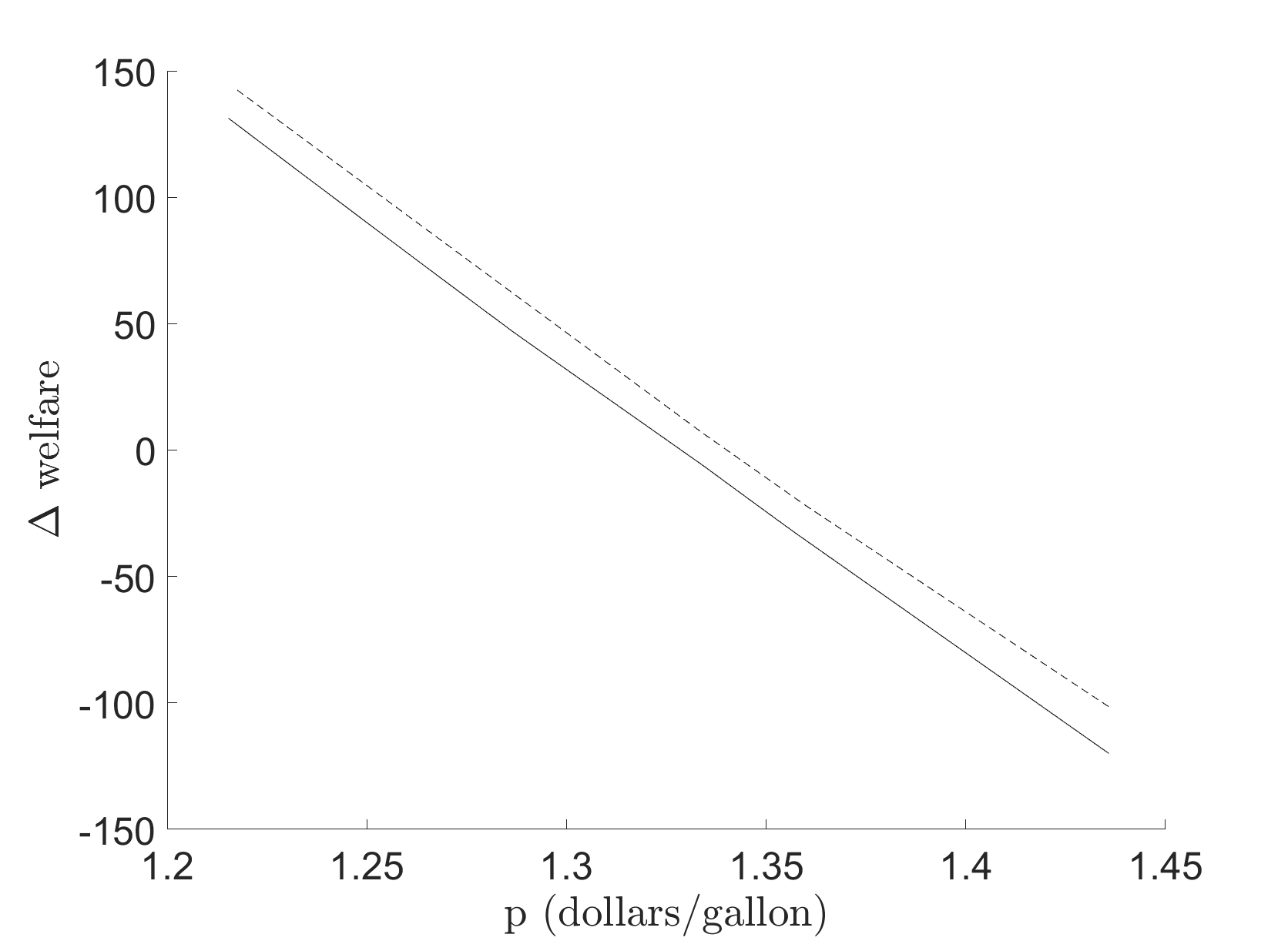}
    \caption{}\label{fig:low_c}
    \end{subfigure}
    \begin{subfigure}[h]{0.48\textwidth}
    \includegraphics[scale = .25]{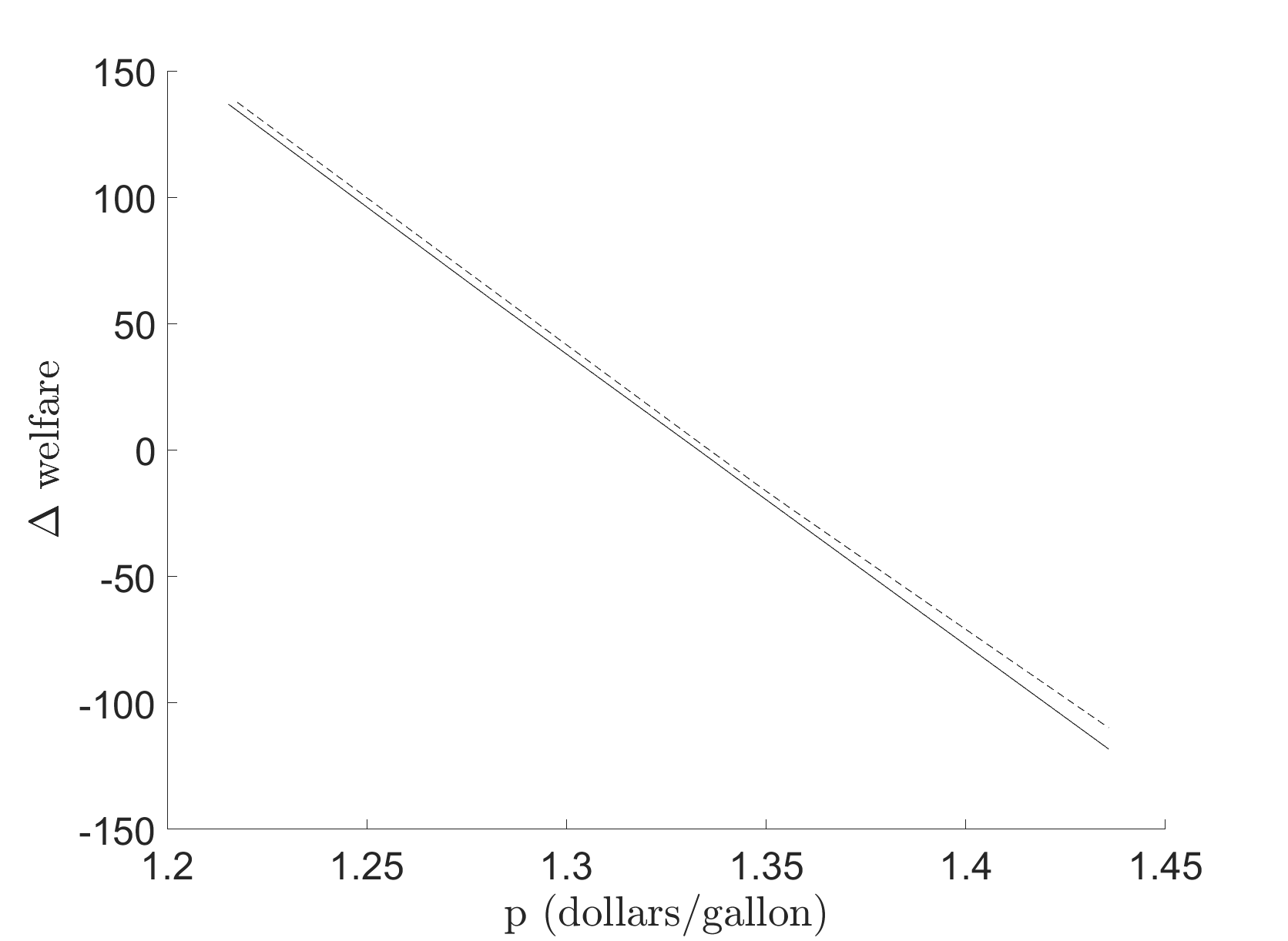}
    \caption{}\label{fig:low_d}
    \end{subfigure}
    
    \begin{subfigure}[h]{0.48\textwidth}
    \includegraphics[scale = .25]{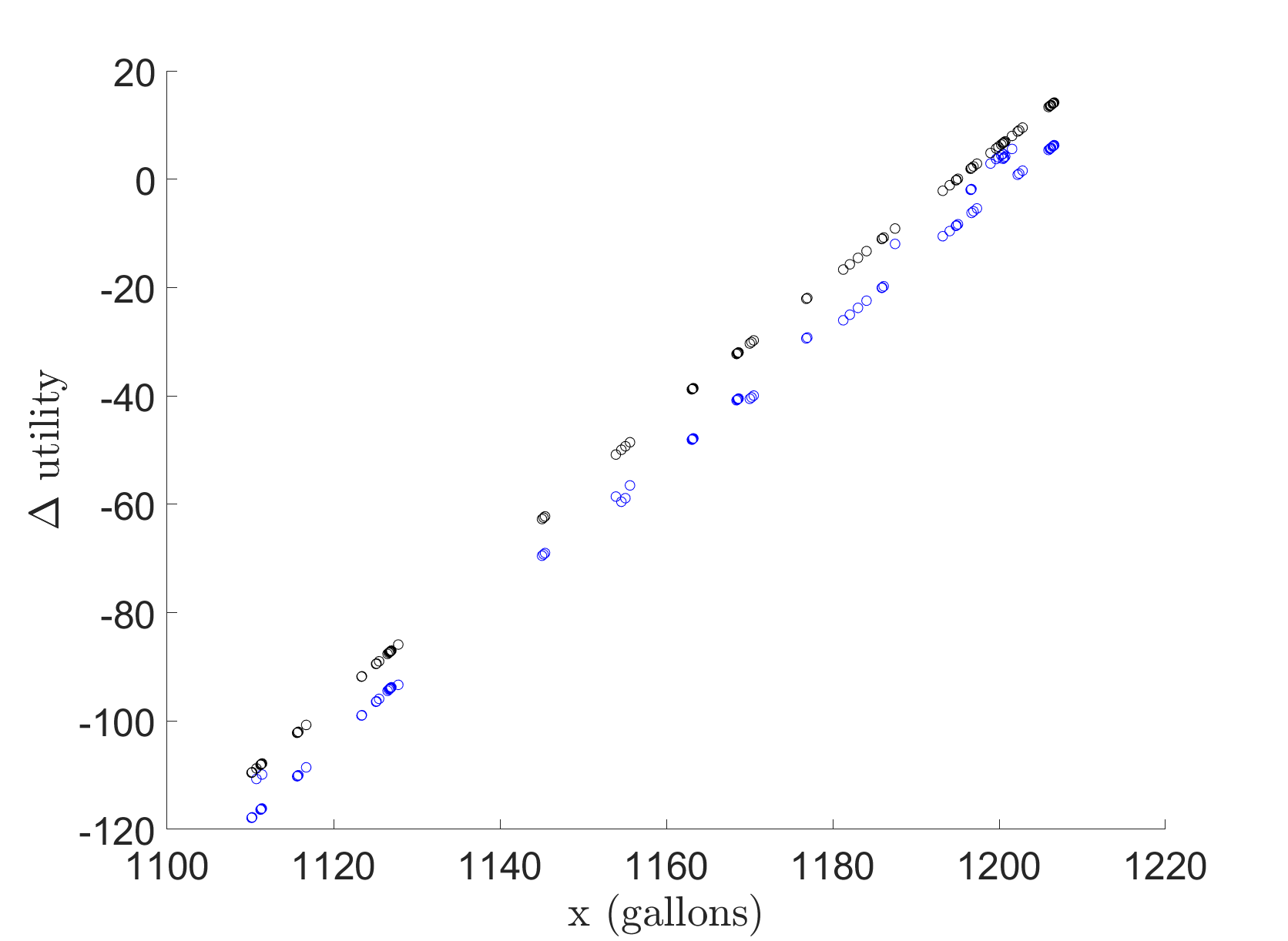}
    \caption{}\label{fig:low_c}
    \end{subfigure}
    \begin{subfigure}[h]{0.48\textwidth}
    \includegraphics[scale = .25]{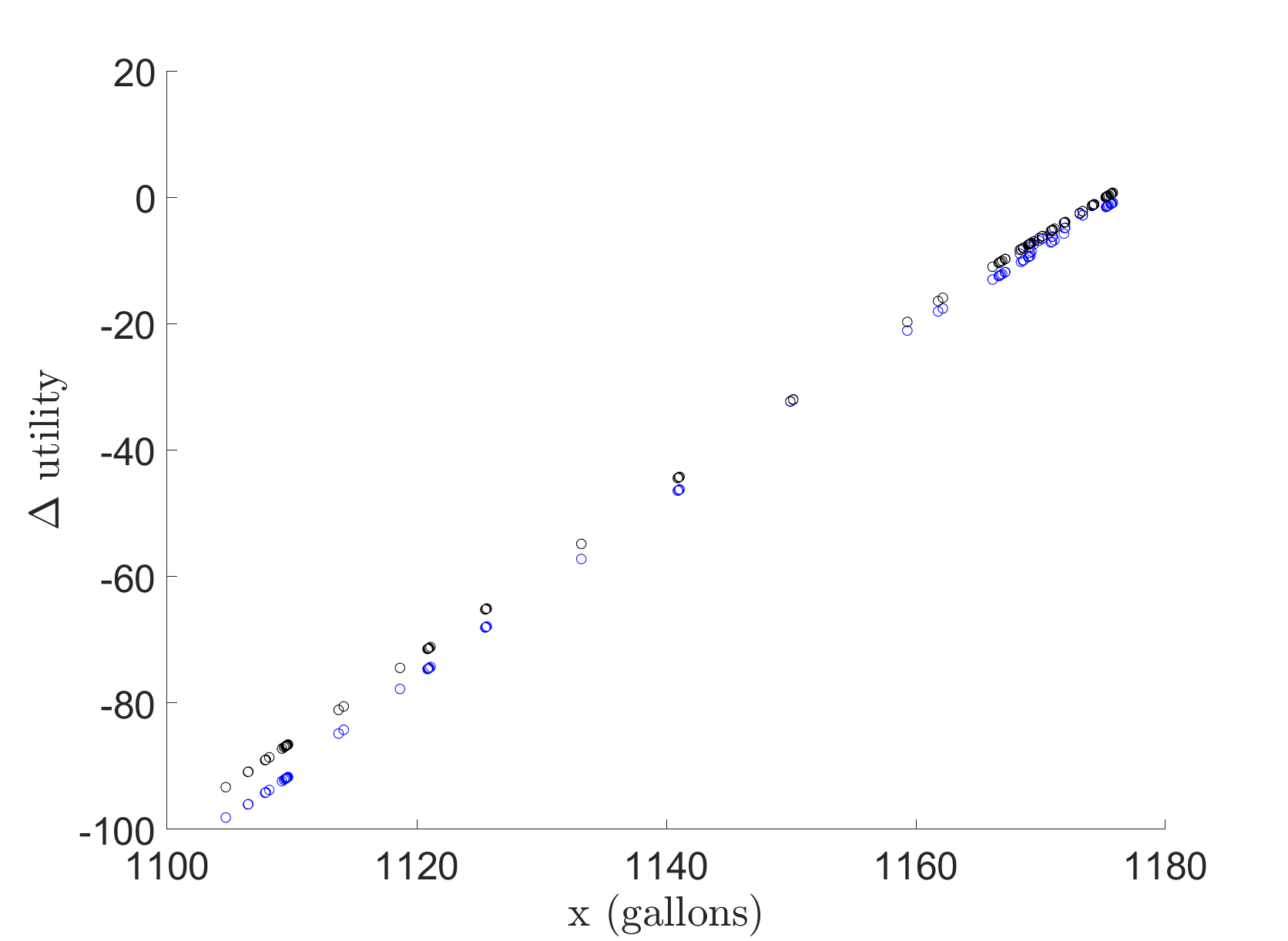}
    \caption{}\label{fig:low_d}
    \end{subfigure}
    
    \caption{Quantity, Approximate Indirect Utility, and Utility Bounds at Two Bandwidths}
    
    \subcaption*{Notes: The top panels depict quantity bounds at new prices. The middle panels depict bounds on approximate indirect utility relative to the mean price. The lower panels depict bounds on utility at certain quantities, relative to the median quantity in the dataset.}
    \label{fig:empirical}
    
\end{figure}

As can be seen from the figures, the choice of bandwidth noticeably alters the informativeness of the counterfactual bounds (upper panels). In contrast, the bounds on approximate indirect utility (middle panels) and utility (lower panels) are relatively narrow for both bandwidths. Similar results obtain for alternative bandwidths and are available upon request.

In the lower panels, a contrast emerges between the lower and upper bounds on utility. Recall that these bounds are for $u(\xa) - u(\xb)$, where $\xb$ is the median quantity, i.e. in the dataset. Because the second argument is in the dataset, Proposition~\ref{prop:ubounds} applies to the upper bounds and establishes monotonicity in the argument $\xa$. The graphs are consistent with this, since the upper bounds are monotone in quantities. In contrast, the lower bounds are not monotone. In order to make lower bounds on utilities monotone, it would be necessary to have the \textit{first argument} be in the dataset and fixed ($\xa$). See Proposition~\ref{prop:lputility}(vi) in Appendix~\ref{a:mainproofs}.

\section{Conclusion} \label{sec:conclusion}

This paper provides a conceptual framework for counterfactual and welfare analysis for approximate models. Our main conceptual assumption is that model approximation error has the same magnitude in new settings as the data we have seen. We formalize this for the quasilinear utility model. This assumption is portable to other settings, and generalizes the standard approach that requires correct specification in both the data we have seen and at hypothetical values.

Engaging with the possibility that a model may not perfectly match data is especially important using the nonparametric revealed preference approach. Indeed, a natural intuition is that if approximation error is ``small,'' then it is second order and we can ignore it for certain questions. Unfortunately, this intuition is false in the standard approach used in the revealed preference literature, since small violations of the model mean it cannot be used for counterfactual or welfare analysis. This paper presents an adaptive approach allowing the analyst to use the model formally viewing it as an approximation. Moreover, our counterfactual/welfare bounds are continuous in the degree of approximation error, and so they continuously transition to the standard framework when approximation error is negligible.

\counterwithin{theorem}{section}
\counterwithin{prop}{section}
\counterwithin{lemma}{section}
\counterwithin{cor}{section}
\counterwithin{assm}{section}
\counterwithin{defn}{section}
\counterwithin{remark}{section}

\begin{appendices}

\section{Proofs of Main Results} \label{a:mainproofs}

This appendix provides proofs of the results in the main text. It also provides explicit descriptions of the linear programs mentioned in the main text. Some of the proofs require additional lemmas contained in Supplemental Appendix~\ref{supp}.

\subsection{Proofs for Section~\ref{sec:counterfactuals}}

\begin{proof}[Proof of Facts \ref{f:const}-\ref{f:min}]
The proofs are in the main text.
\end{proof}

\begin{proof}[Proof of Proposition~\ref{prop:nonempty}]
Emptiness of $X(\tilde{p},D,\eps)$ when $\eps < \eps^*$ is immediate from Fact~\ref{fact:nonempty}. It remains to show that when $\eps \geq \eps^*$, the set $X(\tilde{p},D,\eps)$ is nonempty.

First, fix $(x^1, p^1) \in \{ (x^t, p^t) \}_{t=1}^T$ and let $\Sigma_1$ denote the set of finite sequences of $t \in \{1,\ldots,T\}$ with no cycles that begins at $(x^1, p^1)$. Define
\[
U(x) = \min_{\sigma \in \Sigma_1} \left\{ p^{\sigma(M)} \cdot \left(x - x^{\sigma(M)}\right) +  \sum_{m=1}^{M-1} p^{\sigma(m)} \cdot (x^{\sigma(m+1)} - x^{\sigma(m)})  + M  \eps \right\},
\]
where $\sigma \in \Sigma_1$ is a sequence of length  $M$, for all $m \in \{1,\ldots, M\}$ it follows that $\sigma(m) \in \{1,\ldots,T\}$, and $\sigma(1)=1$. \cite{allen2020satisficing} have shown that for $\eps \geq \eps^*$, this function $\eps$-rationalizes the data in the sense that for each $t \in \{1, \ldots, T\}$ and each $x \in \mathbb{R}^K_{+}$,
\[
U(x^t) - p^t \cdot x^t \geq U(x) - p^t \cdot x - \eps.
\]
The function $U(x)$ need not induce an $\eps$-maximizer when prices take low values. However, the constructed utility can be modified to guarantee maximizers exist. 

To that end, let \[
\overline{U} = \sup_{x \in \Co(\{x^t\}_{t = 1}^T) } U(x)
\]
where $\Co( \{x^t\}_{t = 1}^T )$ denotes the convex hull, i.e. the smallest convex set containing $\{x^t\}_{t = 1}^T$. We see $\overline{U} < \infty$ since $U$ is continuous and $\Co(\{x^t\}_{t = 1}^T )$ is compact. Define $f : \mathbb{R}^K_+ \rightarrow \mathbb{R}$ by $f(x) = \frac{ \sum_{k = 1}^K x_k }{ 1 + \sum_{k = 1}^K x_k } + \overline{U}$, which is bounded and concave. To see this, note the function $h : \mathbb{R}_+ \rightarrow \mathbb{R}$ given by $h(z)=z/(1 + z)$ is concave by inspecting derivatives. Since $f$ is a composition of a concave function and an affine and strictly increasing function it is concave. 

Now construct $\tilde{U}(x) = \min \{ U(x), f(x) \}$. This function rationalizes the data since for $t \in \{1, \ldots, T\}$ and $x \in \mathbb{R}^K_{+}$ we have
\[
\tilde{U}(x^t) - p^t \cdot x^t = U(x^t) - p^t \cdot x^t \geq U(x) - p^t \cdot x - \eps \geq \tilde{U}(x) - p^t \cdot x - \eps.
\]
In addition, $\tilde{U}$ is concave since it is the minimum of concave functions. Similarly, $\tilde{U}$ is  continuous and strictly increasing as it is the minimum of finitely many continuous and strictly increasing functions. It remains to show this utility admits an $\eps$-maximizer for all prices $p\in \mathbb{R}_{++}$.

To that end, note the indirect utility of $f$, denoted
\[
V_f(p) = \sup_{x \in \mathbb{R}^K_+} f(x) - p \cdot x,
\] is everywhere finite over the region $p \in \mathbb{R}^K_{++}$ because $f$ is bounded between $\overline{U}$ and $\overline{U} + 1$. Moreover, since  $\tilde{U} \leq f$ pointwise, we also have $V_{\tilde{U}}(p) \leq V_f(p)$, so that $V_{\tilde{U}}$ is finite for any $p \in \mathbb{R}^K_{++}$. Since $\mathbb{R}^K_{++}$ is open, from Lemma~\ref{lem:existence} we conclude that
\[
\tilde{U}(x) - p \cdot x
\]
admits an exact maximizer in $x$ for any $p \in \mathbb{R}^K_{++}$. In particular, it admits $\eps$-maximizers, completing the proof.
\end{proof}

\begin{proof}[Proof of Proposition~\ref{prop:Xdetail}]
An equivalent definition of $X(\tilde{p},D,\eps)$ is $\tilde{x} \in X(\tilde{p},D,\eps)$ if and only if the augmented dataset $D \cup (\tilde{x},\tilde{p})$ is $\eps$-rationalized by quasilinear utility. From the characterization in Lemma~\ref{lem:epsqrat}(iii), this is equivalent to showing that certain sequences satisfy an inequality. For each sequence involving the augmented dataset there are two cases. If the sequence does not contain $(\tilde{x},\tilde{p})$, then the inequality in Lemma~\ref{lem:epsqrat}(iii) is satisfied because we assume $\eps \geq \eps^*(D)$. (Note that $\eps^*(D)$ is constructed to have this property.) It remains to check sequences involving $(\tilde{x},\tilde{p})$. Rearranging the inequality of Lemma~\ref{lem:epsqrat}(iii), we see that
\begin{equation}
(\tilde{p} - p^{t_M})\cdot \tilde{x} \leq (M+1) \eps + \tilde{p} \cdot x^{t_1}  - p^{t_M} \cdot x^{t_M} - \sum_{m = 1}^{M-1} p^{t_m} \cdot (x^{t_m} - x^{t_{m+1}})
\end{equation}
must hold for all finite sequences $\{ t_m \}_{m = 1}^M$ without cycles where $t_m \in \{1, \ldots, T\}$ and $M \geq 1$. Here we have the $M+1$ coefficient on $\eps$ since the sequences include the counterfactual observation and a length $M$ sequence. This characterizes the set $X(\tilde{p},D,\eps)$ as an intersection of finitely many half-spaces. Thus, $X(\tilde{p},D,\eps)$ is a closed, convex polyhedron.
\end{proof}

To prove Proposition~\ref{prop:qcounterfactuals}, we prove a stronger result that explicitly describes the linear program.

\begin{prop} \label{a:qcounterfactuals}
For a dataset $D = \{ (x^t, p^t) \}_{ t = 1}^T$, let $\eps \geq \eps^*$. Then whenever $X_k(\tilde{p},D,\eps)$ is bounded above, its maximum is given by the linear program
\begin{align*}
\overline{x}_k(\tilde{p},\eps) = \max_{ \substack{ \tilde{x} \in \mathbb{R}_+^K \\ u^1, \ldots, u^T, \tilde{u} \in \mathbb{R}_+}} &\quad \tilde{x}_k & \\
        \text{s.t.} \quad u^s & \le u^r + p^r \cdot (x^s-x^r) + \varepsilon &\quad \text{for all}\; r,s \in \{1,\ldots, T\} \\
        \tilde{u} & \le u^r + p^r \cdot (\tilde{x}-x^r) + \varepsilon &\quad \text{for all}\; r \in \{1,\ldots, T\} \\
        u^{r} & \le \tilde{u} + \tilde{p} \cdot (x^r-\tilde{x}) + \varepsilon &\quad \text{for all}\; r \in \{1,\ldots, T\}.
\end{align*}
The upper bound $\overline{x}_k(\tilde{p},\eps)$ may be equivalently calculated as
\begin{align*}
\overline{x}_k(\tilde{p},\eps) = \max_{ \tilde{x} \in \mathbb{R}_+^K } &\quad \tilde{x}_k \\
\text{s.t.}&\quad (\tilde{p} - p^{t_M})\cdot \tilde{x} \leq (M+1) \eps + \tilde{p} \cdot x^{t_1} - p^{t_M} \cdot x^{t_M} - \sum_{m = 1}^{M-1} p^{t_m} \cdot (x^{t_m} - x^{t_{m+1}}),
\end{align*}
where this inequality must hold for all finite sequences $\{t_m\}_{m=1}^M$ with $t_m \in \{1,\ldots,T\}$ and $M\ge 1$. The value of  $\underline{x}_k(\tilde{p}, \eps)$ is calculated as the minimum of the objective with either constraint set of the above linear programs. The value $\underline{x}_k(\tilde{p}, \eps)$ is weakly increasing in $\eps$ over the region $\eps \geq \eps^*$, and $\underline{x}_k (\tilde{p},\eps)$ is weakly decreasing in $\eps$ over the region $\eps \geq \eps^*$.
\end{prop}

The first linear program is easy to implement as it has order $(T+1)^2$ constraints and $T+1+K$ unknowns. The second linear program is useful to understand the mapping from data to bounds. However, directly operationalizing the second linear program would require enumerating all finite sequences of the dataset that do not contain cycles, which is computationally costly.

Related bounds have appeared in \cite{chiong2017counterfactual} and \cite{ARidentification}, which focus on latent utility models with observable characteristics of goods other than prices. The result here differs since $\eps$ can be nonzero and the first set of bounds directly describes a convenient linear program used to compute bounds.\footnote{\cite{chiong2017counterfactual} essentially start with the second formulation of the bounds (in terms of cycles) and show that while there are many cycles, only a certain number are effectively binding. The first formulation of Proposition~\ref{a:qcounterfactuals} complements their analysis by describing an explicit linear program with order $(T+1)^2$ scalar inequalities. \cite{ARidentification} describe bounds in certain models with characteristics in place of prices, and use a characterization similar to the cycles condition, but do not study extreme points or describe computations.} We take $\overline{x}_k$ to be positive infinity when there is no finite upper bound. 

\begin{proof}[Proof of Proposition~\ref{a:qcounterfactuals} (and Proposition~\ref{prop:qcounterfactuals})]
The linear programming formulations are immediate from Lemma~\ref{lem:epsqrat} and the proof of Proposition~\ref{prop:Xdetail}. From Lemma~\ref{lem:lpexist}, the maximum is attained because the linear program has a bounded value function by construction. Recall that for the second formulation, we only need to consider cycles involving the counterfactual quantity-price tuple because we have assumed $\eps \geq \eps^*$. Recall that all cycles that do not involve the counterfactual quantity are necessarily less than $\eps^*$ and will not bind. We leveraged these properties in Proposition~\ref{prop:Xdetail} already. The proof for $\underline{x}_k$ are analogous and are omitted.

Now we argue that these bounds cannot be improved under Assumption~\ref{assm:prime}. To see this, note from Proposition~\ref{prop:Xfinite} that the bounds are weakly monotone in $\eps$. Moreover, when $\eps < \eps^*$ we know the set $X(\tilde{p},D,\eps)$ (and hence $X_k(\tilde{p},D,\eps)$) is empty from Proposition~\ref{prop:nonempty}.
\end{proof}

\begin{proof}[Proof of Proposition~\ref{prop:Xfinite}]
We begin by showing that $\overline{x}_k(\tilde{p},\eps)$ is finite if and only if $\tilde{p} \in \text{int} \text{CCo} ( \{p^t\}_{t=1}^T )$. First, let $\tilde{p} \in \text{int} \text{CCo} ( \{p^t\}_{t=1}^T )$ so that 
$\tilde{p} > \sum_{t=1}^T \alpha_t p^t$ for some nonnegative $\alpha_t$ such that $\sum_{t}^T \alpha_t=1$. Note that for each $t \in \{1,\ldots T\}$, the approximate law of demand yields
\[ (\tilde{p}-p^t)\cdot (\tilde{x}-x^t) \le 2\eps .\]
Multiplying by $\alpha^{t}$ and summing up the inequalities gives that 
\[
\left(\tilde{p}-\sum_{t=1}^{T} \alpha_t p^{t} \right)\cdot \left(\tilde{x}-\sum_{t=1}^L \alpha_t x^{t} \right) \le 2\eps .\]
Thus, 
\[ \left(\tilde{p}- \sum_{t=1}^{T} \alpha_t p^{t} \right) \cdot \tilde{x} \le 2\eps + \left(\tilde{p}-\sum_{t=1}^{T} \alpha_t p^{t} \right) \cdot \sum_{t=1}^L \left(\alpha_t x^{t} \right) .\]
Since $\left( \tilde{p}- \sum_{t=1}^{T} \alpha_t p^{t} \right)>0$ and $\tilde{x}\in \mathbb{R}_{+}^K$ one can bound the values on each dimension of $\tilde{x}$ so that 
\[ \tilde{x} \in \prod_{k=1}^K \left[0, \frac{2\eps + \left(\tilde{p}-\sum_{t=1}^{T} \alpha_t p^{t}\right) \cdot \sum_{t=1}^L \alpha_t x^{t})}{\tilde{p}_k- \sum_{t=1}^{T} \alpha_t p_k^{t}} \right]. \]
This shows that $\overline{x}_k(\tilde{p},\eps)$ is finite when $\tilde{p} \in \text{int} \text{CCo} ( \{p^t\}_{t=1}^T )$.

Next, we show when $\tilde{p} \notin \text{int} \text{CCo} ( \{p^t\}_{t=1}^T )$ that $\overline{x}_k(\tilde{p},\eps)$ is unbounded. Suppose that $\tilde{p} \notin \text{int} \text{CCo}( \{p^t\}_{t=1}^T )$. This means for all $t \in \{1,\ldots,T\}$ that $\tilde{p}\le p^t$. From Proposition~\ref{prop:Xdetail} we know $\tilde{x} \in X(\tilde{p},D,\eps)$ if and only if for any sequence $\{t_m\}_{m=1}^M$ with $M \geq 1$,
\[
(\tilde{p}-p^{t_M})\cdot(\tilde{x}-x^{t_1}) \le (M+1)\eps + p^{t_M} \cdot x^{t_1} - p^{t_M} \cdot x^{t_M} - \sum_{m=1}^{M-1} p^{t_m}\cdot (x^{t_m}-x^{t_{m+1}}).
\]
Note that the right hand side of the expression is always weakly positive. Moreover, $\tilde{p}-p^{t_M}\le 0$ for every dimension. If all terms are zero, this places no restrictions on $\tilde{x}$ so that one can choose arbitrarily positive amounts of every good. If there is a dimension $k$ such that $\tilde{p}_k-p_k^{t_M}<0$, then one can choose arbitrarily high amounts of $\tilde{x}_k$ to satisfy all such inequalities. This establishes that $\overline{x}_k(\tilde{p},\eps)$ is unbounded above.

Note that the lower bound $\underline{x}_k(\tilde{p},\eps)$ is always finite because it is bounded below by $0$.

To show monotonicity in $\eps$ note that the feasibility region is weakly increasing (with regard to set inclusion) as $\eps$ increases. Thus, $\overline{x}_k$ is weakly increasing in $\eps$, and $\underline{x}_k$ is weakly decreasing in $\eps$.
\end{proof} 

\begin{proof}[Proof of Proposition~\ref{prop:singleprice}]
First, let $\Sigma$ be the set of sequences $\{t_m\}_{m=1}^M$ that contain no cycles where  $t_m \in \{1,\ldots,T\}$. From Proposition~\ref{prop:Xdetail}, the counterfactual bounds on demand are given by inequalities of the form
\begin{equation}\label{eq:CDbounds}
(\tilde{p} - p^{t_M})\tilde{x} \leq (M+1) \eps + \tilde{p} x^{t_1}  - p^{t_M} x^{t_M} - \sum_{m = 1}^{M-1} p^{t_m} (x^{t_m} - x^{t_{m+1}}),
\end{equation}
which much hold for every sequence in $\Sigma$. Dot products are removed since all objects are one-dimensional.

Whether a certain inequality of the form (\ref{eq:CDbounds}) provides an upper bound or lower bound on $\tilde{x}$ depends on the sign of $\tilde{p} - p^{t_M}$. To see this, note that rearranging Equation~\ref{eq:CDbounds} when $\tilde{p}>p^{t_M}$ yields
\begin{align}
 \tilde{x} &\le \frac{ (M+1) \eps + \tilde{p} x^{t_1}  - p^{t_M} x^{t_M} - \sum_{m = 1}^{M-1} p^{t_m} (x^{t_m} - x^{t_{m+1}}) }{\tilde{p} - p^{t_M}} \nonumber\\
    &=x^{t_1} + \frac{(M+1)\eps + p^{t_M}x^{t_{1}} - p^{t_M} x^{t_M} - \sum_{m = 1}^{M-1} p^{t_m} (x^{t_m} - x^{t_{m+1}}) }{\tilde{p} - p^{t_M}} \label{propmon:second}. 
\end{align}
Note the the expression in Equation~\ref{propmon:second} that is divided by $\tilde{p}-p^{t_m}$ is positive since the terms above are those for a cycle of length $M$. To see this, note that
\[
p^{t_M} (x^{t_{M}} - x^{t_1}) + \sum_{m = 1}^{M-1} p^{t_m} (x^{t_m} - x^{t_{m+1}}) \leq M \eps \leq (M+1) \eps,
\]
where the first inequality holds because $\eps \geq \eps^*$ and the left hand side is a sequence of length $M$. Thus, such sequences constitute upper bounds. 

When instead $\tilde{p} < p^{t_M}$, the sequence yields a lower bound since one is dividing by a negative number. Since the sign of the difference matters, we partition the set of sequences in $\Sigma$ as follows. We consider the counterfactual prices where $\tilde{p}^0<\tilde{p}^1$ without loss of generality. Let $\{t_m\}_{m=1}^M = \sigma \in \Sigma^1$ when $p^{t_M} \le \tilde{p}^0 < \tilde{p}^1$. Let $\{t_m\}_{m=1}^M = \sigma \in \Sigma^0$ when $\tilde{p}^0 < p^{t_M} < \tilde{p}^1$. Lastly, let $\{t_m\}_{m=1}^M = \sigma \in \Sigma^{-1}$ when $\tilde{p}^0 < \tilde{p}^1 \le p^{t_M}$. 

Upper bounds on counterfactual demand for the price $\tilde{p}^0$ involve sequences where  $p^{t_M} \leq \tilde{p}^0$ (i.e. sequences in $\Sigma^1$). Upper bounds on counterfactual demand for the price $\tilde{p}^1$ involve sequences where $p^{t_M} \leq \tilde{p}^1$ (i.e. sequences in $\Sigma^1 \cup \Sigma^0$). We denote the upper bound inequalities by
\[ UB(\tilde{p}^0) = \{ \tilde{x} \in \mathbb{R}_+ \mid \text{Equation~\ref{eq:CDbounds} holds for sequences} \; \sigma \in \Sigma^1  \; \text{with} \; \tilde{p}=\tilde{p}^0 \} \]
and 
\[ UB(\tilde{p}^1)= \{ \tilde{x} \in \mathbb{R}_+ \mid \text{Equation~\ref{eq:CDbounds} holds for sequences} \; \sigma \in \Sigma^1 \cup \Sigma^0 \; \text{with} \; \tilde{p}=\tilde{p}^1 \} .\]

We use Equation~(\ref{propmon:second}) to show that $UB(\tilde{p}^1) \subseteq  UB(\tilde{p}^0)$. First, if $\tilde{x} \in UB(\tilde{p}^1)$ then for every sequence $\{t_m\}_{m=1}^M=\sigma \in \Sigma^1$ with $p^{t_M}<\tilde{p}^0<\tilde{p}^1$ it follows that
\begin{align*}
 \tilde{x} &\le x^{t_1} + \frac{(M+1)\eps + p^{t_M}x^{t_{1}} - p^{t_M} x^{t_M} - \sum_{m = 1}^{M-1} p^{t_m} (x^{t_m} - x^{t_{m+1}}) }{\tilde{p}^1 - p^{t_M}} \nonumber \\
    &\le x^{t_1} + \frac{(M+1)\eps + p^{t_M}x^{t_{1}} - p^{t_M} x^{t_M} - \sum_{m = 1}^{M-1} p^{t_m} (x^{t_m} - x^{t_{m+1}}) }{\tilde{p}^0 - p^{t_M}}
\end{align*}
where the second inequality holds since the numerator is positive and $0<\tilde{p}^0-p^{t_M}<\tilde{p}^1-p^{t_M}$. If $\tilde{p}^0=p^{t_M}$ for a sequence $\sigma \in \Sigma^1$, then there is no restriction on the counterfactual demands. Since $UB(\tilde{p}^0)$ only is restricted by sequences in $\Sigma^1$ while $UB(\tilde{p}^1)$ is restricted by sequences in $\Sigma^1$ and $\Sigma^0$, this shows $UB(\tilde{p}^1) \subseteq  UB(\tilde{p}^0)$. This proves that $\bar{x}(\tilde{p}^1,\eps) \le \bar{x}(\tilde{p}^0,\eps)$ since $\bar{x}$ is the maximum, the upper bounds satisfy $UB(\tilde{p}^1) \subseteq  UB(\tilde{p}^0)$, and a maximum over a larger set is weakly larger.

Next note that the lower bounds on counterfactual demand $\tilde{x}$ are given by the following 
\[ LB(\tilde{p}^0) = \{ \tilde{x} \in \mathbb{R}_+ \mid \text{Equation~\ref{eq:CDbounds} holds for sequences} \; \sigma \in \Sigma^{-1}\cup \Sigma^0  \; \text{with} \; \tilde{p}=\tilde{p}^0 \} \]
and 
\[ LB(\tilde{p}^1)= \{ \tilde{x} \in \mathbb{R}_+ \mid \text{Equation~\ref{eq:CDbounds} holds for sequences} \; \sigma \in \Sigma^{-1} \; \text{with} \; \tilde{p}=\tilde{p}^1 \} .\]
To see this, note that rearranging Equation~\ref{eq:CDbounds} when $\tilde{p}<p^{t_M}$ yields
\begin{align*}
 \tilde{x} &\ge x^{t_1} - \frac{(M+1)\eps + p^{t_M}x^{t_{1}} - p^{t_M} x^{t_M} - \sum_{m = 1}^{M-1} p^{t_m} (x^{t_m} - x^{t_{m+1}}) }{p^{t_M}-\tilde{p}}. 
\end{align*}

We now show that $LB(\tilde{p}^0) \subseteq  LB(\tilde{p}^1)$. If $\tilde{x} \in LB(\tilde{p}^0)$, then for every sequence $\{t_m\}_{m=1}^M=\sigma \in \Sigma^{-1}$ with $\tilde{p}^0<\tilde{p}^1<p^{t_M}$ it follows that  
\begin{align*}
\tilde{x} &\ge x^{t_1} - \frac{(M+1)\eps + p^{t_M}x^{t_{1}} - p^{t_M} x^{t_M} - \sum_{m = 1}^{M-1} p^{t_m} (x^{t_m} - x^{t_{m+1}}) }{p^{t_M}-\tilde{p}^0} \\
    &\ge x^{t_1} - \frac{(M+1)\eps + p^{t_M}x^{t_{1}} - p^{t_M} x^{t_M} - \sum_{m = 1}^{M-1} p^{t_m} (x^{t_m} - x^{t_{m+1}}) }{p^{t_M}-\tilde{p}^1}
\end{align*}
since the term being subtracted weakly increases when dividing by a smaller difference since $0<p^{t_M}-\tilde{p}^0<p^{t_M}-\tilde{p}^1$. (Recall the numerator in each fraction is positive.) When the sequence $\sigma \in \Sigma^{-1}$ has $\tilde{p}^{1}=p^{t_M}$ there is no restriction on counterfactual demands. Since $LB(\tilde{p}^1)$ only is restricted from sequences in $\Sigma^{-1}$ while $LB(\tilde{p}^0)$ is restricted by sequences in $\Sigma^{-1}$ and $\Sigma^0$, this shows $LB(\tilde{p}^0) \subseteq  LB(\tilde{p}^1)$. This also shows that $\underline{x}(\tilde{p}^1,\eps) \le \underline{x}(\tilde{p}^0,\eps)$ since $\underline{x}$ is a minimum, the constraint set on the lower bounds  $LB(\tilde{p}^0) \subseteq  LB(\tilde{p}^1)$, and a minimum over a smaller set is weakly larger.

\end{proof}

%Miscellaneous stuff that was in the proof
\iffalse
Similarly, lower-bounds are given by
\begin{equation}\label{eq:counterub} LB(k,\tilde{p}) = \left\{ \tilde{x}_k \mid \forall \; r \in \{1,\ldots, T\} \quad  \tilde{x}_k \ge \frac{\tilde{u}}{p_k^r}-\frac{u^r}{p_k^r}-\frac{p^{r}_{-k}\tilde{x}_{-k}}{p_k^r} + \frac{p^rx^r}{p_k^r} + \frac{\eps}{p^r_k}  \right\} \end{equation}.
\[ \tilde{u} & \le u^r + p^r \cdot (\tilde{x}-x^r) + \varepsilon \quad \text{for all}\; r \in \{1,\ldots, T\} \]
\fi

\subsection{Proofs for Section~\ref{sec:welfare}}

\subsubsection{Proofs for Section~\ref{sec:utilitybounds}}

Propositions~\ref{prop:utilitylp} and~\ref{prop:ubounds} are proven together in the following result.
\begin{prop} \label{prop:lputility}
For a dataset $\{ (x^t, p^t) \}_{ t = 1}^T$, let $\eps \geq \eps^*$.

\begin{enumerate}[i.]
\item If $\xb$ is in the dataset, i.e. $\xb = x^S$ for some $S \in \{1, \ldots, T \}$, and $\xa \neq \xb$, then 
\begin{align*}
\overline{u}(\xa, \xb, \eps) & = \max_{ \substack{ u^1, \ldots, u^T, \tilde{u}^{1} \in \mathbb{R}_+}} \quad \tilde{u}^{1} - u^S \\
\text{s.t.}&\quad u^s \le u^r + p^r \cdot (x^s-x^r) + \varepsilon \quad \text{for all}\; r,s \in \{1,\ldots, T\} \\
                & \quad \tilde{u}^{1} \le u^r + p^r \cdot (\xa-x^r) + \varepsilon \quad \text{for all}\; r \in \{1,\ldots, T\} \\ 
                & \quad \tilde{u}^{1} = u^r \quad \textit{for all}\; r \in \{1, \ldots, T \} \; \textit{ with } \xa = x^r.
\end{align*}

\item If $\xb = x^S$ is in the dataset and $\xa \neq \xb$, then the upper bound is equivalently given by \[
\overline{u}(\xa, \xb, \eps) = \min_{\sigma \in \Sigma_{S}} \left\{ p^{\sigma(M)}\cdot(\xa-x^{\sigma(M)}) + \sum_{m=1}^{M-1}p^{\sigma(m)}\cdot(x^{\sigma(m+1)}-x^{\sigma(m)}) + M\eps \right\},
\] 
where $\Sigma_S$ is the set of sequences that start with $\sigma(1)=S$, have no cycles, and have length at least $M\ge 1$. 

\item If $\xb = x^S$ is in the dataset, the function $\overline{u}$ is strictly increasing and continuous in $(\xa,\eps)$ over the region that excludes $\xa = \xb$. In particular, under Assumption~\ref{assm:maint} the bound cannot be improved.

\item If $\xa$ is in the dataset, i.e. $\xa = x^F$ for some $F \in \{1, \ldots, T \}$, and $\xa \neq \xb$, then 
\begin{align*}
\overline{u}(\xa, \xb, \eps) & = \min_{ \substack{ u^1, \ldots, u^T, \tilde{u}^{0} \in \mathbb{R}_+}} \quad u^{F} - \tilde{u}^{0} \\
\text{s.t.}&\quad u^s \le u^r + p^r \cdot (x^s-x^r) + \varepsilon \quad \text{for all}\; r,s \in \{1,\ldots, T\} \\
                &\quad \tilde{u}^{0} \le u^r + p^r \cdot (\xb-x^r) + \varepsilon \quad \text{for all}\; r \in \{1,\ldots, T\} \\
                & \quad \tilde{u}^{0} = u^r \quad \textit{for all}\; r \in \{1, \ldots, T \} \textit{ with } \xb = x^r.
\end{align*}

\item 
If $\xa = x^F$ is in the dataset and $\xa \neq \xb$, then the lower bound is equivalently given by
\[
\underline{u}(\xa, \xb, \eps) = \max_{ \sigma \in \Sigma_{F} } \left\{ p^{\sigma(M)}\cdot(x^{\sigma(M)} - \xb) + \sum_{m=1}^{M-1} p^{\sigma(m)}\cdot(x^{\sigma(m)}-x^{\sigma(m+1)}) - M\eps \right\}.
\]

\item If $\xa = x^F$ is in the dataset, the function $\underline{u}$ is strictly decreasing and continuous in $(\xb,\eps)$ over the region that excludes $\xb = \xa$. In particular, under Assumption~\ref{assm:maint} the bound cannot be improved.
\end{enumerate}
\end{prop}
Parts (i) and (iv) describe the linear programs used for computation and stated as Proposition~\ref{prop:utilitylp} in the main text. Note that parts (ii) and (v) show that the bounds are finite, as claimed in Proposition~\ref{prop:utilitylp}. The other parts cover Proposition~\ref{prop:ubounds} stated in the main text. Parts (ii) and (v) provide analytical characterizations of the bounds on utility differences. Parts (iii) and (vi) describe shape restrictions of the bounds.

\begin{proof}[Proof of Proposition~\ref{prop:lputility}]
We first prove parts (i) and (ii).

The definition of $\eps$-rationalizability yields
\begin{align*}
\overline{u}(\xa, x^S, \eps) & \leq \sup_{ \substack{ u^1, \ldots, u^T, \tilde{u}^{1} \in \mathbb{R}_+}} \quad \tilde{u}^{1} - u^S \\
\text{s.t.}&\quad u^s \le u^r + p^r \cdot (x^s-x^r) + \varepsilon \quad \text{for all}\; r,s \in \{1,\ldots, T\} \\
                &\quad \tilde{u}^{1} \le u^r + p^r \cdot (\xa-x^r) + \varepsilon \quad \text{for all}\; r \in \{1,\ldots, T\} \\
                & \quad \tilde{u}^{1} = u^r \quad \textit{for all}\; r \in \{1, \ldots, T \} \textit{ with } \xa = x^r.
\end{align*}
We shall show the opposite inequality holds to prove (i), and in doing so characterize the maximum as stated in part (ii). First, note that the problem on the right hand side is feasible since for the dataset $D = \{ (x^t, p^t) \}_{t = 1}^T$, we assumed $\eps\ge \eps^*(D)$. We show that there is a utility function $\tilde{u}$ such that for any $u^1, \ldots, u^T, \tilde{u}^{1} \in \mathbb{R}_+$ that are feasible,
\[
\tilde{u}^{1} - u^S \leq \tilde{u}(\xa) - \tilde{u}(x^S).
\]

To that end, first consider feasible values $u^1, \ldots, u^T, \tilde{u}^{1}$. For any sequence that begins at $\sigma(1) = S$, we can sum up the inequalities in the program to obtain
\[
\tilde{u}^{1} - u^S \leq p^{\sigma(M)}\cdot (\xa-x^{\sigma(M)}) + \sum_{m=1}^{M-1}p^{\sigma(m)}\cdot(x^{\sigma(m+1)}-x^{\sigma(m)}) + M\eps.
\]
Thus,
\[
\tilde{u}^{1} - u^S \leq \min_{\sigma \in \Sigma_S} \left\{ p^{\sigma(M)}\cdot(\xa-x^{\sigma(M)}) + \sum_{m=1}^{M-1}p^{\sigma(m)}\cdot(x^{\sigma(m+1)}-x^{\sigma(m)}) + M\eps \right \},
\]
where $\Sigma_S$ is the set of sequences with $\sigma(1)=S$, have no cycles, and have length at least $M\ge 1$.  We show in particular that provided $\xa \neq x^S$, the upper bound on the right hand side can be attained by the utility function $\tilde{u}$, defined for $x \neq x^S$ by
\[
\tilde{u}(x)= \min_{\sigma \in \Sigma_{S}} \left\{ p^{\sigma(M)}\cdot(x-x^{\sigma(M)}) + \sum_{m=1}^{M-1}p^{\sigma(m)}\cdot(x^{\sigma(m+1)}-x^{\sigma(m)}) + M\eps \right\},
\] 
and defined for $x^S$ by  $\tilde{u}(x^S) = 0$. Note that the summation on the right side defining $\tilde{u}(x)$ is zero whenever $M = 1$ because it is a summation over an empty set of indices. Note that $\tilde{u}$ is not continuous at $x^S$, which is key for our arguments.

We show that this utility function rationalizes the data. For any $x \in \mathbb{R}_{+}^K$, it follows that for any $t \in \{1,\ldots, T\}$ such that $x^t \neq x^S$,
\begin{align*}
    \tilde{u}(x)-p^t\cdot x & \leq p^t \cdot (x - x^t) + p^{\sigma^{*,t}(M^{*,t})} \cdot (x^t - x^{\sigma^{*,t}(M^{*,t})}) + \\
    & \quad \sum_{m=1}^{M^{*,t}-1} p^{\sigma^{*,t}(m)}\cdot(x^{\sigma^{*,t}(m+1)}-x^{\sigma^{*,t}(m)}) + (M^{*,t}+1)\eps-p^t \cdot x \\
    &= p^{\sigma^{*,t}(M^{*,t})} \cdot (x^t - x^{\sigma^{*,t}(M^{*,t})}) + \\
    & \quad \sum_{m=1}^{M^{*,t}-1} p^{\sigma^{*,t}(m)}\cdot(x^{\sigma^{*,t}(m+1)}-x^{\sigma^{*,t}(m)}) + (M^{*,t}+1)\eps-p^t \cdot x^t \\
    &= \tilde{u}(x^t)-p^t \cdot x^t+\eps
\end{align*}
where $\sigma^{*,t}\in \Sigma_{S}$ is a sequence that obtains the minimum of $\tilde{u}(x^t)$ and $M^{*,t}$ is the length of that sequence. 

Lastly, consider the observation $S \in \{1,\ldots, T\}$. For any $x \in \mathbb{R}_+^K$, it follows that
\begin{align*}
    \tilde{u}(x)-p^S \cdot x &\le p^S \cdot(x-x^S) + \eps - p^S \cdot x \\
    & = \tilde{u}(x^S)-p^S\cdot x^S + \eps
\end{align*}
where the inequality follows by looking at the sequence length one which only has observation $S$ and the equality follows since $\tilde{u}(x^S)=0$.

This utility function gives
\begin{align*}
\tilde{u}(\xa)-\tilde{u}(x^S) & =\min_{\sigma \in \Sigma_{S}} \left\{ p^{\sigma(M)}\cdot(\xa-x^{\sigma(M)}) + \sum_{m=1}^{M-1}p^{\sigma(m)}\cdot(x^{\sigma(m+1)}-x^{\sigma(m)}) + M\eps \right\} \\
& \leq \overline{u}(\xa, x^S, \eps).
\end{align*}
The inequality holds because $\tilde{u}$ $\eps$-rationalizes the dataset. The first part of the proof of the proposition established $ \overline{u}(\xa, x^S, \eps) \leq \tilde{u}(\xa)-\tilde{u}(x^S)$. This proves part (ii).

To prove part (i), note that we can use the function $\tilde{u}(x)$ to generate utility numbers that satisfy the inequality and equality conditions in the linear programming formulation. Indeed, set $u^t=\tilde{u}(x^t)-\min_{x \in \{x^s\}_{s=1}^T \cup \tilde{x}^1} \{ \tilde{u}(x) \}$ for $t \in \{1, \ldots, T\}$ and $\tilde{u}^1=\tilde{u}(\tilde{x}^1)-\min_{x \in \{x^s\}_{s=1}^T \cup \tilde{x}^1} \{ \tilde{u}(x) \}$. We subtract the minimum because $\tilde{u}(x)$ can be negative. Here the values $u^t$ and $\tilde{u}^1$ are weakly positive and satisfy the inequalities.

To prove (iii), recall that the upper bound is the minimum of finitely many affine functions as shown in (ii), except at $\xb = x^S$ when $\eps > 0$. Each such function is strictly increasing and continuous in $(\xb,\eps)$. From this we conclude that $\overline{u}$ is strictly increasing and continuous in $(\xb,\eps)$ (except at $\xb = x^S$ when $\eps > 0$). Thus, taking the minimum over $\eps$ subject to the constraint that the bound is defined, we see that under Assumption~\ref{assm:prime} the bound cannot be tightened. Note that here we use that for the special case $\xa = \xb$, $\overline{u}(\xa,\xb,\eps) = 0$, which is also the tightest possible under Assumption~\ref{assm:prime}.

The proofs for (iv)-(vi) are analogous since $\overline{u}(\xa,\xb,\eps) = -\underline{u}(\xb,\xa,\eps)$, and are omitted.
\end{proof}

\begin{proof}[Proof of Proposition~\ref{prop:unboundedu}]
First suppose $\xb$ is not in the dataset. Since $\eps \geq \eps^*$ there is some utility function $u$ that $\eps$-rationalizes the dataset. If we modify the utility function to make $u(\xb)$ arbitrarily negative, the modified function still rationalizes the dataset. Note that this modified function satisfies local nonsatiation, but is not (globally) strictly increasing or concave. This proves $\overline{u}(\xb,\xa,\eps) = \infty$.

Now instead suppose $\xa$ is not in the dataset. For any utility function that $\eps$-rationalizes the dataset, we can modify $u(\xa)$ to be arbitrariliy negative. Such a modified utility function still rationalizes the dataset, and so $\underline{u}(\xb,\xa,\eps) = -\infty$.
\end{proof}

\subsection{Proofs for Section~\ref{sec:indirectbounds}}

We prove a stronger and more formal version of Proposition~\ref{prop:welfarecomp}, explicitly describing a tractable linear program. To state this result, relative to the main text we use argument $\tilde{p}^{T+2}$ in place of $\pa$ and $\tilde{p}^{T+1}$ in place of $\pb$. We use this notation because in the proofs it is helpful to think of these as extra observations relative to a dataset of $T$ observations.

\begin{prop}\label{prop:lpwelfare}
For a dataset $\{ (x^t, p^t) \}_{ t = 1}^T$, let $\eps \geq \eps^*$. Whenever $\overline{V}(\tilde{p}^{T+2},\tilde{p}^{T+1},\eps)$ is finite, the change in the approximate indirect utility can be bounded by the following linear program:
\begin{align*}
\overline{V}(\tilde{p}^{T+2},\tilde{p}^{T+1},\eps) = & \max_{ \substack{ \tilde{x}^{T+1},\tilde{x}^{T+2}\in \mathbb{R}_+^K \\ u^1, \ldots, u^T, \tilde{u}^{T+1}, \tilde{u}^{T+2} \in \mathbb{R}_+}} \quad \tilde{u}^{T+2} - \tilde{p}^{T+2} \cdot \tilde{x}^{T+2}  - \tilde{u}^{T+1} + \tilde{p}^{T+1} \cdot \tilde{x}^{T+1}  \\
\text{s.t.}&\quad u^s \le u^r + p^r \cdot (x^s-x^r) + \varepsilon \quad \text{for all}\; r,s \in \{1,\ldots, T\} \\
                &\quad \tilde{u}^{T+1} \le u^r + p^r \cdot (\tilde{x}^{T+1}-x^r) + \varepsilon \quad \text{for all}\; r \in \{1,\ldots, T\} \\
                &\quad \tilde{u}^{T+2} \le u^r + p^r \cdot (\tilde{x}^{T+2}-x^r) + \varepsilon \quad \text{for all}\; r \in \{1,\ldots, T\}\\
                &\quad u^{r} \le \tilde{u}^{T+1} + \tilde{p}^{T+1} \cdot (x^r-\tilde{x}^{T+1}) + \varepsilon \quad \text{for all}\; r \in \{1,\ldots, T\}\\
                 &\quad u^{r} \le \tilde{u}^{T+2} + \tilde{p}^{T+2} \cdot (x^r-\tilde{x}^{T+2}) + \varepsilon \quad \text{for all}\; r \in \{1,\ldots, T\}\\
                &\quad \tilde{u}^{T+1} \le \tilde{u}^{T+2} + \tilde{p}^{T+2} \cdot (\tilde{x}^{T+1}-\tilde{x}^{T+2}) + \varepsilon \\
                &\quad \tilde{u}^{T+2} \le \tilde{u}^{T+1} + \tilde{p}^{T+1} \cdot (\tilde{x}^{T+2}-\tilde{x}^{T+1}) + \varepsilon.
\end{align*}
Moreover, when $\underline{V}(\tilde{p}^{T+2},\tilde{p}^{T+1},\eps)$ is finite it is the minimum of the same problem.

Under Assumption~\ref{assm:prime} ($\eps = \eps^*$), these bounds cannot be improved.
\end{prop}
Note that we do not impose the constraint that if $\tilde{p}^{T+2} = p^t$ for some $t$, then $\tilde{x}^{T+2} = x^t$. This is because  the observed demand $x^{t}$ is not known to \textit{exactly} maximize utility at the price $p^t$ so we must account for the fact that $\tilde{x}^{T+2}$ can differ.

\begin{proof}[Proof of Proposition~\ref{prop:lpwelfare} (and Proposition~\ref{prop:welfarecomp})]

Recall
\[
\overline{V}(\tilde{p}^{T+2},\tilde{p}^{T+1},\eps) = \sup_{ \left\{u \mid u \text{ } \eps-\text{rationalizes } \{ (x^t, p^t ) \}_{t = 1}^T \right\}}  \left\{ \overline{V}_{u,A}(\tilde{p}^{T+2},\eps) - \underline{V}_{u,A}(\tilde{p}^{T+1},\eps) \right\}.
\]

For a utility function $u$, price $\tilde{p}$, and bound on approximate optimization given by $\eps$, let the set of approximate optimizers be given by
\[
AO_u(\tilde{p},\eps) = \{ x \in \mathbb{R}^K_{+} \mid u(x) - p \cdot x \geq V_u(p) - \eps \}.
\]
We can write
\begin{align*}
\overline{V} (\tilde{p}^{T+2}  , \tilde{p}^{T+1},\eps) = \sup_{ \left\{u \mid u \text{ } \eps-\text{rationalizes } \{ (x^t, p^t ) \}_{t = 1}^T \right\}}  & \Bigg\{ \sup_{ \substack{\tilde{x}^{T+2} \in AO_{u}(\tilde{p}^{T+2},\eps)} } \left( u(\tilde{x}^{T+2}) - \tilde{p}^{T+2} \cdot \tilde{x}^{T+2} \right) \\
& - \inf_{\substack{\tilde{x}^{T+1} \in AO_{u}(\tilde{p}^{T+1},\eps)} } \left(u(\tilde{x}^{T+1}) - \tilde{p}^{T+1} \cdot \tilde{x}^{T+1} \right) \Bigg\}.
\end{align*}
We can write the difference as
\begin{align*}
\overline{V}(\tilde{p}^{T+2},\tilde{p}^{T+1},\eps) \leq \sup_{ \substack{u \\ \tilde{x}^{T+1},\tilde{x}^{T+2}\in \mathbb{R}_+^K} } & u(\tilde{x}^{T+2}) - \tilde{p}^{T+2} \cdot \tilde{x}^{T+2} - u(\tilde{x}^{T+1}) + \tilde{p}^{T+1} \cdot \tilde{x}^{T+1} \\
\text{s.t. } & u(x^t) - p^t \cdot x^t \geq \sup_{x \in \mathbb{R}^K_+} u(x) - p^t \cdot x - \eps \quad \forall t \in \{1, \ldots, T \} \\
& u(\tilde{x}^{T+2}) - \tilde{p}^{T+2} \cdot \tilde{x}^{T+2} \geq \sup_{x \in \mathbb{R}^K_+} u(x) - \tilde{p}^{T+2} \cdot x - \eps \\
& u(\tilde{x}^{T+1}) - \tilde{p}^{T+1} \cdot \tilde{x}^{T+1} \geq \sup_{x \in \mathbb{R}^K_+} u(x) - \tilde{p}^{T+1} \cdot x - \eps.
\end{align*}
The first inequality constraint imposes the requirement that $u$ $\eps$-rationalizes the dataset. The second inequality constraint only involves the variables $\tilde{x}^{T+2}$, and so when we take a supremum this is the upper approximate indirect utility $\overline{V}_{u,A}(\tilde{p}^{T+2},\eps) = V_u(\tilde{p}^{T+2})$. The third inequality constraint has infimum (over $\tilde{x}^{T+1}$) at the lower approximate indirect utility $\underline{V}_{u,A}(\tilde{p}^{T+1},\eps)$. 

Consider the feasibility region of this problem. Checking the inequalities for all $x$ is weakly more restrictive than checking for $x \in \{ x^1, \ldots, x^T, \tilde{x}^{T+1}, \tilde{x}^{T+2} \}$. Thus, we will replace the suprema over all $x \in \mathbb{R}^K_{+}$ with a finite collection of inequalities involving $\{ x^1, \ldots, x^T, \tilde{x}^{T+1}, \tilde{x}^{T+2} \}$.  In addition, searching over all utility functions to satisfy these inequalities is weakly more restrictive than searching over all utility numbers. From these two monotonicity observations, and the fact that the value function is monotone in its feasibility region (with regard to set inclusion), we obtain 
\begin{align*}
\overline{V}(\pa,\pb,\eps) \leq & \sup_{ \substack{ \tilde{x}^{T+1},\tilde{x}^{T+2}\in \mathbb{R}_+^K \\ u^1, \ldots, u^T, \tilde{u}^{T+1}, \tilde{u}^{T+2} \in \mathbb{R}_+}} \quad \tilde{u}^{T+2} - \tilde{p}^{T+2} \cdot \tilde{x}^{T+2}  - \tilde{u}^{T+1} + \tilde{p}^{T+1} \cdot \tilde{x}^{T+1}  \\
\text{s.t.}&\quad u^s \le u^r + p^r \cdot (x^s-x^r) + \varepsilon \quad \text{for all}\; r,s \in \{1,\ldots, T\} \\
                &\quad \tilde{u}^{T+1} \le u^r + p^r \cdot (\tilde{x}^{T+1}-x^r) + \varepsilon \quad \text{for all}\; r \in \{1,\ldots, T\} \\
                &\quad \tilde{u}^{T+2} \le u^r + p^r \cdot (\tilde{x}^{T+2}-x^r) + \varepsilon \quad \text{for all}\; r \in \{1,\ldots, T\}\\
                &\quad u^{r} \le \tilde{u}^{T+1} + \tilde{p}^{T+1} \cdot (x^r-\tilde{x}^{T+1}) + \varepsilon \quad \text{for all}\; r \in \{1,\ldots, T\}\\
                 &\quad u^{r} \le \tilde{u}^{T+2} + \tilde{p}^{T+2} \cdot (x^r-\tilde{x}^{T+2}) + \varepsilon \quad \text{for all}\; r \in \{1,\ldots, T\}\\
                &\quad \tilde{u}^{T+1} \le \tilde{u}^{T+2} + \tilde{p}^{T+2} \cdot (\tilde{x}^{T+1}-\tilde{x}^{T+2}) + \varepsilon \\
                &\quad \tilde{u}^{T+2} \le \tilde{u}^{T+1} + \tilde{p}^{T+1} \cdot (\tilde{x}^{T+2}-\tilde{x}^{T+1}) + \varepsilon.
\end{align*}

We will now show the opposite inequality holds. First, recall that Proposition~\ref{prop:nonempty} shows that this program is feasible provided $\eps \geq \eps^*$. Let $u^1, \ldots, u^T, \tilde{u}^{T+1}, \tilde{u}^{T+2}$, $\tilde{x}^{T+1}$, and $\tilde{x}^{T+2}$ denote some values that are feasible. Construct the augmented dataset $\{ (x^t, p^t) \}_{t = 1}^{T+2}$ that has $(x^{T+1},p^{T+1}) = (\tilde{x}^{T+1},\tilde{p}^{T+1})$ and $(x^{T+2},p^{T+2}) = (\tilde{x}^{T+2},\tilde{p}^{T+2})$. Construct a utility function as
\[
\tilde{u}(x) = \min_{\sigma \in \Sigma_{T+1}} \left\{ p^{\sigma(M)}\cdot(x-x^{\sigma(M)}) + \sum_{m=1}^{M-1}p^{\sigma(m)}\cdot(x^{\sigma(m+1)}-x^{\sigma(m)}) + M\eps \right \},
\]
for $x \neq x^{T+1}$, where $\Sigma_{T+1}$ is the set of sequences in the augmented dataset that start with $\sigma(1)=T+1$, have no cycles, and have length at least $M\ge 1$. Finally, set $\tilde{u}(x^{T+1}) = 0$. The proof of Proposition~\ref{prop:lputility} shows that this function $\eps$-rationalizes the augmented dataset $\{ (x^t, p^t) \}_{t = 1}^{T+2}$. Moreover, Proposition~\ref{prop:lputility} also established
\begin{equation} \label{ubound}
\tilde{u}^{T+2} - \tilde{u}^{T+1} \leq \tilde{u}(x^{T+2}) - \tilde{u}(x^{T+1}).
\end{equation}
Recall we set $(x^{T+1},p^{T+1}) = (\tilde{x}^{T+1},\tilde{p}^{T+1})$ and $(x^{T+2},p^{T+2}) = (\tilde{x}^{T+2},\tilde{p}^{T+2})$ for hypotheticals. We conclude
\begin{align*}
u^{T+2} - p^{T+2} \cdot x^{T+2} - ( u^{T+1} & - p^{T+1} \cdot x^{T+1}) \\
& \leq \tilde{u}(x^{T+2}) - p^{T+2} \cdot x^{T+2}  - ( \tilde{u}(x^{T+1}) - p^{T+1} \cdot x^{T+1}) \\
& \leq \overline{V}_{\tilde{u},A} (p^{T+2},\eps) - \underline{V}_{\tilde{u},A} (p^{T+1},\eps) \\
& \leq \overline{V}(\tilde{p}^{T+2},\tilde{p}^{T+1},\eps).
\end{align*}
The first inequality uses (\ref{ubound}). The second inequality holds because for $\tilde{u}$, $x^{T+2}$ is an approximate optimizer given $p^{T+2}$ and similarly for $x^{T+1}$. The third inequality is the definition of the bounds on approximate indirect utility. Since this is true for any feasible values, we conclude that $\overline{V}(\tilde{p}^{T+2},\tilde{p}^{T+1},\eps)$ is obtained by the linear program described in the proposition. Recall that while we have used suprema throughout, in this last step since we have established $\overline{V}(\tilde{p}^{T+2},\tilde{p}^{T+1},\eps)$ as the (bounded) value of a linear program, we know the supremum is attained by Lemma~\ref{lem:lpcont}.

Finally, we note that the bounds are weakly monotone in $\eps$ by Proposition~\ref{prop:shapev}. For $\eps < \eps^*$ we know the bounds are not defined because no utility function $\eps$-rationalizes the dataset. Thus, the bounds are the tightest possible under Assumption~\ref{assm:prime} ($\eps = \eps^*$).
\end{proof}

\begin{proof}[Proof of Proposition~\ref{prop:welfaresandwich}]

Step 1 provides the upper bound on $\overline{V}$. Step 2 provides the lower bound on $\overline{V}$. Recall that $\pa=P^S$ for some $S \in \{1,\ldots,T\}$. 

\textbf{Step 1.}
Recall from Equation~\ref{eq:indineq} that for any $r \in \{1 , \ldots, T \}$ and $p \in \mathbb{R}^K_{++}$,
\[
V_u(p^r) - V_u(p) \leq x^r \cdot (p - p^r) + \eps.
\]
By summing up such inequalities over sequences, we obtain the upper bound
\[
V_u(p^S) - V_u(\pb) \leq \min_{\sigma \in \Sigma_S} \left \{ x^{\sigma(M)} \cdot (\pb - p^{\sigma(M)}) + \sum_{m = 1}^{M-1} x^{\sigma(m)} \cdot (p^{\sigma(m+1)} - p^{\sigma(m)}) + M \eps \right\}.
\]
Since
\[
\overline{V}_{u,A}(p^S,\eps) - \underline{V}_{u,A}(\pb,\eps) \le V_u(p^S) - V_u(\pb) + \eps
\]
by construction, we prove that
\[
\overline{V}(p^{S},\pb,\eps) \leq \min_{\sigma \in \Sigma_S} \left\{ x^{\sigma(M)} \cdot (\pb - p^{\sigma(M)}) + \sum_{m = 1}^{M-1} x^{\sigma(m)} \cdot (p^{\sigma(m+1)} - p^{\sigma(m)}) + M \eps \right\} + \eps.
\]

\textbf{Step 2.}
We now establish the lower bound for $\overline{V}(p^S,\pb,\eps)$. Define
\[
V'(p) = -\min_{\sigma \in \Sigma_S} \left \{ x^{\sigma(M)} \cdot (p - p^{\sigma(M)}) + \sum_{m = 1}^{M-1} x^{\sigma(m)} \cdot (p^{\sigma(m+1)} - p^{\sigma(m)}) + M \eps \right\}.
\]
Let $\underline{V}'$ denote the minimum of $V'(p)$ over the convex hull of $\{ p^t \}_{t = 1}^T$, which is attained and finite because $V'$ is continuous and the convex hull here is compact. Define $V(p) = \max \{ V'(p), \underline{V}' \}$.

First we show that $V$ satisfies a set of inequalities that are a dual version of $\eps$-rationalizability. Duality is considered in more detail in the Supplemental Appendix~\ref{supp:duality}, and we use several results from that Appendix.

For some $t \in \{1, \ldots, T \}$, let $\sigma^{*,t} \in \Sigma_S$ be a sequence that obtains the minimum of $V'(p^t)$, and let $M^{*,t}$ be the length of that sequence. We have
\begin{align*}
-V'(p) - x^t \cdot p & \leq x^t \cdot (p - p^t) + x^{\sigma^{*,t}(M^{*,t})} \cdot (p^t - p^{\sigma^{*,t}(M^{*,t})}) + \\
& \sum_{m = 1}^{M^{*,t} - 1} p^{\sigma^{*,t}(m)} \cdot (x^{\sigma^{*,t}(m + 1)} - x^{\sigma^{*,t}(m)}) + (M^{*,t} + 1) \eps - x^t \cdot p \\
& = -V'(p^t) - x^t \cdot p^t + \eps.
\end{align*}
Recall that $V(p) \geq V'(p)$ and $V(p^t) = V'(p^t)$ for $t \in \{1, \ldots, T\}$. This implies
\begin{equation} \label{eq:vsubg}
V(p) \geq V(p^t) - x^t \cdot (p - p^t) - \eps
\end{equation}
for any $t \in \{1, \ldots, T\}$ and any $p \in \mathbb{R}^K$. 

Define $u_V : \mathbb{R}^K \rightarrow \mathbb{R} \cup \{ - \infty \}$ by $u_V(x) = \inf_{p \in \mathbb{R}^K_{+}} V(p) + p \cdot x$. Proposition~\ref{slem:duality}(ii) shows that $u_V$ $\eps$-rationalizes the dataset. Note that since $u_V(x) \geq \underline{V'}$, $u_V$ is everywhere finite.

The function $V$ is the maximum of finitely many affine functions, each weakly decreasing in $p$, and is hence continuous, weakly decreasing, and convex. We conclude from Lemma~\ref{slem:doubleconj} that $V = V_{u_V}$. In particular, $V$ is the indirect utility function for $u_V$ and we have established that $u_V$ $\eps$-rationalizes the dataset. We conclude that 
\begin{equation} \label{eq:lowerboundv}
V(p^S) - V(\pb) \leq \overline{V}_{u_V,A}(p^S,\eps) - \underline{V}_{u_V,A}(\pb,\eps)\leq \overline{V}(p^S,\pb,\eps).
\end{equation}

We now characterize $V(p^S) - V(\pb)$ to state the lower bound on the proposition. Since $p^S$ is in the convex hull of prices, $V(p^S) = V'(p^S)$. Note that $V'(p^S) \leq 0$ because the dataset is $\eps$-rationalized by quasilinear utility (see Lemma~\ref{supp:misc}); this relies on the fact that the sum of each sequence defining $V'$ makes a cycle because it begins and ends at $p^S$. In addition, by considering a sequence of length $1$, $V'(p^S) \geq - \left( x^{S} \cdot (p^S - p^S) + \eps \right) = -\eps$. Thus,
\begin{align*}
V(p^S) & - V(\pb) \geq -\eps - V(\pb) \\ 
& = \min_{\sigma \in \Sigma_S} \left \{ x^{\sigma(M)} \cdot (\pb - p^{\sigma(M)}) + \sum_{m = 1}^{M-1} x^{\sigma(m)} \cdot (p^{\sigma(m+1)} - p^{\sigma(m)}) + M \eps \right\} - \eps.
\end{align*}
So from (\ref{eq:lowerboundv}),
\begin{align*}
\overline{V}(p^S,\pb,\eps) & \geq \min_{\sigma \in \Sigma_S} \left \{ x^{\sigma(M)} \cdot (\pb - p^{\sigma(M)}) + \sum_{m = 1}^{M-1} x^{\sigma(m)} \cdot (p^{\sigma(m+1)} - p^{\sigma(m)}) + M \eps \right\} - \eps \\
& = h(\tilde{p}^0) - \eps,
\end{align*}
establishing the lower bound.

It is worth noting that the arguments above rely on the inequality $-\eps \leq V(p^S) \leq 0$. If we instead rely on the true value for $V(p^S)$ we actually prove the stronger result here that
\begin{align*}
\overline{V}(p^S,\pb,\eps) & \geq \min_{\sigma \in \Sigma_S} \left \{ x^{\sigma(M)} \cdot (\pb - p^{\sigma(M)}) + \sum_{m = 1}^{M-1} x^{\sigma(m)} \cdot (p^{\sigma(m+1)} - p^{\sigma(m)}) + M \eps \right\} \\
& -\min_{\sigma \in \Sigma_S} \left \{ x^{\sigma(M)} \cdot (p^S - p^{\sigma(M)}) + \sum_{m = 1}^{M-1} x^{\sigma(m)} \cdot (p^{\sigma(m+1)} - p^{\sigma(m)}) + M \eps \right\} \\
& = h(\pb) \\
& -\min_{\sigma \in \Sigma_S} \left \{ x^{\sigma(M)} \cdot (p^S - p^{\sigma(M)}) + \sum_{m = 1}^{M-1} x^{\sigma(m)} \cdot (p^{\sigma(m+1)} - p^{\sigma(m)}) + M \eps \right\}.
\end{align*}
Thus, the lower bound $h(\pb)$ stated in the proposition can be tightened a bit.
\end{proof}

\begin{proof}[Proof of Proposition~\ref{prop:intformula}]

Recall for this proof we consider when $K=1$ so quantities and prices are scalar. We suppress dot product notation for this proof. We first show that when $\pa > \min \{p^1, \ldots, p^T, \pb\}$,
\[
\overline{V}(\pa,\pb,0) \leq \int^1_0 \overline{x}_1 (t \pa + (1 - t) \pb, 0) \left(\pb - \pa \right)  dt.
\]
For any utility $u$ function that admits maximizers over $p > \min \{p^1, \ldots, p^T, \pb\}$, write
\[
x_u(p) \in \argmax_{x \in \mathbb{R}_{+}} u(x) - p x
\]
for some selector from the argmax correspondence. We first argue
\[
\overline{V}(\pa,\pb,0) = \sup_{ \left\{u \in \mathcal{U} \mid u \text{ } \text{rationalizes } \{ (x^t, p^t ) \}_{t = 1}^T \right\} } \int x_{u} \left(t \pa + (1 - t) \pb \right) \left (\pb - \pa \right) dt
\]
where the supremum is over the set of utility functions $\mathcal{U}$ that admit a maximizer over $p > \min \{p^1, \ldots, p^T, \pb\}$.

To that end, note that $\pa \in \CCo ( \{ p^t \}_{t = 1}^T \cup \pb )$. From Proposition~\ref{prop:shapev}, $\overline{V}(\pa,\pb,0)$ can be written as a supremum of differences in indirect utility, where each indirect utility $V_u$ is finite for $p \geq \{p_1, \ldots, p_T, \pb\}$. Each $V_u$ is convex, and for such functions, the subdifferential
\[
\partial V_{u}(p) = \left\{ x \in \mathbb{R}_+ \mid V_{u}(\tilde{p}) \geq V_{u} (p) + x (\tilde{p} - p) \qquad \forall \tilde{p} \in \mathbb{R}_{++} \right\} 
\]
is nonempty for any $p > \min \{p^1, \ldots, p^T, \pb\}$. See the proof of Lemma~\ref{lem:existence}. We can construct a function $\tilde{x}_{V_u}$ by selecting from the subdifferential, so that for each $p > \min \{p_1, \ldots, p_T, \pb\}$
\[
\tilde{x}_{V_u}(p) \in \partial V_u(p).
\]
Any such function satisfies the formula
\begin{equation} \label{eq:consumer}
V_u(\pa) - V_u(\pb) = \int^1_0 \tilde{x}_{V_u} \left(t \pa + (1 - t) \pb \right) \left (\pb - \pa \right) dt.
\end{equation}
See for example \cite{rockafellar2015convex}, Corollary 24.2.1, or \cite{chambers2017characterization}, Theorem 2. (The selector $\tilde{x}_{V_u}(p)$ is always Reimann integrable.)

We note that for any $V_u$ such that $u$ $\eps$-rationalizes the dataset, the utility function $u_{V_u}(x) = \inf_{p \in \mathbb{R}^K} V(p) + p x$ also $\eps$-rationalizes the dataset from Lemma~\ref{slem:duality}. Moreover, $u_{V_u}$ is concave, weakly increasing, and satisfies $V_{u_{V_u}} = V_u$ from Lemmas~\ref{slem:shapeconjugate} and~\ref{slem:doubleconj}. From this and the proof of Lemma~\ref{lem:existence} we conclude that $\tilde{x}_{V_u}$ is a maximizer of $u_{V_u}$. Summing up, we conclude that it is without loss of generality to consider utility functions $\mathcal{U}$ that induce a maximizer for all $p > \min\{p^1, \ldots, p^T, \tilde{p}^0\}$..

Putting these arguments together, we conclude that
\begin{align*}
\overline{V}(\pa,\pb,0) & = \sup_{ \left\{u \mid u \text{ } \text{rationalizes } \{ (x^t, p^t ) \}_{t = 1}^T \right\} } \left\{ \overline{V}_{u,A}(\pa,0) - \underline{V}_{u,A}(\pb,0) \right\} \\
& = \sup_{ \left\{u \mid u \text{ } \text{rationalizes } \{ (x^t, p^t ) \}_{t = 1}^T \right\} } \left\{ V_u(\pa) - V_u(\pb) \right\} \\
& = \sup_{ \left\{u \in \mathcal{U} \mid u \text{ } \text{rationalizes } \{ (x^t, p^t ) \}_{t = 1}^T \right\} } \left\{ V_u(\pa) - V_u(\pb) \right\} \\
& = \sup_{ \left\{u \in \mathcal{U} \mid u \text{ } \text{rationalizes } \{ (x^t, p^t ) \}_{t = 1}^T \right\} } \int^1_0 x_{u} \left(t \pa + (1 - t) \pb \right) \left (\pb - \pa \right) dt \\
& \leq \int^1_0 \overline{x}_1 (t \pa + (1 - t) \pb, 0) \left(\pb - \pa \right)  dt.
\end{align*}
The first equality is the definition. The second equality uses the fact that when $\eps = 0$, the approximate indirect utilities equal the indirect utility. The third equality uses the arguments above to conclude it is without loss of generality to consider utility functions that induce a maximizer for $p > \min \{p^1, \ldots, p^T, \pb \}$. The fourth equality uses (\ref{eq:consumer}). The first inequality uses the fact that $\overline{x}_1$ is a pointwise maximizer of demand induced by $u$ that are $\eps$-rationalized by the dataset.

Now it remains to show the opposite inequality. We first show that $\overline{x}_1$ is the demand induced by some quasilinear utility function. When $K = 1$ and $\eps = 0$, for $p > \min\{p^1, \ldots, p^T, \pb\}$, it follows that $\overline{x}_1(p)$ is finite from Proposition~\ref{prop:Xfinite}. Consider the sets $E^1 = \{ (\overline{x}_1 (p),p) \}_{p > \min\{p^1, \ldots, p^T, \pb \}}$ and $E^2 = \{ (x^t, p^t) \}_{t = 1}^T$. We argue that for $e = (e_x, e_p) \in \{E^1 \cup E^2 \}$ and $e' = (e'_x, e'_p)\in \{ E^1 \cup E^2 \}$, the inequality
\[
(e_x - e'_x)(e_p - e'_p) \leq 0
\]
holds. We argue by cases. Indeed, when $e,e' \in E^1$ this follows from Proposition~\ref{prop:singleprice}. When $e,e' \in E^2$, the inequality holds by Lemma~\ref{lem:epsqrat} and the fact that $\eps = 0$. When $e \in E^1$ and $e' \in E^2$, the result follows from Equation~\ref{eq:xhalfspace}. This covers all cases.

Now consider the set $E = E^1 \cup E^2$. From \cite{rockafellar2015convex}, p. 240 the set $E$ is the graph of a multivalued mapping that is cyclically monotononically decreasing. From \cite{rockafellar2015convex}, Theorem 24.3 there is some lower semicontinuous convex function $f$ such that for any $e = (e_x, e_p) \in E$,
\[
-e_x \in \partial f(e_p),
\]
where $\partial f$ denotes the subdifferential of $f$.\footnote{Note that we differ from the statement of Theorem 24.3 in \cite{rockafellar2015convex} because we consider cyclically monotonically decreasing mappings while that result considers increasing mappings; this is why we need to take a negative involving $e_x$.} Let $f^*$ denote the convex conjugate of $f$, which is defined in Appendix~\ref{supp:duality}. We conclude from \cite{rockafellar2015convex}, Theorem 23.5 that for each $e = (e_x,e_p) \in E$,
\[
-e_x \in \argmax_{x} x e_p - f^*(x).
\]
Since $e_x \geq 0$, we conclude via a change in variables that
\[
e_x \in \argmax_{x \geq 0} - x e_p - f^*(-x).
\]

We conclude that by setting $\tilde{u}(x) = -f^*(-x)$, we have for each $e \in E$, the price $e_p$ induces the demand $e_x$ as some exact maximizer of a quasilinear utility function. In particular, recalling $E = E^1 \cup E^2$ and using that $E^1$ is the graph of $\overline{x}_1$ over $p > \min\{p^1, \ldots, p^T, \pb\}$, we conclude that $\overline{x}_1$ is a demand function generated by quasilinear utility with $\tilde{u}$. Moreover, recall $\tilde{u}$ rationalizes the dataset. Thus,
\begin{align*}
\int \overline{x}_1 (t & \pa + (1 - t) \pb) dt \left (\pb - \pa \right) \\ & \leq \sup_{ \left\{u \in \mathcal{U} \mid u \text{ rationalizes } \{ (x^t, p^t ) \}_{t = 1}^T \right\}}\int^1_0 x_u \left(t \pa + (1 - t) \pb \right) \left (\pb - \pa \right) dt \\
& = \overline{V}(\pa,\pb,0).
\end{align*}

\end{proof}

\begin{proof}[Proof of Proposition~\ref{prop:shapev}]
To show shape restrictions on $\overline{V}(\pa,\pb,\eps)$, recall
\[
\overline{V}(\pa,\pb,\eps) = \sup_{ \left\{u \mid u \text{ } \eps-\text{rationalizes } \{ (x^t, p^t ) \}_{t = 1}^T \right\}}  \left\{ \overline{V}_{u,A}(\pa,\eps) - \underline{V}_{u,A}(\pb,\eps) \right\},
\]
where the supremum is over $u$ such that $ \underline{V}_{u,A}(\pb,\eps) \neq \infty$. For any $u$ that $\eps$-rationalizes the data, we can add or subtract a constant to $u$ and the new utility function rationalizes the data as well. In addition, if we let $u + a$ be the utility $u$ plus the constant $a$, then
\[
\overline{V}_{u,A}(\pa,\eps) - \underline{V}_{u,A}(\pb,\eps) = \overline{V}_{u + a,A}(\pa,\eps) - \underline{V}_{u + a,A}(\pb,\eps).
\]
Thus, it is thus without loss of generality to restrict $u$ such that $\underline{V}_{u,A}(\pb,\eps) = 0$. Now recall the upper approximate indirect utility satisfies $\overline{V}_{u,A}(\pa,\eps) = V_u(\pa)$. Combining these arguments we can write
\[
\overline{V}(\pa,\pb,\eps) = \sup_{ \left\{u \mid u \text{ } \eps-\text{rationalizes } \{ (x^t, p^t ) \}_{t = 1}^T \right\}}  \left\{ V_{u}(\pa) - 0 \right\},
\]
where the supremum is over $u$ such that $\underline{V}_{u,A}(\pb,\eps) = 0$. We know that each $V_u$ is convex, weakly decreasing, and lower semicontinuous by Lemma~\ref{slem:shapeconjugate}. Note that this is true for any $u$, regardless of whether it is concave or upper semicontinuous.

We conclude that when viewing $\overline{V}(\pa,\pb,\eps)$ only as a function of $\pa$, it is the supremum (over $u$) of convex, weakly decreasing, lower semicontinuous functions. It is therefore convex, weakly decreasing, and lower semicontinuous in $\pa$ from \cite{rockafellar2015convex}, Theorems 5.5 and 9.4.

To see that $\overline{V}(\pa,\pb,\eps)$ is weakly increasing in $\eps$, we recall the characterization as a linear program in Proposition~\ref{prop:lpwelfare}. The feasibility region of the linear program is weakly increasing (with regard to set inclusion) in $\eps$, and so the value function of the problem is weakly increasing in $\eps$.

We now establish the finiteness properties in the proposition. Let $\pa \in \CCo(\{p^t\}_{t = 1}^T)$. We show $\overline{V}(\pa,\pb,\eps)$ is finite. We can write
\[
\pa \geq \sum_{t = 1}^T \alpha_t p^t
\]
for some nonnegative $\alpha_1, \ldots, \alpha_T$ that sum to $1$ where the inequality holds componentwise. We have
\begin{align*}
V_u(\pa) & \leq V_u \left(\sum_{t = 1}^T \alpha_t p^t \right)  \\
& \leq \sum_{t = 1}^T \alpha_t V_u(p^t),
\end{align*}
where the first inequality follows because $V_u$ is weakly decreasing, and the second inequality follows because $V_u$ is convex. Recall from (\ref{eq:indineq}) in the main text that for any $u$ that $\eps$-rationalizes the dataset, we have
\[
\overline{V}_{u,A}(p^t,\eps) - \underline{V}_{u,A}(\pb,\eps) \leq V_u(p^t) - V_u(\pb) + \eps \leq (x^t \cdot (\pb - p^t) + \eps) + \eps
\]
for any $t \in \{1, \ldots, T \}$. Combining the previous steps we obtain
\[
\overline{V}(\pa,\pb,\eps) \leq \sum_{t = 1}^T \alpha_t x^t \cdot (\pb - p^t) + 2 \eps < \infty.
\]

Now we show that if $\pa \not\in \CCo(\{p^t\}_{t = 1}^T \cup \pb)$, then $\overline{V}(\pa,\pb,\eps) = \infty$. First note that from Proposition~\ref{prop:nonempty}, there is some utility function that $\eps$-rationalizes the dataset and has a maximizer at the price $\pb$. Let $\tilde{x}^0$ denote such a maximizer. Now construct the augmented dataset $\{ (x^t, p^t) \}_{t = 1}^{T+1}$, where $(x^{T+1},p^{T+1}) = (\tilde{x}^0, \pb)$. Note that by construction, there is some utility function $u$ that $\eps$-rationalizes the augmented dataset $\{ (x^t, p^t) \}_{t = 1}^{T+1}$. For any such function $u$, the indirect utility satisfies $V_u(p) < \infty$ for any $p \in \CCo(\{p^t\}_{t = 1}^T \cup \pb)$ from the finiteness arguments above. Thus, it remains to show that when $p \not\in \CCo(\{p^t\}_{t = 1}^T \cup \pb)$, there is some $u$ that $\eps$-rationalizes the augmented dataset and satisfies $V_u(p) = \infty$. That end, fix $(x^1, p^1) \in \{ (x^t, p^t) \}_{t=1}^{T+1}$ and let $\Sigma_1$ denote the set of finite sequences of $t \in \{1,\ldots,T\}$ with no cycles that begins at $\sigma(1) = 1$. Define
\[
U(x) = \min_{\sigma \in \Sigma_1} \left\{ p^{\sigma(M)} \cdot \left(x - x^{\sigma(M)}\right) +  \sum_{m=1}^{M-1} p^{\sigma(m)} \cdot (x^{\sigma(m+1)} - x^{\sigma(m)})  + M  \eps \right\},
\]
where $M$ corresponds to the length of a particular sequence. \cite{allen2020satisficing} have shown that for $\eps \geq \eps^*$, this function $\eps$-rationalizes the augmented dataset $\{ (x^t, p^t) \}_{t = 1}^{T+1}$.

Since $\pa \not\in \CCo(\{p^t\}_{t = 1}^T, \pb)$ from the separating hyerplane theorem, there is some $x \in \mathbb{R}^K$ with $x\neq0$ such that
\[
(p - \pa) \cdot x > 0 \qquad \text{ for all } p \in \CCo(\{p^t\}_{t = 1}^T \cup \pb).
\]
We argue by contradiction that $x$ contains no negative components. Indeed, suppose it does so that $x_k < 0$ for some $k$. Since $\pa \in \mathbb{R}^K_{++}$ and $\CCo(\{p^t\}_{t = 1}^T \cup \pb)$ is upper comprehensive, we can find some $p \in \CCo(\{p^t\}_{t = 1}^T \cup \pb)$ with $p_k$ high enough so that $(p - \pa) \cdot x < 0$. We reach a contradiction and conclude $x \in \mathbb{R}^K_{+}$ and $x \neq 0$.

Note that in the definition of $U(x)$, the minimum is taken over certain functions that involve $p^t \cdot x$ for some $t$, plus a constant. Thus, for shorthand write
\[
U(x) = \min_{t \in \{1, \ldots, T + 1 \}} \min_{a \in A^t} \{ p^t \cdot x + a \}
\]
for certain finite sets $A^t$ corresponding to the sums in the construction of $U(x)$. We conclude that
\begin{align*}
\lim_{\lambda \rightarrow \infty} U(\lambda x) - \pa \cdot \lambda x=
\lim_{\lambda \rightarrow \infty} \min_{t \in \{1, \ldots, T + 1 \}} \min_{a \in A^t}  \left\{ (p^t - \pa) \cdot \lambda x + a \right\} = \infty.
\end{align*}
This establishes that $V_U(\pa) = \infty$ and so since $V_U(\pb) < \infty$ we conclude $\overline{V}(\pa,\pb,\eps) \geq V_U(\pa) - V_U(\pb) = \infty$.

\end{proof}

\subsection{Proofs for Section~\ref{sec:shape}}

In order to prove Proposition~\ref{prop:convexcount}, we first prove a lemma. The lemma establishes a convexity property of the set of counterfactual quantities at a given price,
\[
X(\tilde{p},D,\eps) = \{ \tilde{x} \mid (\tilde{x}, \tilde{p}) \in C(D,\eps) \}.
\]
\begin{lemma} \label{lem:convexfeas}
Let $D^0 = \{ (d^{0,t}, p^t) \}_{t = 1}^T$ and $D^1 = \{ (d^{1,t}, p^t) \}_{t = 1}^T$ differ only for quantities. If $\tilde{d}^j \in X\left(\tilde{p},  D^j, \eps^j\right)$ for $j \in \{ 0 ,1 \}$, then for any $\alpha \in [0,1]$,
\[
\alpha \tilde{d}^0 + (1 - \alpha) \tilde{d}^1 \in X(\tilde{p}, \alpha D^0 + (1 - \alpha) D^1, \alpha \eps^0 + (1 - \alpha) \eps^1).
\]
In addition, the set $A$ described in Proposition~\ref{prop:convexcount} is convex.
\end{lemma}

\begin{proof}
The sets $C(D,\eps)$ and $X(\tilde{p},D,\eps)$ can be characterized by using any of the equivalent statements of Lemma~\ref{lem:epsqrat}, applied to the counterfactual-augmented dataset. In particular, for an arbitrary (hypothetical) dataset $D^j = \left\{ \left(d^{j,t}, p^t \right) \right\}_{t = 1}^T$, $\tilde{d}^j \in X\left(\tilde{p},  D^j, \eps^j\right)$ means there are numbers $u^{j,1}, \ldots, u^{j,K}, \tilde{u}^j \in \mathbb{R}_+$ such that
\begin{align*}
u^{j,s} & \le u^{j,r} + p^r \cdot \left(d^{j,s}-d^{j,r} \right) + \varepsilon &\quad \text{for all}\; r,s \in \{1,\ldots, T\} \\
u^j & \le u^{j,r} + p^r \cdot \left(\tilde{d}^j-d^{j,r} \right) + \varepsilon &\quad \text{for all}\; r \in \{1,\ldots, T\} \\
u^{j,r} & \le \tilde{u}^j + \tilde{p} \cdot \left(d^{j,r}-\tilde{d}^j \right) + \varepsilon &\quad \text{for all}\; r \in \{1,\ldots, T\}.
\end{align*}
In addition, $\tilde{d}^j$ must be non-negative. We can take a convex combination of the values for $j = 0$ and $j = 1$ and the inequalities are preserved. For example, considering the first inequalities that involve a pair $r,s$, we have
\begin{align*}
\alpha u^{0,s} & + (1 - \alpha) u^{1,s} \leq \alpha u^{0,r} + (1 - \alpha) u^{1,r} + \\
& p^r \cdot (\alpha d^{0,s} + (1 - \alpha) d^{1,s} - (\alpha d^{0,r} + (1 - \alpha) d^{1,r})) + \alpha \eps^0 + (1 - \alpha) \eps^1.
\end{align*} 
Since by Lemma~\ref{lem:epsqrat} the inequalities displayed above are the only ones we need to check, we obtain that
\[
\alpha \tilde{d}^0 + (1 - \alpha) \tilde{d}^1 \in X(\tilde{p}, \alpha D^0 + (1 - \alpha) D^1, \alpha \eps^0 + ( 1- \alpha) \eps^1).
\]
Finally, convexity of $A$ follows from similar averaging of the inequalities characterizing $\eps$-rationalizability in Lemma~\ref{lem:epsqrat}(iii).
\end{proof}

\begin{proof}[Proof of Proposition~\ref{prop:convexcount}]
Recall we now include quantities as arguments of the bounds. In more detail, write
\[
\overline{x}_k \left(d,\tilde{p},\eps \right) = \max_{\tilde{x} \in X\left(\tilde{p}, D, \eps \right)} \tilde{x}_k,
\]
where $D = \left\{ \left(d^t, p^t \right) \right\}_{t = 1}^T$ and we work in the extended reals so that $\overline{x}_k(d,\tilde{p},\eps)$ may be $\infty$. Addition is defined as $\infty + a = \infty$ provided $a$ is not $-\infty$.

Let $\tilde{d}^0, \tilde{d}^1 \in \mathbb{R}^{K \times T}$ be arbitrary quantities datasets. Since $\overline{x}_k$ is a maximum, we obtain
\[
\overline{x}_k \left(\alpha \tilde{d}^0 + (1 - \alpha) \tilde{d}^1, \tilde{p}, \alpha \eps^0 + (1 - \alpha) \eps^1 \right) \geq \alpha \overline{x}_k \left(\tilde{d}^0, \tilde{p}, \eps^0 \right) + (1 - \alpha) \overline{x}_k \left(\tilde{d}^1, \tilde{p}, \eps^1 \right),
\]
because from Lemma~\ref{lem:convexfeas} the weighted average of the values is feasible. This establishes concavity of $\overline{x}_k$ in its non-price arguments. Since $\underline{x}_k$ is a minimum we obtain
\[
\underline{x}_k \left(\alpha \tilde{d}^0 + (1 - \alpha) \tilde{d}^1, \tilde{p}, \alpha \eps^0 + (1 - \alpha) \eps^1 \right) \leq \alpha \underline{x}_k \left(\tilde{d}^0, \tilde{p}, \eps^0 \right) + (1 - \alpha) \underline{x}_k \left(\tilde{d}^1, \tilde{p}, \eps^1 \right),
\]
and so the lower bound is convex in its non-price arguments.

To establish continuity of $\overline{x}_k(\cdot, \tilde{p}, \cdot)$ as stated in the proposition, recall the linear programming formulation in \ref{a:qcounterfactuals}. We see that quantities $d$ and approximation error $\eps$ enter additively (relative to the choice variables) in the inequalities describing feasibility region. That is, they are part of the ``b'' in the canonical linear programming formulation from Appendix~\ref{supp:lp}.  Continuity then follows from Lemma~\ref{lem:lpcont}.
\end{proof}

\begin{proof}[Proof of Proposition~\ref{prop:convexwelfare}]
The proof is analogous to the proof of Proposition~\ref{prop:convexcount} and so we only outline it. Recall that $\overline{V}$ and $\underline{V}$ are described by a linear program in Proposition~\ref{prop:welfarecomp}. By inspecting the feasibility region of this program we see convexity holds similar to Lemma~\ref{lem:convexfeas}. From this, we conclude that $\overline{V}$ satisfies the concavity property in the proposition because it is a maximum, and $\overline{V}$ satisfies the convexity property in the proposition because it is a minimum.

Continuity of $\overline{V}$ and $\underline{V}$ in $(d,\eps)$ over the region where these bounds are finite follows from the linear programming formulation in Proposition~\ref{prop:lpwelfare} and Lemma~\ref{lem:lpcont}.
\end{proof}

\begin{proof}[Proof of Proposition~\ref{prop:meascont}]
Convexity is established in Proposition 3 in \cite{allen2020satisficing}. Continuity follows from the linear programming characterization of $\eps^*$ in Proposition 2 in \cite{allen2020satisficing}. Indeed, $\eps^*$ can be written as a maximum of finitely many functions that are affine in quantities $d$.
\end{proof}

\begin{proof}[Proof of Corollary~\ref{cor:counterfactualconsistency}]
By assumption $\overline{x}_k(d, \tilde{p}, \eps^*(d))$ is finite. By construction, $\left(\hat{d}^n,\eps \left(\hat{d}^n \right) \right) \in A$ for each $n$. Then from Propositions~\ref{prop:convexcount},~\ref{prop:meascont}, and the continuous mapping theorem, $\overline{x}\left(\hat{d}^n, \tilde{p}, \eps(\hat{d}^n)\right) \xrightarrow{p} \overline{x}_k(d, \tilde{p}, \eps^*(d))$. The arguments for $\overline{x}, \underline{V}$, and $\overline{V}$ are analogous.

\end{proof}

\begin{proof}[Proof of Proposition~\ref{prop:ushape}]
First note that the set $A(\tilde{x})$ is convex for any $\tilde{x} \in \mathbb{R}^K_+$ because it is a projection of the set $A$, which is convex by Lemma~\ref{lem:convexfeas}. The proof is analogous to the proofs of Propositions~\ref{prop:convexcount} and~\ref{prop:convexwelfare}. The feasibility region of the program describing $\overline{u}$ and $\underline{u}$ is given in Proposition~\ref{prop:lputility}. This feasibility regions of $\overline{u}$ and $\underline{u}$ are convex in quantities and degree of approximation error over the sets $A(\tilde{x}^0)$ and $A(\tilde{x}^1)$ respectively, similar to Lemma~\ref{lem:convexfeas}. Since $\overline{u}$ is a supremum it is concave, and since $\underline{u}$ is an infimum it is convex.

Continuity of $\overline{u}$ and $\underline{u}$ in $(d,\eps)$ over the region stated in the proposition follows from the characterization of the bounds in Proposition~\ref{prop:lputility}(ii) and (iii).  Indeed, $\overline{u}$ is the maximum of finitely many functions that are each affine in $(d,\eps)$, and $\underline{u}$ is the minimum.
\end{proof}

\section{Alternative Approaches} \label{sec:alternative}
Assumption~\ref{assm:maint} is the key conceptual assumption for this paper, which posits that approximation error is the same in new settings as the data we have seen. We have operationalized this for counterfactual and welfare analysis with a number controlling approximation error as in \cite{allen2020satisficing}. We now describe other potential ways to conduct counterfactual or welfare analysis. We also elaborate on the measurement and prediction wedges.

This paper focuses on approximation error being controlled by a single scalar. An alternative approach is to consider a multidimensional notion along the lines of \cite{afriat1972efficiency}, \cite{varian1990goodness}, \cite{varian1991goodness}, \cite{halevy2018parametric}, and \cite{masten2018salvaging}. We pursue this by allowing each observation to have its own value $\eps^t_V$ of approximation error relative to exact optimization.
\begin{defn} A dataset $\{ (x^t,p^t) \}_{t=1}^T$ is $\varepsilon_{\text{V}}$-rationalized by quasilinear utility for $\eps_{\text{V}}=(\eps_{\text{V}}^1,\ldots,\eps_{\text{V}}^T) \in \mathbb{R}_+^T$ if there exists a utility function $u: \mathbb{R}_+^K \rightarrow \mathbb{R}$ such that for all $t \in \{1,\ldots,T\}$ and for all $x \in \mathbb{R}_+^K$, the following inequality holds:
\[ u(x^t)-p^t\cdot x^t + \ge  u(x) - p^t \cdot x - \varepsilon_{\text{V}}^t.\]
We also refer to the above by saying a dataset is $\eps_{\text{V}}$-quasilinear rationalized.
\end{defn}

We can apply this concept to counterfactual analysis, as in the main text, by considering datasets in which the last observation is the hypothetical. That is, for a dataset $\{ (x^t,p^t) \}_{t=1}^{T+1}$, interpret the first $T$ observations as data we have seen and the last $T+1$ observation as be the hypothetical. In this case, the \textit{measurement wedges} are controlled by the collection $(\eps^1_V, \ldots, \eps^T_V)$ of values for the observed data, while the \textit{prediction wedge} is the scalar $\eps^{T+1}_V$. In principle we can consider counterfactuals involving several observations such as $T+1$ and $T+2$. We focus on the case of a single counterfactual for brevity.

In the main text, we set a single number controlling the measurement wedge and the prediction wedge. The full strength of Assumption~\ref{assm:prime} also imposes that these are equal to the minimal approximation error needed to explain the data we have seen. In this appendix we drop the assumptions that these wedges are the same, which shows how to generalize our framework when Assumption~\ref{assm:prime} is relaxed.

\subsection{Counterfactuals with Multidimensional Approximation Error}

Because $\eps_V$-rationalization is multidimensional, in general there is no single ``smallest'' vector $\eps_V$ such that the dataset is $\eps_V$-rationalized by quasilinear utility. Nonetheless, we can define a set of such rationalizing vectors via
\[
E\left(\tilde{D} \right) = \left\{ \eps_V \in \mathbb{R}^T_{+} \mid \tilde{D} \text{ is } \eps_V\text{-quasilinear rationalized} \right\},
\]
where we let $\tilde{D} = ((x^1, p^1),\ldots,(x^T,p^T))$ be the observed dataset written as an ordered tuple. Note that we switch from an unordered dataset to an ordered tuple. The reason we care about the order of observations now is that the $t$-th dimension of $E\left(\tilde{D} \right)$ corresponds to approximation error associated with the $t$-th observation. We can conduct counterfactual analysis as before by considering the set of quantity-price tuples that do not make approximation error worse. 

Formalizing worse here leads to some ambiguity since we do not have a total order on vectors. We consider two possibilities. To formalize these, let $\pi_T : \mathbb{R}^{T + 1} \rightarrow \mathbb{R}^T$ denote the projection onto the first $T$ components. We can then define the sets of lower and upper  approximate counterfactuals by
\begin{align*}
\underline{AC}\left(\tilde{D} \right) &= \left\{ (\tilde{x},\tilde{p}) \in \mathbb{R}_{+}^K \times \mathbb{R}_{++}^K \mid  \pi_T \left(E \left(\tilde{D}  \times (\tilde{x},\tilde{p})\right)\right) = E \left(\tilde{D} \right) \right\} \\
\overline{AC}\left(\tilde{D} \right) &= \left\{ (\tilde{x},\tilde{p}) \in \mathbb{R}_{+}^K  \times \mathbb{R}_{++}^K \mid \pi_T \left(E \left(\tilde{D}  \times (\tilde{x},\tilde{p}) \right) \right) \cap  E \left(\tilde{D} \right) \neq \emptyset \right\}.
\end{align*}
With minor abuse of notation we define $E(\cdot)$ in the obvious way for datasets of different dimensions. Clearly, $\underline{AC}\left(\tilde{D} \right) \subseteq \overline{AC}\left(\tilde{D} \right)$. The smaller set formalizes that counterfactuals do not change the potential $\eps_V$ vectors that rationalize the data we see when we add an existing observation. This smaller set is conceptually closer to the original adapative counterfactual set $AC(\cdot)$. The larger set formalizes that there is \textit{some} $\eps_V$ vector that $\eps_V$-rationalizes both the original dataset and the counterfactual-augmented dataset $\tilde{D} \times (\tilde{x},\tilde{p})$.

While the sets $\underline{AC}$ and $\overline{AC}$ may appear to be intuitive alternatives to the adaptive approach presented in the main text, unfortunately these sets are trivial. To see this, note that the set $\underline{AC}\left(\tilde{D} \right)$ allows the prediction wedge $\eps_V^{T+1}$ for the counterfactual value to be unbounded. This leads to trivial restrictions for both $\underline{AC}\left(\tilde{D} \right)$ and the larger set $\overline{AC}\left(\tilde{D} \right)$. However, nontriviality can be restored if we modify the sets by placing an \textit{a priori} bound on approximation error at the new observation. To formalize this, let $\overline{\pi} : \mathbb{R}^{T + 1} \rightarrow \mathbb{R}$ denote the projection of the last component. Then we can modify the smaller set via
\begin{align*}
\underline{AC}'\left(\tilde{D},\eps^{T+1}_V \right) = \Big\{ (\tilde{x},\tilde{p}) \in & \mathbb{R}_{+}^K \times \mathbb{R}_{++}^K \mid  \pi_{T} \left(E \left(\tilde{D}  \times (\tilde{x},\tilde{p})\right)\right) = E \left(\tilde{D} \right), \\
& \overline{\pi}\left(E \left(\tilde{D}  \times (\tilde{x},\tilde{p})\right)\right) \leq \eps^{T+1}_V \Big\}.
\end{align*}
Thus, the prediction wedge is restricted by the number $\eps^{T+1}_V$. This set is not data-adaptive in the sense that $\eps^{T+1}_V$ needs to be chosen by the researcher. However, one can make this data adaptive by using information obtained from other measures of approximation error discussed below. 

\subsection{Other Measures of Approximation Error}
It is natural to wonder for the multidimensional vector of approximation errors, $\eps_V$,  whether other intuitive one-dimensional summaries can be used for counterfactual analysis. In fact, there can be many ways to do this depending on how one aggregates the approximation error.  We consider general aggregators of the elements of multidimensional approximation error $\eps_V$ that turn it into a one-dimensional measure of approximation error. Formally, an aggregator can be written $e^T : \mathbb{R}^T_+ \rightarrow [0,\infty)$. Higher values of the aggregator can be interpreted as more approximation error. \cite{varian1990goodness} and \cite{halevy2018parametric} consider a related notion in the standard consumer problem for general utility maximization.\footnote{Many of the convenient shape restrictions we obtain in this paper do not hold for the case of general utility maximization since the constraint set of consistent utility indices is non-convex.}

Given a dataset $\tilde{D} = ((x^1, p^1),\ldots,(x^T,p^T))$ and an aggregator $e^T$, we can define a measure of approximation error via
\[
e^{T*}\left(\tilde{D} \right) = \inf_{\eps_V \in E\left(\tilde{D} \right) } e^T(\eps_V).
\]
One such aggregator is the max aggregator of $e_M^T(\eps_V) = \max_{t \in T} \eps^t_V$. The measure of approximation error for the max aggregator agrees with the one presented in the main text, i.e. $e_M^{T*} = \eps^*$. In general, an aggregator can depend on the sample size. For example, consider the average approximation error aggregator $e_A^T(\eps_V)=\frac{1}{T}\sum_{t =1}^T \eps_V^t$. Our leading measure, $\eps^*$, does not depend on $T$.   

For an arbitrary aggregator, similar to how $AC$ was constructed with the measure $\eps^*$, we can define a set of counterfactuals such that approximation error does not get worse:
\begin{equation} \label{eq:aggregatorcount}
\left\{ (\tilde{x},\tilde{p}) \in \mathbb{R}_{+}^K \times \mathbb{R}_{++}^K \mid e^{(T+1)*} \left(\tilde{D} \times  (\tilde{x},\tilde{p}) \right) \leq e^{T*}\left(\tilde{D} \right) \right\}.
\end{equation}
This construction does not separately control the prediction wedge and approximation wedge as in \underline{AC}' described at the end of the previous subsection. Instead, it lumps together both prediction and approximation wedges via the aggregators $e^{(T+1)*}$ and $e^{T*}$.

To understand properties of this set, consider an aggregator that sums up the observation-specific bounds on approximate optimization, $e_S^T(\eps_V) = \sum_{t = 1}^T \eps^t_V$.
With this choice of aggregator, each conjectured observation in (\ref{eq:aggregatorcount}) must be perfectly consistent with the model, i.e. $\eps^{T+1}_V = 0$, for approximation error to not be made worse. In other words, for each element of (\ref{eq:aggregatorcount}), there must exist some utility function that approximately explains the existing dataset $\tilde{D}$, but \textit{exactly} explains the counterfactual. Thus, there is no prediction wedge. This property may be desirable when one thinks the observed dataset $\tilde{D}$ comes from a ``true'' dataset that is generated by the quasilinear model but has been measured with error. Using the previous terminology, in this case we may wish to conduct counterfactual analysis without a prediction wedge. If instead we think approximation error propogates to new settings, then we may wish to allow the prediction wedge. This is one motivation for $\eps^*$ and the adaptive set $AC$.

One potential way to address this limitation of the sum-type aggregator $e^T_S$ is to adjust it by dividing by the sample size to obtain the average approximation error aggregator so that $\tilde{e}^{T*}_A = \frac{1}{T} e^{T*}_S$. This division allows one to construct a set analogously to (\ref{eq:aggregatorcount}) that allows a prediction wedge when generating counterfactual information.

\section{Supplemental Appendix} \label{supp}

This appendix contains additional results needed for proofs of the main results. Section~\ref{supp:misc} contains miscellaneous lemmas, Section~\ref{supp:lp} presents lemmas specifically for linear programming results, and Section~\ref{supp:duality} presents duality results used in proofs for approximate indirect utility in Section~\ref{sec:indirectbounds}.

\subsection{Miscellaneous Lemmas} \label{supp:misc}

\begin{lemma}[\cite{allen2020satisficing}]\label{lem:epsqrat}
For any dataset $\{(x^t, p^t)\}_{t=1}^T$ and $\varepsilon \ge 0$, the following are equivalent:
	\begin{enumerate}[(i)]
		\item $\{(x^t, p^t )\}_{t=1}^T$ is $\varepsilon$-rationalized by quasilinear utility.
		\item There exist numbers $\{ u^t \}_{t=1}^T$ that satisfy the following inequalities for all $r,s \in \{1,\ldots,T\}$:
		\[  u^s \le u^r + p^r \cdot (x^s-x^r) + \varepsilon.\]
		\item For all finite sequences $\{t_m\}_{m=1}^M$ with $t_m \in \{1,\ldots,T\}$ and $M\ge 2$, the inequality
        \[
        \frac{1}{M} \sum_{m = 1}^M p^{t_m} \cdot (x^{t_m} - x^{t_{m+1}}) \leq \eps
        \]
        holds, where $(x^{t_{M+1}},p^{t_{M+1}}) = (x^{t_1}, p^{t_1})$.
		\end{enumerate}
\end{lemma}

We require a lemma that will be used in the proof of Proposition~\ref{prop:nonempty} to ensure a maximizer exists. In contrast with models with compact budget constraints, continuity of the utility function $u$ is \textit{not} enough to ensure a maximizer exists, which is why we require the following lemma. To state the lemma, recall that for a utility function $u$, the indirect utility is defined as
\[
V_u(p) = \sup_{x \in \mathbb{R}^K_+} u(x) - p \cdot x.
\]
\begin{lemma} \label{lem:existence}
Suppose $u : \mathbb{R}^K_{+} \rightarrow \mathbb{R}$ is concave, monotonically increasing, and continuous. Moreover, suppose $V_u(p)$ is finite over some open set $O \subseteq \mathbb{R}^K$. It follows that for any price $p \in \ri( \Co(O) )$,\footnote{For a set $S \subseteq \mathbb{R}^K$, $\ri(S)$ gives the relative interior of the set $S$ as defined in \cite{rockafellar2015convex}.}
\[
u(x) - p \cdot x
\]
admits a maximizer for $x\in \mathbb{R}_+$.
\end{lemma}
\begin{proof}[Proof of Lemma~\ref{lem:existence}]
Since $V_u(p)$ is convex and finite over $O$, then $V_u(p)$ is finite on the $\ri( \Co(O) )$. Thus, the subdifferential
\[
\partial V_{u}(p) = \left\{ x \mid V_{u}(\tilde{p}) \geq V_{u} (p) + x \cdot (\tilde{p} - p) \qquad \forall \tilde{p} \in \mathbb{R}^K \right\} 
\]
is nonempty for any $p \in \ri(\Co(O))$ by \cite{rockafellar2015convex}, Theorem 23.4. Extend $u$ to all of $\mathbb{R}^K$ by setting $u(x) = -\infty$ for any $x \in \mathbb{R}^K \setminus \mathbb{R}_+^K$. Recall that the original $u$ defined on $\mathbb{R}^K_{+}$ is continuous, so that it is upper semicontinuous and $\{ x \mid u(x) \ge a \}$ is closed for any $a \in \mathbb{R}$ by Theorem 7.1 in \cite{rockafellar2015convex}. Note that the extension is also upper semicontinuous because it does not change the topological properties of the upper contour sets for all $a\in \mathbb{R}$. Since $u$ is upper semicontinuous and concave, we conclude from \cite{rockafellar2015convex} Theorem 23.5 parts $(b)$ and $(a^*)$ that for any $p \in \ri(\Co(O))$ there is some $x^* \in \mathbb{R}^K$ such that
\[
u(x) - p \cdot x
\]
is maximized over $x \in \mathbb{R}$ at $x^*$ since $\partial V_{u}(p)$ is nonempty. Since $u$ is $-\infty$ outside of $\mathbb{R}^K_+$, we conclude $x^* \in \mathbb{R}^K_+$. This establishes existence of an exact maximizer for the utility function $u$ over the region $p \in \ri(\Co(O))$, which completes the proof.
\end{proof}

\subsection{Linear Programming Lemmas} \label{supp:lp}

We require some existing results from the theory of linear programming. In canonical form, a linear program is written as
\begin{align*}
\max_{x \in \mathbb{R}^{J_1}} c\cdot x & \\
\text{s.t.} \qquad Ax & \leq b \\
x & \geq 0.
\end{align*}
Here, $c,x \in \mathbb{R}^{J_1}$, and $b \in \mathbb{R}^{J_2}$ are vectors and $A \in \mathbb{R}^{J_{2}\times J_{1}}$ is a matrix.

This is written as a maximum rather than a supremum because provided the supremum is finite, the maximum is attained as we formalize now.
\begin{lemma} \label{lem:lpexist}
If the value function of a linear program is finite, then the maximum is attained.
\end{lemma}
\begin{proof}
See e.g. \cite{bertsekas2009convex}, Proposition 1.4.12.
\end{proof}

Fixing all other variables, let $B \subseteq \mathbb{R}^{J_2}$ be the set of $b$ where the linear program is bounded. Write the value function as a function of $b$ so that
\begin{align*}
G(b) = \sup_{x} c \cdot x & \\
\text{s.t.} \qquad Ax & \leq b \\
x & \geq 0.
\end{align*}

\begin{lemma} \label{lem:lpcont}
Let $b_m \rightarrow b^*$ where $b^* \in B$ and for each $b_m$, the set
\[
\{ x \mid Ax \leq b_m , x \geq 0 \}
\]
is nonempty. It follows that $G(b_m) \rightarrow G(b^*)$.
\end{lemma}
\begin{proof}
Let $\eta > 0$ and define
\[
\overline{G}(b) = \min \{ G(b), G(b^*) + \eta \}.
\]
The value $\overline{G}(b)$ is the value function of a linear program defined by $G(b)$ that appends the inequality constraint $c\cdot x \leq G(b^*) + \eta$. Because $G(b^*) + \eta$ is finite, $\overline{G}(b)$ is finite for any feasible $b$. \cite{bohm1975continuity}, Theorem 1 states that $\overline{G}(b)$ is continuous over the set of $b$ such that the feasibility region is nonempty. Since $b_m \rightarrow b^*$ and each $b_m$ and $b^* \in B$ are feasible, we then obtain $\overline{G}(b_m) \rightarrow \overline{G}(b^*) = G(b^*)$.
\end{proof}

\subsection{Duality} \label{supp:duality}

The focus of the paper is on counterfactuals with approximate utility maximization. In other words, we consider utility functions $u$ such that the inequality
\[
u(x^t) - p^t \cdot x^t \geq u(x) - p^t \cdot x - \eps
\]
holds for every $t \in \{1, \ldots, T \}$ and $x \in \mathbb{R}^K_{+}$. In this supplement, we consider a dual approach involving functions $V$ such that
\[
V(p) \geq V(p^t) - x^t \cdot (p - p^t) - \eps
\]
holds for every $t \in \{1, \ldots, T \}$ and $p \in \mathbb{R}^K_{++}$. We also mention some results from convex analysis. Results from this section are used to prove several results in the main text.

Our first result formalizes that finding a utility function $u$ that satisfies the first set of inequalities is equivalent to finding a $V$ function that satisfies the second set of inequalities.

To state the result, recall the indirect utility function of $u$ is given by $V_u : \mathbb{R}^K_{+} \rightarrow \mathbb{R}^K \cup \{-\infty, \infty\}$
\[
V_u(p) = \sup_{x \in \mathbb{R}^K_+} u(x) - p \cdot x.
\]
We make use of a ``dual'' utility function $u_V : \mathbb{R}^K_{+} \rightarrow \mathbb{R}^K \cup \{-\infty, \infty\}$ constructed via
\[
u_V(x) = \inf_{p \in \mathbb{R}^K_{+}} V(x) + p \cdot x.
\]
These operations can be defined for any extended real-valued functions $u : \mathbb{R}^K \rightarrow \mathbb{R}^K \cup \{ -\infty , \infty \}$ and $V : \mathbb{R}^K \rightarrow \mathbb{R}^K \cup \{ -\infty , \infty \}$.

\begin{prop} \label{slem:duality}
Let $\eps \geq 0$ and let $\{ (x^t,p^t) \}_{t = 1}^T$ be an arbitrary dataset.
\begin{enumerate}[i.]
    \item Suppose $u : \mathbb{R}^K_{+} \rightarrow \mathbb{R}$ satisfies
    \[
    u(x^t) - p^t \cdot x^t \geq u(x) - p^t \cdot x - \eps
    \]
    for every $t \in \{1, \ldots, T \}$ and every $x \in \mathbb{R}^K_{+}$. It follows that $V_u$ satisfies 
    \[
    V_u(p) \geq V_u(p^t) - x^t \cdot (p - p^t) - \eps
    \]
    for every $t \in \{1, \ldots, T \}$ and every $ p \in \mathbb{R}^K_{+}$.
    \item Suppose $V : \mathbb{R}^K_{+} \rightarrow \mathbb{R}$ satisfies
    \[
    V(p) \geq V(p^t) - x^t \cdot (p - p^t) - \eps
    \]
    for every $t \in \{1, \ldots, T \}$ and every $p \in \mathbb{R}^K_{+}$. It follows that $u_V$ satisfies 
    \[
    u_V(x^t) - p^t \cdot x^t \geq u_V(x) - p^t \cdot x - \eps
    \]
    for every $t \in \{1, \ldots, T \}$ and every $x \in \mathbb{R}^K_{+}$.
\end{enumerate}

\end{prop}

\begin{proof}
First we show (i). For arbitrary $t \in \{ 1, \ldots, T \}$, write
\begin{equation} \label{eq:indfinite}
V_u(p^t) = u(x^t) - p^t \cdot x^t + \delta^t
\end{equation}
where $\delta^t=V_u(p^t)-u(x^t)+p^t\cdot x_t \ge 0$ and $\delta^t \le \eps$ since the observed quantities are only approximately optimal. 
For arbitrary $p \in \mathbb{R}^K_{+}$ we have
\[
V_u(p) \geq u(x^t) - p \cdot x^t.
\]
Differencing yields
\[
V_u(p) - V_u(p^t) \geq - x^t \cdot (p - p^t) - \delta^t.
\]
The term $V_u(p)$ may equal $\infty$, in which case we define $\infty - a = \infty$ for any finite $a$. Here, $V_u(p^t)$ is finite from (\ref{eq:indfinite}) and the fact that $u$ is finite. We know for all $t \in \{1,\ldots, T\}$ that $0 \leq \delta^t \leq \eps$ by assumption, so (i) is established.

Now we show (ii). As before, write
\[
u_V(x^t) = V(p^t) + p^t \cdot x^t - \delta^t
\]
where $\delta^t=u_V(x^t)-V(p^t)-p^t\cdot x^t \ge 0$ and $\delta \le \eps$ since the observed quantities and prices are only supposed to satisfy the inequality in (ii) for $V$.
For arbitrary $x \in \mathbb{R}^K_+$ we have
\[
u_V(x) \leq V (p^t) + p^t \cdot x,
\]
and so
\[
u_V(x^t) - p^t \cdot x^t \geq u_V(x) - p^t \cdot x - \delta^t.
\]
As before, $u_V(x)$ can equal $-\infty$, but $u_V(x^t)$ is always finite. Recall, for all $t \in \{1,\ldots,T\}$ that $0 \leq \delta^t \leq \eps$ by assumption, and so (ii) is established.
\end{proof}

The mappings $u \rightarrow V_u$ and $V \rightarrow u_V$ are closely related to convex conjugates, and we can adapt existing results from convex analysis. Recall that for a function $f : \mathbb{R}^K \rightarrow \mathbb{R} \cup \{ - \infty, \infty \}$, the \textit{convex conjugate} is given by
\[
f^*(p) = \sup_{x \in \mathbb{R}^K} p \cdot x - f(x).
\]
The \textit{monotone conjugate} is given by
\[
f^+(p) = \sup_{x \in \mathbb{R}^K_{+}} p \cdot x - f(x).
\]
Let the function $\tilde{f}$ equal $f$ over $x \in \mathbb{R}^K_{+}$, and $\infty$ otherwise. It follows that $\tilde{f}^*(p) = f^{+}(p)$.

We now formalize the relationships between $u_V$ and $V_u$ and monotone conjugates. Following this, we present some immediate consequences. 
\begin{lemma} \label{slem:monconj}
\[
u_V(x) = - V^{+}(-x)
\]
and
\[
V_u(p) = (-u)^+(-p).
\]

\end{lemma}

\begin{proof}

\[
u_V(x) = \inf_{p \in \mathbb{R}^K_{+}} V(p) + p \cdot x
= -\sup_{p \in \mathbb{R}^K_{+}} - V(p) - p \cdot x
= - V^{+}(-x).
\]
and
\[
V_{u}(p) = \sup_{x \in \mathbb{R}^K_{+}} u(x) - p \cdot x
     = \sup_{x \in \mathbb{R}^K_{+}} - (- u(x)) + (-p) \cdot x
     = (-u)^{+}(-p).
\]
\end{proof}

\begin{lemma} \label{slem:shapeconjugate}
The function $u_V$ is concave, weakly increasing, and upper semicontinuous. The function $V_u$ is convex, weakly decreasing, and lower semicontinuous.
\end{lemma}

\begin{proof}
To see that $V_u$ is weakly decreasing, consider $p^a, p^b$ with $p^a \geq p^b$. For $x \in \mathbb{R}^K_{+}$ we have $p^a \cdot x \geq p^b \cdot x $ and so
\[
V_u(p^a) = \sup_{x \in \mathbb{R}^K_+} u(x) - p^a \cdot x 
\leq \sup_{x \in \mathbb{R}^K_+} u(x) - p^b \cdot x 
 \leq V_u(p^b).
\]
Also, $V_u$ is convex and lower semicontinuous from \cite{bertsekas2009convex}, p. 83. The arguments for $u_V$ are analogous by applying Lemma~\ref{slem:monconj}.
\end{proof}

\begin{lemma} \label{slem:doubleconj}
\begin{enumerate}[i.]
\item Suppose $u : \mathbb{R}^K_+ \rightarrow \mathbb{R}^K \cup \{ - \infty, \infty \}$ is concave, weakly increasing, upper semicontinuous, and finite at $0$. Then $u_{V_u} = u$.
\item Suppose $V : \mathbb{R}^K_+ \rightarrow \mathbb{R}^K \cup \{ - \infty, \infty \}$ is concave, weakly decreasing, lower semicontinuous, and finite at $0$. Then $V_{u_V} = V$.
\end{enumerate}
\end{lemma}

\begin{proof}
Write $f^{++}$ as the monotone conjugate of $f^+$. \cite{rockafellar2015convex}, Theorem 12.4 states $V^{++} = V$ and $(-u)^{++} = -u$. From Lemma~\ref{slem:monconj} we conclude
\[
-u_{V_u}(x) = (V_u)^+(-x) = (-u)^{++}(x) = -u(x) 
\]
and
\[
V_{u_V}(p) = (-u_V)^{+} (-p) = V^{++}(p) = V(p).
\]
\end{proof}

\end{appendices}

\bibliographystyle{plainnat}
\bibliography{ref}

\end{document}